\documentclass[12pt]{article}
\usepackage{epsfig,psfrag,amsmath,amssymb,latexsym}
\usepackage{amscd }

\usepackage{color}
\usepackage{amsfonts}
\usepackage{graphicx}
\usepackage{wasysym}
\newtheorem{fact}{Fact}[section]

\newtheorem{theorem}[fact]{Theorem}

\newtheorem{corollary}[fact]{Corollary}

\newtheorem{definition}[fact]{Definition}

\newtheorem{lemma}[fact]{Lemma}

\newtheorem{remark}[fact]{Remark}

\newenvironment{proof}[1][Proof]{\textbf{#1.} }{\ \rule{0.5em}{0.5em}}

\numberwithin{equation}{section}
\newcommand{\notree}{{\rm no\;tree}}
\newcommand{\bb}{{\rm bb}}
\newcommand{\tback}{T^{\bb}}

\newcommand{\id}{\operatorname{id}}

\newcommand{\0}{{\mathbf{0}}}

\newcommand{\pin}{{\operatorname{pin}}}
\newcommand{\grad}{{\operatorname{grad}}}
\newcommand{\treevar}{{\tt treevar}}
\newcommand{\ul}{{\underline{L}}}
\newcommand{\ol}{{\overline{L}}}
\newcommand{\uroot}{r}
\newcommand{\oroot}{{\overline{r}}}
\newcommand{\rechts}{{\rm right}}
\newcommand{\links}{{\rm left}}
\newcommand{\mitte}{{\rm middle}}
\newcommand{\vertical}{{\rm vert}}
\newcommand{\hor}{{\rm hor}}

\newcommand{\re}{\operatorname{Re}}
\newcommand{\im}{\operatorname{Im}}
\newcommand{\sk}[1]{\left\langle #1 \right\rangle}

\newcommand{\refl}{\leftrightarrow}
\newcommand{\fol}{\leadsto}

\newcommand{\word}{\mathbb{T}}
\newcommand{\glue}{\operatorname{glue}}
\newcommand{\K}{\mathcal{K}}
\newcommand{\one}{\mathbf{1}}

\newcommand{\bomega}{{\boldsymbol\omega}}
\newcommand{\Bomega}{{\boldsymbol\Omega}}
\newcommand{\words}{{\tt words}}
\newcommand{\indices}{_{L}^{\0\ell}}
\newcommand{\Egamma}{E\indices(\alpha)}
\newcommand{\Sgamma}[1]{S\indices(#1)}
\newcommand{\Pint}{\mathbb{P}\indices}
\newcommand{\Zint}{Z\indices}
\newcommand{\Hint}{H\indices}

\newcommand{\betamax}{\beta_{\max}}
\newcommand{\takegrad}{\boldsymbol\Phi}
\newcommand{\aux}{G^{\rm aux}}

\def\R{{\mathbb R}}  
\def\N{{\mathbb N}}  
\def\p{{\mathbb P}}  
\def\Z{{\mathbb Z}}  
\def\C{{\mathbb C}}  

\def\T{{\mathcal{T}}}
\def\G{{\mathcal{G}}}  

\newcommand{\ceins}{c_{12}}
\newcommand{\czwei}{c_{13}}
\newcommand{\cdrei}{c_{11}}
\newcommand{\cvier}{c_{10}}
\newcommand{\cfuenf}{c_{7}}
\newcommand{\csechs}{c_{3}}
\newcommand{\csieben}{c_{4}}
\newcommand{\cacht}{c_{5}}
\newcommand{\cneun}{c_6}
\newcommand{\czehn}{c_{1}}
\newcommand{\celf}{{c_{2}}}
\newcommand{\czwoelf}{c_{8}}
\newcommand{\cdreizehn}{c_{9}}
\newcommand{\cvierzehn}{c_{14}}

\begin{document}
\thispagestyle{empty}

\begin{center}
{\LARGE  Localization for a nonlinear sigma model 
in a strip \\ 
related to  vertex reinforced
 jump processes}\\[3mm]
{\large Margherita Disertori\footnote{Laboratoire de 
Math\'ematiques Rapha\"el Salem, UMR CNRS 6085, 
Universit\'e de Rouen, 76801, France, 
e-mail: margherita.disertori@univ-rouen.fr}
\hspace{1cm} 
Franz Merkl\footnote{Mathematical Institute, University of Munich,
Theresienstr.\ 39,
D-80333 Munich,
Germany.
e-mail: merkl@mathematik.uni-muenchen.de
}
\hspace{1cm} 
Silke W.W.\ Rolles\footnote{Zentrum Mathematik, Bereich M5,
Technische Universit{\"{a}}t M{\"{u}}nchen,
D-85747 Garching bei M{\"{u}}nchen,
Germany.
e-mail: srolles@ma.tum.de}
\\[3mm]
{\small \today}}\\[3mm]
\end{center}

\begin{abstract}
We study a lattice sigma model which is expected to reflect 
Anderson localization and delocalization transition for real 
symmetric band matrices in 3D, but describes  the mixing measure
for a vertex reinforced jump process too. 
For this model we prove exponential localization at any temperature
in a strip, and more generally in any quasi-one dimensional graph, 
with pinning (mass) at only one site. 
The proof uses a Mermin-Wagner type argument and a transfer operator approach. 
\end{abstract}


\section{Introduction}


\paragraph{Nonlinear sigma models and random matrices.}
Nonlinear sigma models appear as effective models at low energy in
a large variety of physical problems where 
some kind of spontaneous symmetry breaking and phase transition is expected. 
They can be viewed, in analogy to statistical mechanics, as models of interacting spins, 
taking values in a nonlinear manifold.  In the context of
disordered conductors and quantum chaos,  the spectral and transport
properties of random Schr\"odinger operators and random band matrices
can be translated 
in the study of the correlation functions for a  statistical mechanics model
where the spin at each lattice site $j$ is replaced by a
 matrix $Q_{j}$, whose elements are both ordinary (bosonic) complex or real 
variables and anticommuting (fermionic)  Grassmann variables. 
This representation
was introduced and developed by  Efetov \cite{efetov-adv,efetov-book}, 
based on seminal work by Wegner \cite{wegner,sw}, 
using the supersymmetric approach. 
In the corresponding nonlinear sigma model, the matrix $Q_{j}$ satisfies $Q_{j}^{2}=\mathbf{1}$
and is restricted to take values on a supermanifold, whose symmetry  
properties depend on the symmetries of the initial random matrix ensemble and
the observable (correlation function) under study.
Efetov's supersymmetric nonlinear sigma model and its variants  
were intensively studied, especially in the physics literature, but they 
still defy a rigorous mathematical understanding. See 
\cite{spencer2012,mirlin2000,mirlin1999,fyodorov} for an 
introduction to these problems.

In this context Zirnbauer introduced  a supersymmetric sigma model 
\cite{zirnbauer-91,migdal-kad}, that is expected to reflect 
Anderson localization and delocalization transition for real 
symmetric band matrices in 3D. 
In this statistical mechanical model the field (or spin) at site $j$ 
is a vector $v_{j}= (x_{j},y_{j},z_{j}, \xi_j, \eta_j)$  where $x,y,z$
are real and $ \xi,\eta$ are Grassmann variables. We endow this vector space 
with a generalization of  the Lorentz metric 
$(v,v')= xx'+yy'-zz'+ \xi\eta'-\eta\xi'$. Imposing the constraint
$(v_{j},v_{j})=-1$, the field $v_{j}$ has four degrees of
 freedom (two bosonic and two fermionic),  and takes values in a target space denoted by $H^{2|2}$,
which is a supermanifold extension of the hyperbolic plane $H^{2}$. 
The effective action for this model is given by
\begin{equation}\label{free-en}
\mathcal{F}(v)= \sum_{i\sim j}\frac{ \beta_{ij}}{2} (v_{i}-v_{j},\, v_{i}-v_{j}) +  
\sum_{j} \varepsilon_{j} (z_{j}-1)
\end{equation}
where the first term is the kinetic energy, and $i\sim j$ denotes edges connecting nearest 
neighbors $i$ and $j$. See  \cite[Sect. 2.1]{disertori-spencer-zirnbauer2010} for  more details.
The parameter $\beta_{ij}=\beta_{ji}>0$  
may be seen as a local inverse temperature along the edge $i\sim j$,
using the language of statistical mechanics. The last term in the action
is needed to break the non-compact symmetry and make the corresponding 
integral finite, so $\varepsilon_{j}\geq 0$ can be seen as the analog of a 
magnetic field, or a mass term.
Note that if we add a new vertex 
$\rho$ to the lattice, we connect  it to all lattice points $j$ with 
$\varepsilon_{j}>0$ and we fix $v_\rho= (0,0,1,0,0)$ we have
\[
\sum_{j} \varepsilon_{j} (z_{j}-1)= 
\sum_{j} \frac{\varepsilon_{j} }{2} (v_{j}-v_\rho,v_{j}-v_\rho).
\]
Then the mass term  may be seen as a kinetic term too.
In the appropriate coordinate system (see \cite{disertori-spencer-zirnbauer2010}) 
the action becomes quadratic in the 
fermionic variables and these variables can be integrated out exactly. We 
 are left with two real variables $t_{j},s_{j}$ at each
lattice site and  a probability measure $d\mu (t,s)$. 
The resulting statistical mechanical model then has
 a probabilistic interpretation. In this paper, we will not use the 
supersymmetric formalism at all, and will work directly on the probability
measure  $d\mu (t,s)$, whose precise form is given in \eqref{eq:tmeasure} 
below\footnote{Actually here we work only with the formula for one pinning point.
For the more general formula see \cite{disertori-spencer2010}.}.

\paragraph{Connection with stochastic processes.}
Recently Sabot and Tarr\`es \cite{sabot-tarres2012} 
proved  a precise relation between $H^{2|2}$  and both   vertex reinforced jump process (VRJP)
and linearly edge reinforced random walk (LERRW) on the graph with the additional
vertex $\rho$. Both  are history dependent
stochastic processes, describing self-organization and learning behavior. 
VRJP was conceived by Werner and studied in  
\cite{davis-volkov1,davis-volkov2,collevecchio2,collevecchio1,basdevant}.
It is a continuous time  process $Y= (Y_{u})_{u\geq 0}$ where the particle jumps
from the lattice site $i$ to $j$ with  rate $\beta_{ij} (1+L_{j} (u))$, 
where $L_{j} (u)$ is the local time at $j$, 
that is the time the particle has  already spent on $j$ up to time $u$.
Here we take the convention $\beta_{i\rho}=\beta_{\rho i}=\varepsilon_{i}$.
In this context large/small $\beta$ corresponds to weak/strong reinforcement.
Indeed, assuming $\beta$ to be constant, we can rescale the time by 
$u'=\beta u$.
Then $L'_{j} (u')=\beta L_{j} (u)$, the jump rate becomes  $1+L'_{j} (u') /\beta$ and
the bigger $\beta$ is, the weaker the influence of the local time.
 Let $\tilde{Y}= (\tilde{Y}_{n})_{n\in \mathbb{N}_{0}}$ 
be the discrete time process associated to $Y$ by taking only the value of $Y_{u}$
immediately before the jump times, ignoring the waiting time between jumps.
Sabot-Tarr\`es  \cite{sabot-tarres2012} proved that on any finite graph
it can be represented as a random walk in a random environment, 
and  more precisely as a mixture of reversible Markov
chains  
\begin{equation}\label{mixing}
  \mathbb{P} (\tilde{Y} \in A)  = 
\int \mathbb{P}^{W (t,s)}( \tilde{Y}\in A) \  d\mu (t,s) 
\end{equation}
for any event $A$ on paths, where $ d\mu (t,s)$ is the measure for 
$H^{2|2}$ defined in \eqref{eq:tmeasure} below.  
Here  $\mathbb{P}^{W (t,s)} (\tilde{Y}\in \cdot)$ 
is the probability law associated to the 
Markovian random walk starting at the root $\rho$  
and jumping  from $i$ to $j$ with probability proportional to 
$W_{ij} (t,s)=W_{ji} (t,s) =\beta_{ij}e^{t_{i}+t_{j}}$ for any $i\sim j$, with the 
convention $t_\rho=0$.
The probability measure  $\mu(W\in \cdot )$ allows to pick randomly the environment 
where the particle moves.  It is called  the mixing measure  for the process.
From the stochastic process perspective, the most natural situation 
is to consider the case of one pinning point 
$\varepsilon_{j}=\varepsilon\delta_{jj_{0}}$, 
where $j_{0}$ is some fixed lattice site.

LERRW is a discrete time  process $X= (X_{n})_{n\in \mathbb{N}_{0}}$ 
where the particle jumps at time $n$ from the lattice site $i$ to $j$ with 
a probability depending on the number of times it has traversed the $i\sim j$ 
edge in the past.
This model is known to be a mixture of reversible Markov chains 
with explicitly known mixing measure 
\cite{Coppersmith-Diaconis-unpublished,Keane-Rolles2000}.
The relation of this model to $H^{2|2}$ was clarified by Sabot and Tarr\`es 
\cite[Thm.~1]{sabot-tarres2012}, who showed that LERRW is obtained 
from the discrete time VRJP as a mixture  
by taking the weights  $(\beta_{ij})$ in $W_{ij}(t,s)$ to be independent 
Gamma distributed random variables. 
Using this relation they proved localization of LERRW in $d\geq 1$ 
for strong reinforcement.

\paragraph{Results and conjectures.}
 Exponential localization for $H^{2|2}$ was established  in $d=1$
  \cite{zirnbauer-91}-\cite{disertori-spencer2010} for any value of $\beta$,  
and in $d\geq 1$  for small $\beta$  \cite{disertori-spencer2010}.
A quasi-diffusive phase was established 
in $d\geq 3$ for large $\beta$ 
thus proving the existence of a phase transition 
\cite{disertori-spencer-zirnbauer2010}. 
In $d=2$ localization is expected to hold for any value of $\beta$, 
with localization length of order $e^{\beta }$.
The proofs in \cite{disertori-spencer2010,disertori-spencer-zirnbauer2010}
are derived for constant parameters $\beta_{ij}=\beta $, but they can be easily 
generalized to the case of variable betas.
In the case of one pinning point $\varepsilon_{j}=\delta_{jj_{0}}$,
the results listed above imply  that
the corresponding VRJP starting at $j_{0}$ is recurrent in $d=1$ for any value of $\beta$,  
and in $d\geq 1$  for small $\beta$  \cite{sabot-tarres2012}.

In this paper we consider $H^{2|2}$ in the case of one pinning point  
on a generalized strip, consisting of  copies of an arbitrary finite connected graph. 
For this model we prove exponential localization 
for any periodic choice of $\beta_{ij}$, uniformly in the number of copies.  
This implies the corresponding discrete time process associated to VRJP is recurrent on 
the infinite strip
and exponentially localized in a finite region with high probability: 
$\mathbb{P} (|\tilde{Y}_{n}|>R)\leq e^{-cR}$ for some constant $c$, 
independent from $n$.
Similar statements could be made about the continuous time process too, although they are not
worked out here.

\paragraph{Idea of the proof.} Though one may expect this should be just a small
modification of the $1D$ proof, it turns out    the
argument used in \cite{disertori-spencer2010} breaks down as soon 
as we leave the perfect one dimensional chain, unless we take $\beta $ small. 
Here we use  a quite different approach, namely a deformation argument on 
probability measures (in the spirit of Mermin-Wagner).
The transfer operator method is used in a non standard way. 
Instead of estimating the top eigenvalue directly,  it  
 is used to bound some of the terms generated by  the
deformation uniformly in the length of the strip. For this purpose,
we need only a reflection symmetry and compactness. 
To set up a transfer operator we need to express our measure as a product of local
functionals. This is not trivial due to the presence of a highly non local determinant 
in the measure. One might write this determinant in terms of  a product of local functions of
Grassmann (anticommuting) variables, but would have to deal with a transfer matrix
involving both real and Grassmann variables. In contrast, here
we write the determinant as a sum over spanning trees,  using the matrix-tree theorem.
Then these trees can be described by a set of local variables and the Boltzmann weight
becomes a product of local non-negative functions.

The arguments are inspired by the methods used by two of us in previous work on the LERRW 
\cite{Merkl-Rolles2005c,merkl-rolles-chain}. 
In contrast to that work though, here our deformation depends on a cut-off function
selecting only small gradients. With this choice, we have to estimate bounded observables only.
An alternative method, not presented here, would be to remove the cut-off function, and deal 
with unbounded observables.

To simplify the proof as much as possible, we did not try to estimate 
the localization length as a function of $\beta $ or $W$, though such an estimate
is maybe achievable by a more detailed analysis.
Some variant of the arguments we use here may perhaps be applied to the case 
of uniform pinning $\varepsilon_{j}=\varepsilon$. On the other hand, the true $d=2$ case
with large beta is much harder and will probably need a different approach.

\vskip1.truecm

\noindent {\bf \textit{Acknowledgements.}} 
It is our pleasure to thank T. Spencer for many useful discussions and 
suggestions related to this paper.

\section{Model and main result}\label{sect2}
\subsection{The model}
Let $G_0=(V_0,E_0)$ be a finite undirected graph with vertex set $V_0$ and edge
set $E_0$. If there is an edge between $v$ and $v'$ in the graph $G_0$, 
we write $(v\sim v')\in E_0$. 
We consider the sigma model on the graph $\G$ obtained 
by putting infinitely many copies of $G_0$ in a row. 
Let $G_n=(V_n,E_n)$ be the copy of $G_0$ at level $n\in\Z$. 
More precisely, 
$\G=(V,E)$ has vertex set $V:=\Z\times V_0$ and  
edge set  $E=\bigcup_{n\in \mathbb{Z}} (E_n\cup E_{n+1/2})$, where 
\begin{align}
E_n:=\{ e_n:=((n,v)\sim (n,v')):\; 
e=(v\sim v')\in E_0\}
\end{align}
is the set of ``vertical'' edges in $G_n$, connecting the copies at level $n$ 
of the vertices $v$ and $v'$, for any edge $(v\sim v')\in E_{0}$, and 
\begin{align}
E_{n+1/2}:=\{v_{n+1/2}:=((n,v)\sim (n+1,v)):\; v\in V_0\}
\end{align}
is the set of ``horizontal'' edges connecting each vertex in $V_n$ with 
its copy in $V_{n+1}$. We say that an edge in $E_{n+1/2}$ is  at level
$n+1/2$. 
Note that, in the special case when $G_{0}$ is a finite segment of $\mathbb{Z}$, 
the graph we obtain is an infinite strip. In this case  the edges $e_{n}$
are vertical lines while the edges $v_{n+1/2}$ are horizontal  lines.
Hence, the names above. 

For $\ul,\ol\in\N$, we set $L:=(-\ul,\ol)$ 
and consider the finite piece $\G_L=(V_L,E_L)$ of $\G$ with vertex set 
$V_{L}:=\{-\ul,\ldots,\ol\}\times V_0$. 
Let $p$ be a fixed vertex in $V_0$. 
We abbreviate $\0:=(0,p)$. The site $\0$ in $G_0$ will be used 
as pinning point (hence the name we chose). 
To each site $j\in V_{L}$ we associate the 
real variables $t_{j}$ and $s_{j}$. 
We abbreviate $t=(t_i)_{i\in V_L}$ and $s=(s_i)_{i\in V_L}$ and  
let $\nabla t=(t_i-t_j)_{i,j\in V_L}$ denote the vector of gradients
of the $t$ variables. 
We introduce the probability measure\footnote{This measure is normalized to one by
supersymmetry, see   \cite[Sect. 4]{disertori-spencer-zirnbauer2010}. The factor $2\pi $
comes from integrating over the fermionic variables. Alternatively
one may notice that this is the mixing measure for a VRJP hence it is normalized to one.}
\begin{equation}
\label{eq:tmeasure}
d\mu^{\0 }_L  (t,s) = 
\prod_{j\in V_L} \frac{dt_{j}ds_{j}e^{-t_{j}}}{2\pi }
e^{-F_L(\nabla t)} e^{-\frac{1}{2}[s,A_L(t) s]} 
\det  [A_L(t)+ \widehat\varepsilon\, ]\, 
e^{- M (t_{\0},s_{\0})} ,
\end{equation}
where $dt_{j}$ and $ds_{j}$ denote the Lebesgue measure on $\R$,  
$A_L (t)=(A_L(t)_{ij})_{i,j\in V_L}$ 
is the positive definite matrix  defined by
\begin{equation}\label{eq:Mmatrix}
A_L(t)_{ij} = 
\left\{\begin{array}{ll}
-\beta_{ij}e^{t_{i}+t_{j}}   & \text{ if }i\sim j,\\
\displaystyle\sum_{k: k\sim j} \beta_{kj}  e^{t_{k}+t_{j}}
 & \text{ if }i=j, \\
0  & \text{ otherwise,}
\end{array}\right.
\end{equation}
and $\widehat\varepsilon$ is the diagonal matrix with entries 
\begin{align}
\widehat\varepsilon_{ij} = \delta_{i\0}\delta_{j\0}\varepsilon e^{t_{\0}} 
\quad\text{ for }i,j\in V_L.
\end{align}
The arguments in the exponent 
are defined by
\begin{align}
F_L (\nabla t)  &= \sum_{(i\sim j)\in E_L } \beta_{ij}   
(\cosh(t_i - t_j)-1), \\
[s,A_L(t) s] &=  
\sum_{(i\sim j)\in E_L } \beta_{ij} (s_i-s_j)^{2} 
e^{t_{i}+t_{j}},\\
M(t_{\0},s_{\0}) &=  \varepsilon \left[  \cosh t_\0 -1  + 
 \frac{s_{\0}^{2}}{2} e^{t_{\0}}\right], 
\end{align} 
where $\varepsilon, \left(\beta_{ij} \right)_{ (i\sim j)\in E_L}$
are positive fixed weights. In the remainder of this article, we 
consider only translation invariant weights:
\begin{align}\label{betadef0} 
\beta_{e_{n}}= \beta_{e_{0}} \quad \forall e\in E_0 \quad\text{ and }\quad
\beta_{v_{n+1/2}}=\beta_{v_{1/2}} \quad \forall v\in V_0,\ \forall n\in\Z. 
\end{align}
Therefore, we can recover $\beta_e$ for all $e\in E_L$ from 
\begin{align}
\vec\beta:=(\beta_e)_{e\in E_0\cup E_{1/2}} . 
\label{betadef}\end{align}

\subsection{The main result} 
With the above definitions we can now state the main result of the paper.

\begin{theorem}\label{th1} 
There exist constants $\czehn,\celf>0$ depending only on 
$G_0$ and $\vec\beta$ such that for all $L=(-\ul,\ol)$ and $l$ with
$-\ul\le l\le\ol$, one has 
\begin{align}
\label{claim-main-thm}
\mathbb{E}_{\mu^{\0 }_L}\left[ e^{\frac{t_{\ell}-t_{\0}}{2}} \right]\leq 
\czehn e^{-\celf l},
\end{align}
where $\ell:=(l,p)$ denotes the copy of the pinning point $p$ at level $l$.
The estimate holds uniformly in $L$.  
Moreover, there exists a probability measure
$\mu^{\0 }_\infty$ on $\R^{V}\times \R^{V}$ such that for any bounded observable $\mathcal{O}$ 
depending only on finitely many $t_{i},s_{i}$ we have
$\mathbb{E}_{\mu^{\0 }_L}[\mathcal{O}]\to \mathbb{E}_{\mu^{\0 }_\infty}[\mathcal{O}]$
as $L=(-\ul,\ol)\to (-\infty,+\infty)$.
\end{theorem}

Using this result we can derive several properties of the VRJP.
Let $\G_\rho$ denote the graph $\G$ with the additional vertex
$\rho$, that is connected only to $\0$. 

\begin{corollary}
\label{cor-vrjp}
The discrete time process associated to the vertex reinforced jump process on the infinite graph 
$\G_\rho$ is a mixture of positive recurrent irreducible reversible 
Markov chains for any translation invariant beta as in \eqref{betadef0} above.
The mixing measure for the random weights, indexed by edges $i\sim j$ in 
$\G_\rho$, is given by the joint distribution 
of $(W_{ij}(t,s)=\beta_{ij}e^{t_{i}+t_{j}})_{i\sim j}$ 
with respect to $\mu^{\0 }_\infty$; 
here we use the convention 
$W_{\rho \0} (t,s)= \beta_{\rho\0  }e^{t_{\0 }+t_{\rho }}= \varepsilon e^{t_{\0 }}$ 
with $t_{\rho }=0$ and $\beta_{\rho\0  }= \varepsilon$.
\end{corollary}
By standard arguments similar to the ones given in  \cite{merkl-rolles-asymptotic} for 
the linearly edge reinforced random walk case, the decay properties of $t_{i}$ as $i\to\infty $  
with respect to $\mu^{\0 }_\infty$  allow to derive several asymptotic properties of VRJP.

In the following, the level of any vertex $v= (m,x)\in V_{m}$, 
$x\in V_{0}$, is denoted by $|v|=m$. 

\begin{corollary}\label{lemma-exp-loc}
For the discrete-time process $(\tilde{Y}_n)_{n\in\N_0}$
associated to the VRJP on $\G_{\rho }$ there exist constants
$\csechs,\csieben>0$ depending only on 
$G_0$ and $\vec\beta$  such that for all $v\in V$,  one
has
\begin{equation}
\sup_{n\in\N_0}\p(\tilde{Y}_n=v)\le
\csechs e^{-\csieben|v|}.
\end{equation}
Furthermore, there exists a constant $\cacht>0$ such that 
$\p$-a.s.\
\[
\max_{k=0,\ldots,n}|\tilde{Y}_k|\le\  \cacht\log n
\quad\text{for all $n$ large enough.}
\]
\end{corollary}

\subsection{Plan of the paper and outline of the proof} 

Before starting the proof we  reorganize the expressions in a 
more convenient way. We  perform a change of
coordinates, replacing the $t,s$ variables by gradient variables 
taken along a fixed tree. We also replace the determinant in \eqref{eq:tmeasure} 
by a sum over the set of spanning trees. 
The measure $\mu^{\0}_{L}$  then factors in a product of a ``pinning'' measure 
on $t_{\0},s_{\0}$ and
a ``gradient''  measure on $\nabla t, y,T$, where $y$ is a rescaling 
of $\nabla s$ and $T$ is a spanning tree. 
This is done in subsections \ref{sect:grad-tree}-\ref{sect:newmeasure}.

Now, since the quantity we want to average is strictly positive,
we can include it in the measure. The problem is then to estimate
the normalization constant of this new ``interpolated measure''. 
This in turn can be translated in the problem of  minimizing a free energy 
with respect to a set of probability measures. The argument is given in 
Lemma~\ref{measurebound}. In the next subsection we introduce a particular
deformation of the interpolated measure: the guiding principle
behind is to move closer to the minimizer without moving too far away 
from the interpolated measure, in order to exploit its symmetry properties.

Sect.\  \ref{sect:entropy}, \ref{sect:locv} and \ref{sect:energy}   
are then devoted to prove that this deformed measure
gives a sufficiently good bound to ensure exponential localization.
The free energy we need to minimize consists of two terms: an energy 
term and an entropy term. In Sect.\ \ref{sect:entropy} we derive 
an upper bound for the entropy by a second order Taylor expansion.
For the energy term though we use 
a transfer operator method in Sect.\ \ref{sect:energy}. 
To apply the transfer operator we need first to rewrite the sum over 
(global) spanning trees in terms of new local tree variables. 
This is done in Sect.\ \ref{sect:locv}. Finally Sect.\ \ref{sect:end} 
puts together all the pieces to complete
the proof of the main theorem. There we also prove the results on VRJP.

In order to be as self-contained as possible, and since  the results  
on transfer operators are somehow scattered in the literature,
 we have collected in the appendix the parts we need 
for the convenience of the reader.

\paragraph{Notation.} In the following, constants are  labelled by $c_{1},c_{2}\dotsc $. 
They keep their meaning throughout the whole paper.


\section{Reorganizing the problem}
\label{sect3}

In this section we perform a change of variables
in order to make the gradient structure of the measure $\mu_L^\0$ more explicit.
We need to introduce first a few definitions.


\subsection{Gradient and tree variables}\label{sect:grad-tree}

\paragraph{Spanning trees and backbones.}
Let $ \T_L$ denote the set of spanning trees on $\G_L$. In the following, 
we write $\T$ instead of $\T_L$, unless there is a risk of confusion.
For each tree in $\T_{L}$ we define the  {\it backbone of $T$}, denoted by   $B(T)$,
as the unique path in $T$ connecting $\uroot= (-\ul,p)$ to  $\oroot= (\ol,p)$.
Moreover, we denote by $B^{c}(T)$ the set of edges in the complement of $B(T)$ 
inside $T$: these are the branches of the tree. Then $T=B (T)\cup B^{c} (T)$.
See Fig.\ref{fig0} for an example. 

We use a fixed reference tree defined in the following way.
 Let $B$ be the set of horizontal edges in $\bigcup_{n\in \Z}E_{n+1/2}$
connecting all copies of $\0 $: 
\begin{equation}\label{backbone}
B= \{p_{n+1/2}: -\ul\leq n\leq \ol-1,\, n\in \Z \}.
\end{equation}
We call this line the {\it backbone}.
Let $S$ be a fixed spanning tree of the finite graph $G_0$. 
We  define the {\it backbone tree} $T^{\bb }_{L}$ to be the spanning tree of $\G_L$
consisting of $\ul+1+\ol$ copies of the spanning tree $S$ 
which are just connected by the horizontal edges in $B$. 
In the following, 
we write $\tback $ instead of $\tback_{L}$, 
unless there is a risk of confusion.
With these definitions we have $B(\tback )=B$ (see Fig.\ref{fig0}).

\paragraph{Orienting the edges.}
We assign to every edge $e=(i\sim j)$ an arbitrary 
orientation from $i$ to $j$ for bookkeeping reasons only.
We define the {\it oriented gradient}
\begin{equation}\label{tydef}
 \nabla t_e=\nabla t_{i,j} := t_{j}-t_{i}, 
\qquad 
\nabla y_e=y_{i,j}:= (s_{j}-s_{i})e^{ \frac{t_{i}+t_{j}}{2}}.
\end{equation}
We will mostly use the notation $\nabla t_e,y_e$ as an argument of 
some even function, where the orientation (hence the sign) will not matter.

Let $T\in \T$ be an arbitrary spanning tree. 
In the following, any tree $T\in\T$ is oriented away from the root $\uroot$.
For each edge $e\in T$, we denote its endpoints by $i_{e,T}$ and $j_{e,T}$ such 
that the orientation of $e$ in $T$ goes from $i_{e,T}$ towards  $j_{e,T}$. 
Then we define the {\it oriented gradient along the tree $T$} as 
\begin{align}\label{gradtdef}
&  \nabla t^{T}_e:= t_{j_{e,T}}-t_{i_{e,T}}, \quad  
 y_e^{T}:=(s_{j_{e,T}}-s_{i_{e,T}})e^{ \frac{t_{i_{e,T}}+t_{j_{e,T}}  }{2}}\quad
\mbox{if} \ e\in T\\
 & \nabla t^{T}_e=    y_e^{T} := 0 \quad \mbox{if} \ e\not\in T. \nonumber
\end{align}
Of particular interest is $T=\tback $. As the other trees, 
the backbone tree is always oriented  away from the 
point $\uroot$. This corresponds to orient all edges in the spanning tree 
$S$ of $G_{0}$ away from the pinning point $\0 = (0,p)$ 
(likewise orient all edges
in the $n$-th copy of $S$ away from $(n,p)$) and orient each edge $p_{n+1/2}$ 
on the backbone $B$
from $(p,n)$ towards $(p,n+1)$. In this case, we abbreviate 
$i_e=i_{e,\tback}$ and $j_e=j_{e,\tback}$. 
\begin{figure}
\centerline{\psfig{figure=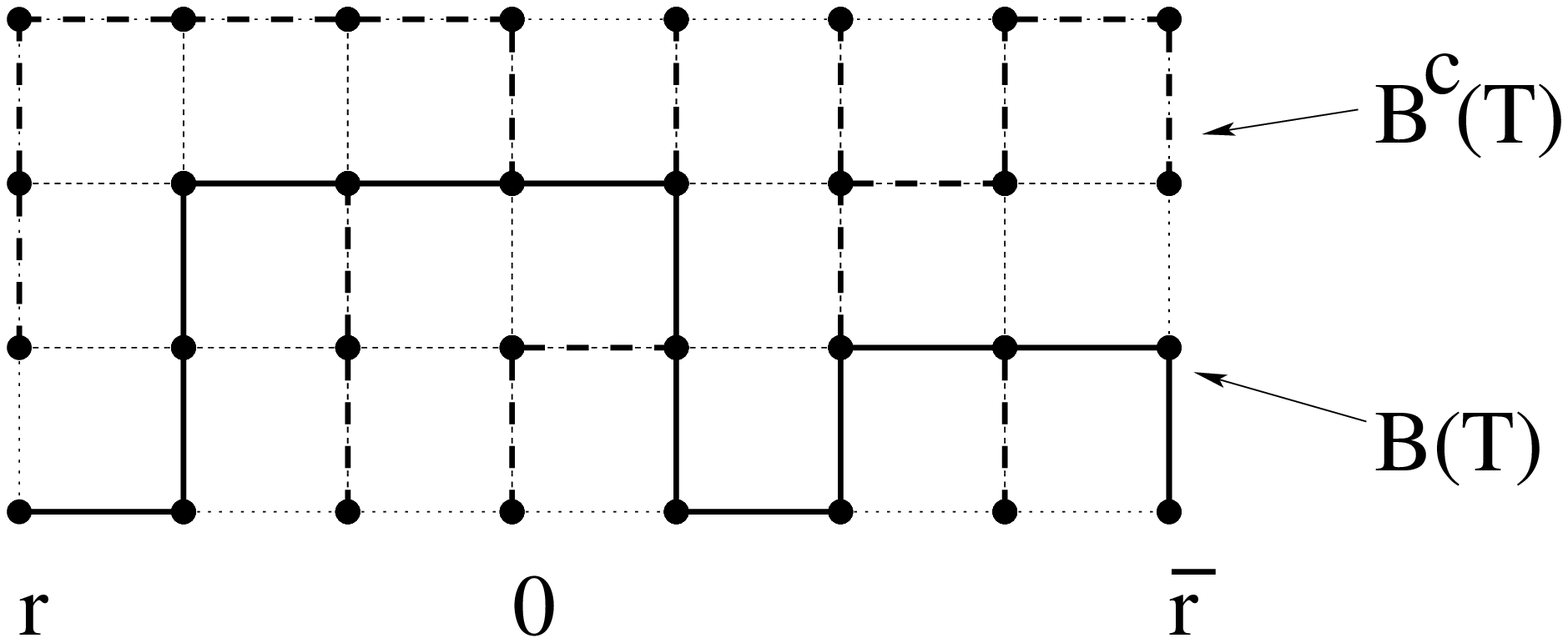 ,width=7cm}\hspace{0.8cm} \psfig{figure=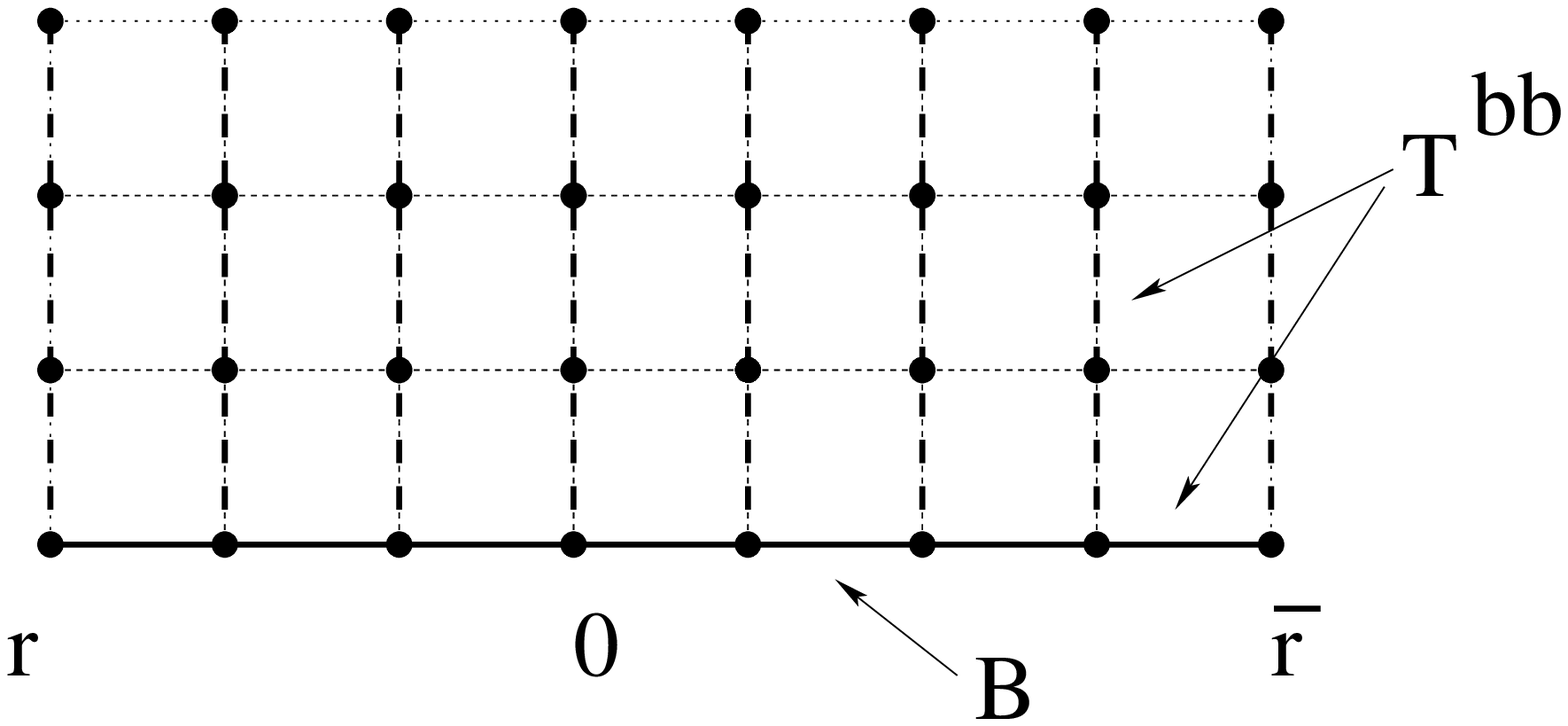 ,width=7cm}}
\caption{(a) an example of spanning tree $T$ with the corresponding sets $B(T)$ and $B^{c}(T)$; (b) 
the backbone tree $\tback$ with its backbone $B$.}
\label{fig0}
\end{figure}

\paragraph{Gradient variables.} In the following we replace 
$\mu_L^\0$ by a measure depending on  the set of spanning trees $\T $ 
defined above 
plus the set of {\it oriented gradients along the backbone tree}
\begin{equation}
 \nabla t_\bb=\ (\nabla t^{\bb }_e)_{e\in\tback}:=  
(\nabla t^{\tback  }_e)_{e\in\tback}, 
\qquad 
y_\bb= (y^{\bb}_e)_{e\in\tback}  :=(y^{\tback }_e)_{e\in\tback}.
\end{equation}
 Let 
\begin{align}
\Omega_L:=\R^{\tback}\times\R^{\tback}\quad\text{and}\quad
\Bomega_L:=\Omega_L\times\T_L
\end{align}
denote the set of all possible values of 
$\vec{\omega}:=(\nabla t_{\bb }, y_{\bb})$
and $\vec{\bomega}:=(\vec\omega,T)$, respectively. Finally we call 
$\omega_{n}$ (and $\omega_{n+1/2}$, respectively) 
the set of gradient variables associated to
``vertical edges'' $e\in  S_{n}$ at the  $n$-th level  (and
the gradient variables
associated to the unique horizontal edge in $\tback$ at level $n+1/2$,
respectively):
\begin{align}
\omega_n:=&(\omega_e)_{e\in S_{n}}:=
(\nabla t^{\bb }_e,y_e^{\bb})_{e\in S_{n}},\quad n=-\ul,\ldots,\ol,
\label{leveli}\\
\omega_{n+1/2}:=& (\nabla t^{\bb }_{p_{n+1/2}},y_{p_{n+1/2}}^{\bb}),
\qquad\qquad  n=-\ul,\ldots,\ol-1.\nonumber
\end{align}
All vertical variables $\omega_{n}$, $ n=-\ul,\ldots,\ol,$ belong to the same set  
$\Omega_{\vertical}=\R^{S}\times\R^{S} $, all horizontal variables  $\omega_{n+1/2}$,
$n=-\ul,\ldots,\ol-1$ belong to the same set  $\Omega_{\hor}=\R\times\R$.
Oriented gradient variables from \eqref{gradtdef} can be viewed as functions 
of gradient variables along the backbone tree as described in the 
following lemma. 

\begin{lemma}\label{lemma3.1}
Let $i,j\in V_L$ be two vertices. We can write $t_j-t_i$ as a function of 
$\nabla t_\bb$ as follows: 
\begin{align}
\label{repr-nabla-t}
t_{j}-t_{i}=
 \sum_{e'\in\gamma^{ij}_{\tback } }\ \nabla t^{\bb }_{e'}  
\left[\one_{\{e'\in  \gamma^{rj}_{\tback } \}}- \one_{\{e'\in \gamma^{ri}_{\tback } \}}  
\right], 
\end{align}
where $\gamma^{ij}_{\tback }$ is the unique path on the backbone tree connecting
the vertices $i$ and $j$, and  $ \one_{\{e\in   \gamma^{ij}_{\tback }\}}$ is the corresponding
indicator function.  For $(i\sim j)\in E_{n}$ 
the path $\gamma^{ij}_{\tback }$
is completely inside $S_{n}$. On the other hand, for $(i\sim j)\in E_{n+1/2}$ 
the path
uses only the edge $p_{n+1/2}$  and edges in $\tback $ at levels $n$ and $n+1$.

Finally, for any edge $e=(i\sim j)\in E$ directed  from 
$i$ to $j$, we have: 
\begin{align}
\label{nablat1}
&  y_{i,j} 
=  \sum_{e'\in \gamma^{ij}_{\tback }}  Y_{e'}, \quad \mbox{where}\quad 
Y_{e'}=Y_{e'}(\nabla t_{\bb },y_{\bb })
\quad\mbox{is given by}
\\ \nonumber 
& Y_{e'}= y_{e'}^\bb \left[  \one_{\{e'\in \gamma^{rj}_{\tback }\}}- 
\one_{\{e'\in\gamma^{ri}_{\tback }  \}}
\right]
\exp\Bigg\{\frac12\sum_{e''\in \gamma^{ij}_{\tback }\backslash \{e' \}} 
\nabla t^{\bb }_{e''}\left[1-2\cdot 
\one_{\{ e''\in \gamma^{ri_{e'}}_{\tback } \}} \right]\Bigg\}  
\end{align}
\end{lemma}
\begin{proof}
Every $t_{i}$ can be expressed 
in terms of $t_{\0 }$ and $\nabla t_{\bb }$: 
\begin{align}\label{ti}
t_{i}= t_{\0 } + (t_{i}-t_{r})- (t_{\0 }-t_{r}) 
&=  t_{\0 }
+\sum_{e\in \tback } \nabla t_e^{\bb } \left[ \one_{\{  e\in \gamma^{ri}_{\tback } \} }-  
\one_{\{  e\in\gamma^{r\0 }_{\tback   }  \}}  \right]
\end{align}
Therefore, for any vertices $i$ and $j$, the difference $t_j-t_i$ 
is a function of  $\nabla t_{\bb }$ only: 
\begin{align}
t_{j} - t_{i} =  \sum_{e\in \tback } \nabla t_e^{\bb }
\left[ \one_{\{e\in  \gamma^{rj}_{\tback } \}}-  
\one_{\{ e\in \gamma^{ri }_{\tback   } \}}  \right] 
=   \sum_{e\in \gamma^{ij}_{\tback } } \nabla t_e^{\bb }
 \left[ \one_{\{e\in  \gamma^{rj}_{\tback } \}}-  
\one_{\{ e\in \gamma^{ri }_{\tback   } \}}  \right], 
\label{grad}\end{align}
where we used $ \one_{\{e\in  \gamma^{rj}_{\tback } \}}-  
\one_{\{ e\in \gamma^{ri  }_{\tback   } \}}=0$ when $e\not\in \gamma^{ij }_{\tback}$.
In the same way $s_{j}$ can be expressed 
in terms of $t_{\0 },s_{\0 }, \nabla t_{\bb },y_{\bb }$:
\begin{equation}\label{si}
s_{i }= s_{\0 }
+\sum_{e\in \tback    } y_e^{\bb } e^{-\frac{t_{i_e}+t_{j_e}  }{2}} 
\left[ \one_{ \{ e\in \gamma^{ri}_{\tback } \}}-  
\one_{\{  e\in \gamma^{r\0 }_{\tback   } \}}  \right],
\end{equation}
where for each $e\in \tback $, 
\[
\frac{t_{i_e}+t_{j_e}  }{2}=
t_{\0 } + \frac{1}{2}
\sum_{e'\in \tback } \nabla t_{e'}^{\bb } 
\left[ \one_{\{e'\in \gamma^{ri_e}_{\tback } \} }+ 
\one_{ \{e'\in\gamma^{rj_e}_{\tback }  \}} - 2 \cdot  
\one_{\{ e'\in \gamma^{r\0 }_{\tback   } \}}  \right].
\]
For any edge $e=(i\sim j)$ directed from 
$i$ to $j$,  $y_{i,j}$ 
is a function of  $\nabla t_{\bb }$ and $y_{\bb }$ only: 
\begin{align}
&y_{i,j} = 
(s_{j}-s_{i}) e^{\frac{t_{i}+t_{j}}{2}}  = 
\sum_{e'\in  \gamma^{ij}_{\tback } } y_{e'}^{\bb } 
e^{\frac{t_{i}+t_{j}-t_{i_{e'}}-t_{j_{e'}}  }{2}} 
 \left[ \one_{\{e'\in  \gamma^{rj}_{\tback } \}}-  
\one_{\{ e'\in \gamma^{ri }_{\tback   } \}}  \right].
\end{align}
 The argument in the exponent is
\begin{align*}
&t_{i}+t_{j}-t_{i_{e'}}-t_{j_{e'}}  =
  \sum_{e''\in \gamma^{ij}_{\tback } } \nabla t_{e''}^{\bb }
\left[ \one_{\{ e''\in \gamma^{ri}_{\tback } \} } +  \one_{\{e''\in\gamma^{rj}_{\tback }  \} }
-  \one_{\{e''\in  \gamma^{ri_{e'}}_{\tback } \}} 
-  \one_{\{ e''\in \gamma^{rj_{e'}}_{\tback } \} }  \right],
\end{align*}
where  we used \eqref{grad}, $\gamma^{i_{e'}j_{e'}}_{\tback }\subset \gamma^{ij}_{\tback },$ and 
\[
\one_{\{ e''\in \gamma^{ri}_{\tback } \} } +  \one_{\{e''\in\gamma^{rj}_{\tback }  \} }
-  \one_{\{e''\in  \gamma^{ri_{e'}}_{\tback } \}} 
-  \one_{\{ e''\in \gamma^{rj_{e'}}_{\tback } \} }=0 \quad  \mbox{when}\quad  
e''\not\in \gamma^{ij}_{\tback}.
\]
Now we remark that
\[
( \gamma^{ri}_{\tback }\cap  \gamma^{rj}_{\tback })\cap  \gamma^{ij}_{\tback }=\emptyset,
\qquad  \gamma^{ij}_{\tback }= (\gamma^{ri}_{\tback }\cap \gamma^{ij}_{\tback } )
\cup  (\gamma^{rj}_{\tback }\cap \gamma^{ij}_{\tback } ).
\]
Then for each $e'\in \gamma^{ij}_{\tback } $ we have
\begin{align*}
&t_{i}+t_{j}-t_{i_{e'}}-t_{j_{e'}}  =
  \sum_{e''\in \gamma^{ij}_{\tback } } \nabla t_{e''}^{\bb }
\left[1
-  \one_{\{e''\in  \gamma^{ri_{e'}}_{\tback } \}} 
-  \one_{\{ e''\in \gamma^{rj_{e'}}_{\tback } \} }  \right]
\cr
&\qquad= \sum_{e''\in \gamma^{ij}_{\tback } } \nabla t_{e''}^{\bb }
\left[1
- 2\cdot  \one_{\{e''\in  \gamma^{ri_{e'}}_{\tback } \}} - \one_{\{ e''=e' \} }  \right]
= 
  \sum_{e''\in \gamma^{ij}_{\tback }\backslash \{e' \} } \hspace{-0.2cm} \nabla t_{e''}^{\bb }
\left[1
- 2\cdot  \one_{\{e''\in  \gamma^{ri_{e'}}_{\tback } \}}   \right]
\end{align*}
where we used  
$ \one_{\{ e''\in \gamma^{rj_{e'}}_{\tback }\}} = \one_{\{ e''\in \gamma^{ri_{e'}}_{\tback }\}}
+ \one_{\{ e''=e' \} }$.
This completes the proof of \eqref{nablat1}.
\end{proof}

\smallskip


\subsection{The measure $\mu_L^\0 $}\label{sect:newmeasure}

In the following $\nabla t_{e}, y_{e}, \nabla t_{e}^{T}$ 
and also $t_{j}-t_{i}$,  for any edge $e$ 
and any two vertices $i,j$, are viewed as functions 
of $\vec{\omega}$. This is possible by Lemma \ref{lemma3.1}.
With all the definitions from above we have

\begin{theorem}\label{th:muzero}
Consider the transformation 
\begin{align}
\label{def takegrad}
\takegrad:
\R^V\times\R^V&\to(\R\times\R)\times\Omega_L, 
\cr (t_i,s_i)_{i\in V_L}&\mapsto 
(t_\0,s_\0, \nabla t_\bb,y_\bb)=(t_\0,s_\0,\vec\omega).
\end{align}
With respect to this transformation, 
the image of the probability measure $\mu_L^\0$ equals the product 
measure $\mu^\pin\times \mu_{L,\notree}^{\grad,\0 }$,
where $\mu^\pin$ and $\mu_{L,\notree}^{\grad,\0 }$ are two probability 
measures defined on $\R\times \R$ and $\Omega_L$, respectively.
These measures are defined as follows:
\begin{align}
& d\mu^\pin (t_{\0 },s_{\0 }) =
e^{-H^{\pin}(t_{\0},s_\0)} dt_{\0 }ds_{\0 } \quad \text{ with}\cr
& H^{\pin}(t_\0,s_\0) = \varepsilon \left[  \cosh t_\0 -1+ 
\frac{s_{\0}^{2} e^{t_{\0}}}{2} \right]
-\ln \frac{\varepsilon }{2\pi }
,
\end{align}
where  $dt_{\0 }\,ds_{\0 }$ denotes
the Lebesgue measure on $\R\times \R$.
The measure $\mu_{L,\notree}^{\grad,\0 }$ is obtained from a 
probability measure $\mu_L^{\grad,\0 }$ on $\Bomega_L$ by summing over 
all spanning trees. More precisely, $\mu_{L,\notree}^{\grad,\0 }$ is 
 the marginal with respect to the 
projection $\Bomega_L\to\Omega_L$, $\vec\bomega=
(\vec\omega,T)\mapsto\vec\omega$ of 
the probability measure $\mu_L^{\grad,\0 }$ on $\Bomega_L$ defined as follows:
\begin{align}
 d\mu_L^{\grad,\0 } (\vec{\bomega})   = & 
e^{-H_L^{\grad,\0 }(\vec{\bomega})} d\vec{\bomega} \quad\text{with }\cr 
 H_L^{\grad,\0 }(\vec{\bomega}) = & \sum_{e\in E_L} \beta_e 
\left[ \cosh \nabla t_e-1 +  \frac{y_e^2}{2} \right] 
+ \sum_{e\in T} \nabla t_e^{T}  - \sum_{e\in \tback } \frac{\nabla t_e^{\bb } }{2}
\cr
&
 + t_r -t_{\0}- \sum_{e\in T} \ln \frac{\beta_e }{2\pi } \label{mugrad}
\end{align}
and  $d\vec{\bomega}=d\vec\omega\,dT=
\prod_{e\in \tback  }d\nabla t_e^\bb\,dy_e^\bb\, dT$  
is the Lebesgue measure
on $\Omega_{L}$ times the counting measure $dT$ on $\T$.
\end{theorem}
\begin{proof}
By the matrix tree theorem the determinant of $A_L(t)+ \widehat{\varepsilon}$ 
can be written as
\begin{align}
\det [A_L(t)+ \widehat{\varepsilon}\,] = 
\varepsilon e^{t_{\0}} 
\sum_{T\in \T} \prod_{(i\sim j) \in T} 
\beta_{ij} e^{t_{i}+t_{j}}=\sum_{T\in \T}  
e^{ t_\0+ \ln\varepsilon + \sum_{(i\sim j) \in T} 
\left( t_{i}+t_{j} + \ln\beta_{ij} \right)}.
\end{align}
Therefore we can rewrite the measure $\mu^{\0 }_L$ as a marginal, by taking the spanning tree $T$ 
as additional variable. We have then 
\[
 d\mu^{\0 }_L (t,s)  = \int_{T\in \T} e^{-H_L (t,s,T)} d[t,s,T], \quad 
\mbox{where}\quad  d[t,s,T]  =    
 \prod_{j\in  V_L} \frac{dt_{j}ds_{j}}{2\pi } dT,\quad  \mbox{and}
\]
\begin{align}
H_L (t,s,T) & =  \sum_{(i\sim j)\in E_L } \beta_{ij} 
\left[\cosh(t_i - t_j)-1+ \frac{(s_i-s_j)^{2}}{2}e^{t_{i}+t_{j}} \right]
+ \sum_{j\in V_L} t_{j}   \cr
& +\varepsilon\left [\cosh t_{\0}-1+ \frac{s_{\0 }^{2}}{2}e^{t_{\0 }}\right ]
 - \left [ t_\0+ \ln\varepsilon + \sum_{(i\sim j) \in T} 
\left( t_{i}+t_{j} + \ln\beta_{ij} \right)\right].
\label{Hamiltonian-HL}
\end{align}
The normalizing constant $(2\pi)^{-|V_L|}$ is distributed in pieces 
among the terms $\ln(\beta_e/(2\pi))$, $e\in T$, and 
$\ln(\varepsilon/(2\pi))$  appearing in $H_L^{\grad,\0 }$ and $H^\pin$, 
respectively. 
The transformation $\takegrad$ from \eqref{def takegrad} is a  bijection. 
Changing the variables according to 
$(t_\0,s_\0,\nabla t_\bb,y_\bb)=\takegrad(s,t)$
yields the transformed measure 
\begin{align}
\takegrad[e^{-H_L(t,s,T)}dt ds] = e^{-H_L(  \takegrad^{-1}(t_\0,s_\0,\nabla t_\bb,y_\bb),T)} J
\,
dt_\0\,ds_\0\,d\nabla t_\bb\,dy_\bb
\end{align}
where the last term $J$
is the Jacobian
\begin{align}
\label{Jacobian-J}
J
=\prod_{(i\sim j)\in\tback}e^{-\frac12(t_i+t_j)}= e^{-\left[ \sum_{j\in V_{L}}t_{j}-t_{r}-
\frac{1}{2}\sum_{e\in \tback } \nabla t_e^{\bb }  \right]  },
\end{align}
and in the last equality  we used  Lemma \ref{lemma-htree-rewritten} below.
Using Lemma \ref{lemma-htree-rewritten} again 
\[
\sum_{j\in V_L} t_j  -  t_\0 - \sum_{(i\sim j) \in T} 
\left( t_{i}+t_{j}  \right)  =- \sum_{j\in V_L} t_j 
 +  t_{r} + (t_{r}-t_{\0 }) + \sum_{e\in T}\nabla t_e^{T} .
\]
Inserting this in the Hamiltonian \eqref{Hamiltonian-HL}, we get
\begin{align}
H_L (t,s,T) &=  \sum_{e\in E}\beta_e
\left[\cosh \nabla t_e-1+ \frac{y_e^{2}}{2} \right]
+  \sum_{e\in T} \nabla t^{T}_e - \sum_{j\in V_L} t_{j} + t_{r}+ (t_{r}-t_{\0 })\cr
&-  \sum_{e\in T}\ln \beta_e 
+\varepsilon \left[\cosh t_{\0}-1 + \frac{s_{\0 }^{2}e^{t_{\0 }}}{2} \right] 
- \ln \varepsilon.   \label{Hl}
\end{align}
Adding the contribution \eqref{Jacobian-J} from the Jacobian the result follows.
\end{proof}

\begin{lemma}
\label{lemma-htree-rewritten}
For every $T\in\T$, one has 
\begin{equation}
\sum_{e\in T} (t_{i_e}+t_{j_e})  -2\sum_{j\in V_L} t_j
+2t_{r} + \sum_{e\in T}\nabla t_e^{T}=0.
\end{equation}
\end{lemma}
\begin{proof}
Recall that for each edge $e\in T$, we denote its endpoints 
by $i_{e,T}$ and $j_{e,T}$ such 
that $j_{e,T}$ is farther away from $r$ in the tree $T$ than $i_{e,T}$: 
$t_{j_{e,T}}-t_{i_{e,T}}= \nabla t^{T}_e.$
For every vertex $j\in V_L\setminus\{r\}$ there is a unique 
edge $e\in T$ with $j_{e,T}=j$, therefore 
\begin{align} 
\sum_{j\in V_L} t_j = t_r+\sum_{e\in T}t_{j_{e,T}}.
\end{align} 
Using this, we get
\begin{align}
\sum_{e\in T} (t_{i_e}+t_{j_e})  -2\sum_{j\in V_L} t_j
+ 2t_r
= &\sum_{e\in T} (t_{i_{e,T}}+t_{j_{e,T}})  
-2\left(  t_r+\sum_{e\in T}t_{j_{e,T}} \right)
+2t_r\nonumber\\
= & \sum_{e\in T}(t_{i_{e,T}} - t_{j_{e,T}})
=  - \sum_{e\in T} \nabla t^{T}_e. 
\end{align}
The claim follows. 
\end{proof}

\subsection{The interpolated  measure}

Let $l\in\N$ with $0<l\leq \ol$. The reader may imagine $1\ll l\ll\ol$.
We set $\ell:=(l,p)$. Thus, $\ell$ is the copy of $p$ at level $l$. 
We are interested in studying the average
$\mathbb{E}_{\mu_{L}^{\0 }}\left[ e^{\frac{t_{\ell}-t_{\0 } }{2}} \right]$.
Note that
\begin{align}
\label{expectation-exp-tl-t0}
\mathbb{E}_{\mu_{L}^{\0 }} \left[e^{\frac{t_{\ell}-t_{\0 } }{2}} \right]
= \int  e^{\frac{t_{\ell}-t_{\0 } }{2}}  d\mu_L^{\grad, \0 } (\vec{\bomega}), 
\end{align}
since 
\[
\int d\mu^\pin (t_{\0 },s_{\0 })  = 1.
\]
The last expression is true  by supersymmetry,  being the partition function for the 
probability measure \eqref{eq:tmeasure}  in 
the special case of a single vertex. In this simple case one may check the identity also
by direct computation.
We now merge the observable $e^{\frac{t_{\ell}-t_{\0 } }{2}}$
 with $H_L^{\grad, \0 }(\vec{\bomega})$ to define a new probability 
measure, called the {\it interpolated measure}, 
\begin{equation}
\label{def-Z}
d\Pint  = \frac{ e^{\Delta H}d\mu_L^{\grad,\0 }}{\Zint},\quad 
\mbox{where} 
\ \Delta H= \frac{t_{\ell}-t_{\0 }}{2}
\ , \ \Zint = \mathbb{E}_{\mu_L^{\grad,\0 }} 
\left[e^{\frac{t_{\ell}-t_{\0 } }{2}} \right].
\end{equation}
The normalization constant of this new measure 
is exactly the observable we want to estimate.

\begin{theorem}\label{interp-m} 
Using the gradient variables above the interpolated measure $d\Pint $
can be written as
\begin{equation}\label{decomp HLlo}
d\Pint  (\vec{\bomega})  = \frac{e^{- H_{L}^{\0 \ell}(\vec{\bomega}) } }{\Zint}
d\vec{\bomega},\; 
\mbox{with}\; H_{L}^{\0 \ell}(\vec{\bomega}) 
   = \sum_{e\in E_{L}} h_e (\vec{\bomega})
-\hspace{-0.1cm} 
\sum_{n=-\ul}^{-1} \frac{\nabla t^{\bb }_{p_{n+1/2}} }{2}  
+\hspace{-0.1cm}  
\sum_{n=l}^{\ol-1 } \frac{\nabla t^{\bb}_{p_{n+1/2}} }{2},
\end{equation} 
where
\begin{align}
  h_e (\vec{\bomega}) &=
\beta_e \left[ \cosh \nabla t_e-1 +  \frac{y_e^2 }{2}  \right] +
 \nabla t_e^{T} \one_{\{e\in B^{c} (T) \}} - \frac{\nabla t_e^{\bb } }{2} 
\one_{\{e\in B^{c}(\tback) \}}-\ln \frac{\beta_e}{2\pi }\, \one_{\{e\in T \}}.
\label{Hn12}
\end{align}
\end{theorem}
\begin{proof}
By \eqref{mugrad} above, the Hamiltonian for the interpolated measure is 
\begin{align*}
H_{L}^{\0 \ell}(\vec{\bomega})&= H_L^{\grad,\0 } + \frac{t_{\0 }-t_{\ell}}{2} = 
 \sum_{e\in E_L} \beta_e 
\left[ \cosh \nabla t_e-1 +  \frac{y_e^2}{2} \right] 
 - \sum_{e\in T} \ln \frac{\beta_e }{2\pi }
\cr
&+  \sum_{e\in T} \nabla t_e^{T}  - \sum_{e\in \tback } \frac{\nabla t_e^{\bb } }{2}
 - (t_{\0}-t_{\uroot}) -  \frac{t_{\ell }-t_{\0}}{2} .
\end{align*}
We reorganize the terms that are not already in local form as
\[
- (t_{\0}-t_{\uroot})  -  \frac{t_{\ell }-t_{\0}}{2} = 
- \frac{t_{\0}- t_{\uroot}}{2}  -  \frac{t_{\oroot}- t_{\uroot}}{2}  + 
   \frac{t_{\oroot}- t_{\ell}}{2}. 
\]
We decompose the middle term (that no longer depends on $\0 $ or $\ell$)
as
\[
 -  \frac{t_{\oroot}- t_{\uroot}}{2} =  - (t_{\oroot}- t_{\uroot}) 
+  \frac{t_{\oroot}- t_{\uroot}}{2} =
-  \sum_{e\in B (T)} \nabla t_e^{T}+  \sum_{e\in B }
\frac{\nabla t_e^{\bb } }{2},
\]
where we used the definition of $B(T)$ and $B$ given in Sect.\  
\ref{sect:grad-tree}.
The other terms can be decomposed as sums along the backbone $B$.
Inserting all this in  the formula above, and writing $t_{\0}- t_{\uroot}$, 
$t_{\oroot}- t_{\ell}$ as telescopic sums along the backbone,  
we have
\begin{align*}
H_{L}^{\0 \ell}(\vec{\bomega})&= \sum_{e\in E_L} \beta_e 
\left[ \cosh \nabla t_e-1 +  \frac{y_e^2}{2} \right] 
 - \sum_{e\in T} \ln \frac{\beta_e }{2\pi }
+  \sum_{e\in T} \nabla t_e^{T} \left[1-\one_{\{e\in B(T) \}}  \right]\cr
& 
-  \sum_{e\in \tback }\frac{\nabla t_e^{\bb } }{2}    \left[1-\one_{\{e\in B \}}  \right] 
- \sum_{n=-\ul }^{-1} \frac{\nabla t^{\bb }_{p_{n+1/2}} }{2}  
+ \sum_{n=l }^{\ol-1} \frac{\nabla t^{\bb }_{p_{n+1/2}} }{2} .
\end{align*}
Isolating the contribution of each edge the result follows. 
\end{proof}

\begin{remark}
\label{rem-path-gamma} 
Note that from \eqref{repr-nabla-t} and \eqref{nablat1} 
the gradients $\nabla t_e^T,y_e^T$
for $e= (i\sim j)\in E$, given the direction of $e$ in $T$,
depend only on the independent
(backbone tree) 
variables $\nabla t^{\bb }_{e'},y^{\bb }_{e'}$
associated to the unique path in $\tback$ connecting $i$ to $j$.
When $e\in E_{n}$ this path belongs completely to $S_{n}$, therefore 
contains only vertical edges at level $n$. 
On the other hand when  $e\in E_{n+1/2}$
the path may contain edges in $S_{n}$, edges in $S_{n+1}$, plus the unique
edge in $\tback \cap E_{n+1/2}$.
Therefore the contribution $h_e$ to the Hamiltonian for $e\in E_{n}$  
depends only on gradient variables associated
to vertical edges at level $n$ (plus the tree $T$).
On the other hand, when $e\in E_{n+1/2}$, $h_e$ depends on
 gradient variables associated to vertical edges at levels $n$ and $n+1$, i.e.\ 
$\omega_{n}$ and $\omega_{n+1}$ plus
  the gradient variables $\omega_{n+1/2}$ 
associated to the unique horizontal  edge $p_{n+1/2}$
in $\tback $ at level $n+1/2$, and 
finally the tree $T$. 
\end{remark}

The problem is now to estimate the normalization constant of the
interpolated measure. We will need the following result 
(in the context of edge-reinforced random walks, such an estimate is
shown as Lemma 5.1 of \cite{merkl-rolles-2d}).
\begin{lemma}\label{measurebound}
The normalization constant $\Zint$ satisfies 
\begin{equation}\label{lnZbound}
\ln \Zint \leq  \mathbb{E}_{\Pi } [\Delta H] +  
\mathbb{E}_{\Pi  }\left[ \ln \left( \frac{d\Pi }{d\Pint  } 
\right)  \right]
\end{equation}
for any probability measure $\Pi$ having a positive density with respect to 
$\Pint $, such that $ \mathbb{E}_{\Pi } [\Delta H]$ and 
$\mathbb{E}_{\Pi}\left[  \ln \frac{d\Pi  }{d\Pint }  \right]   $
are both  finite.
Equality holds at 
$\Pi = \mu^{\grad,\0}_{L}$.
Moreover 
\begin{equation}\label{entropylowerb}
\mathbb{E}_{\Pi  }\left[ \ln \left( \frac{d\Pi }{d\Pint  } 
\right)  \right] \geq 0 
\end{equation}
and takes value zero at $\Pi=\Pint $.
\end{lemma}
\begin{proof}
By definition we have
\[
d\Pint  = \frac{ e^{\Delta H}d\mu_L^{\grad,\0 }}{\Zint}\quad 
\Rightarrow\quad  \Zint= e^{\Delta H} \frac{d\mu_L^{\grad,\0 } }{d\Pint }
\quad 
\Rightarrow\quad  \ln  \Zint = \Delta H + \ln \frac{d\mu_L^{\grad,\0 } }{d\Pint }
\]
where $d\mu_L^{\grad,\0 }/d\Pint $ is the Radon-Nikodym derivative.
This equality is true pointwise hence
\[
 \ln  \Zint = \mathbb{E}_{\Pi }\left[  \ln  \Zint\right] =
\mathbb{E}_{\Pi}\left[  \Delta H \right]\ + \ 
\mathbb{E}_{\Pi}\left[  \ln \frac{d\mu_L^{\grad,\0 } }{d\Pint }  \right]
\]
for any probability measure $\Pi$ such that  $ \mathbb{E}_{\Pi }[\Delta H]$  
is finite.  Moreover, given that 
$\mathbb{E}_{\Pi}\left[  \ln \frac{d\Pi  }{d\Pint }  \right]  $ is finite,
we have  
\begin{align}
 \mathbb{E}_{\Pi}\left[  \ln \frac{d\mu_L^{\grad,\0 } }{d\Pint }  \right] 
&= \mathbb{E}_{\Pi}\left[  \ln \frac{d\Pi  }{d\Pint }  \right]  
+  \mathbb{E}_{\Pi}\left[  \ln \frac{ d\mu_L^{\grad,\0 } }{ d\Pi } \right]\cr  
& \leq   \mathbb{E}_{\Pi}\left[  \ln \frac{d\Pi  }{d\Pint }  \right]  
+  \mathbb{E}_{\Pi}\left[   \frac{ d\mu_L^{\grad,\0 } }{ d\Pi }-1 \right]= 
 \mathbb{E}_{\Pi}\left[  \ln \frac{d\Pi  }{d\Pint }  \right] 
\end{align}
where we applied   $\ln x\leq x-1$ $\forall x>0$. 
The inequality becomes sharp for
 $\Pi = \mu_L^{\grad,\0 }$.
Finally 
\begin{equation}\label{logpibound}
 \mathbb{E}_{\Pi}\left[  \ln \frac{d\Pi  }{d\Pint }  \right]= -
\mathbb{E}_{\Pi}\left[  \ln \frac{d \Pint  }{d\Pi}\right] \geq -
\mathbb{E}_{\Pi}\left[   \frac{d \Pint  }{d\Pi}-1\right] = 0
\end{equation}
 This concludes the proof.
\end{proof}

\subsection{The deformed measure} 
\label{sect:defm}

Using Lemma \ref{measurebound} above, the problem of estimating the decay of 
$\mathbb{E}_{\mu_{L}^{\0 }} \left[e^{\frac{t_{\ell}-t_{\0 } }{2}} \right]$ 
can be translated in
bounding the free energy $\ln \Zint$ by \eqref{lnZbound}. To prove exponential decay of
$ \Zint$ we  need to find
a measure $\Pi $ such that 
\begin{itemize}
\item [(a)] the energy term in \eqref{lnZbound} satisfies 
$\mathbb{E}_{\Pi}\left[  \Delta H \right]\leq C_0- C_1 l$
for some positive constants $C_0,C_1$,
\item [(b)]  the entropy term in \eqref{lnZbound}  satisfies 
$\mathbb{E}_{\Pi}\left[  \ln \frac{d\Pi  }{d\Pint }  \right]\leq  C_2 l$
with $0\leq C_2 < C_1$.  
\end{itemize}
Since $\mathbb{E}_{\Pi}[  \ln \frac{d\Pi  }{d\Pint }]\geq 0$
by \eqref{entropylowerb}, 
 there is no hope to get a negative contribution from 
the entropy term.  Ideally we should take $\Pi=\mu_L^{\grad,\0 }$
to optimize the estimate, 
but in practise this is too hard. On the other hand,  
if we take $\Pi =\Pint $ the
entropy term is exactly zero. 
Guided by these facts, we will take a deformation  $\Pi_\alpha$ of 
$\Pint $, where $\alpha\in \R$ is a deformation parameter such that
\begin{itemize}
\item  $\alpha $ is close enough to zero so that the entropy term remains near zero and
\item   the deformed measure $\Pi_\alpha$ 
is close to the minimum $\mu_L^{\grad,\0 }$. 
\end{itemize}
This will be made more precise below.
 
\paragraph{The deformation.}
We introduce a small deformation $\xi_{\alpha }$ acting  only on the gradient 
variables 
$\nabla t^{\bb }_{p_{n+1/2}}$ of the backbone  in the unique path connecting 
$\0 $ to $\ell$:
\begin{equation}\label{deform}
\xi_{\alpha } :  \Bomega_L  \to  \Bomega_L ,\qquad    \vec{\bomega}=(\nabla t_\bb,y_\bb,T) 
\mapsto 
 \vec{\bomega}_\alpha=(\nabla t^\alpha , y_\bb,T),
\end{equation}
where
\begin{align}
 \nabla t^\alpha_e & =  \nabla t^{\bb }_e\qquad \qquad \qquad   \qquad \mbox{if}\ 
e \in \tback \backslash\gamma^{\0 \ell}_{\tback }\cr %
\nabla t^\alpha_{p_{n+1/2}} & =  \nabla t^{\bb }_{p_{n+1/2}}
 + \alpha  \chi_{n+1/2}\qquad  \mbox{if} \ 0\leq n\leq l-1\label{defxi1}
\end{align}
and $\chi_{n+1/2}=\chi(\omega_n,\omega_{n+1/2},\omega_{n+1})$ with
a cutoff function 
$\chi:\Omega_\vertical\times\Omega_\hor\times\Omega_\vertical\to[0,1]$
given by
\begin{align}\label{chidef}
\chi(\omega,\omega_\hor,\omega')=  
\tilde{\chi}(\eta^{-2}\|\omega_\hor\|^2)
\prod_{e\in S}\left[\tilde{\chi}(\eta^{-2}\|\omega_e\|^2)
\tilde{\chi}(\eta^{-2}\|\omega'_e\|^2)\right]
\end{align}
with $\tilde\chi: \mathbb{R}\to[0,1]$ being a smooth 
decreasing non negative function 
such that $\tilde\chi(x)=1$ for $x\leq 1/2$ and $\tilde\chi(x)=0$ 
for $x\geq 1$. The variables $\omega_n$, $\omega_{n+1/2}$ and $\omega_e$ 
were defined in equation \eqref{leveli}, and $\|{\cdot}\|$ denotes the Euclidean norm.
Finally $\eta$ is any fixed positive constant.
Here, its value is irrelevant, but it may become important to optimize
the bounds quantitatively.

\begin{lemma}
\label{lemma-xi-gamma-invertible}
The transformation $\xi_{\alpha }$ defined in 
\eqref{deform} is invertible for all
values of $|\alpha|\leq \eta \cneun$, where $\cneun$ is 
any fixed positive constant satisfying
$\cneun <  (2\|\tilde{\chi}'\|_{\infty})^{-1}$. Furthermore, for all 
$n=0,\ldots,l-1$, one has 
\begin{align}
\label{bound-derivative-chi}
\left| \frac{\partial\chi_{n+1/2}}{\partial\nabla t^\bb_{p_{n+1/2}}}\right| 
\le 2\eta^{-1} \|\tilde{\chi}'\|_\infty.
\end{align}
\end{lemma}
\begin{proof}
By \eqref{defxi1}, only variables $\nabla t^\alpha_{p_{n+1/2}}$
with $n=0,\ldots ,l-1$ are modified and each $\nabla t^\alpha_{p_{n+1/2}}$
does not depend on the other $\nabla t^\alpha_{p_{m+1/2}}$, $m\neq n$. 
Let us freeze all the unchanged variables.
The transformation $\xi_{\alpha }$
then reduces to a set of $l$ one dimensional functions
$\xi_{n,\alpha }:\R \to \R $, $n=0,\ldots,l-1$, defined by 
\[
\xi_{n,\alpha } (u)=u+\alpha f_{n} (u), \mbox{ where} \ 
f_{n} (u)= \tilde{\chi}(\eta^{-2} (u^{2}+y_{p_{n+1/2}}^{2} ))
\prod_{e\in S}\left[\tilde{\chi}(\eta^{-2}\|\omega_{e_n}\|^2)
\tilde{\chi}(\eta^{-2}\|\omega_{e_{n+1}}\|^2)\right].
\]
These functions $\xi_{n,\alpha}$
coincide with the identity for any $|u|>\eta$ and are injective for all
$\alpha $ satisfying 
$2\eta^{-1} |\alpha| \|\tilde{\chi}'\|_\infty <1$ because of 
\begin{align}
\label{upper-bound-f-prime}
\sup_{u\in\R} |f_n' (u)|=\sup_{|u|\le \eta} |f_n' (u)|
\leq  2\eta^{-1}  \|\tilde{\chi}'\|_\infty .
\end{align} 
Taking a constant  $0<\cneun <(2\|\tilde{\chi}'\|_{\infty})^{-1}$,
we conclude that
for all $|\alpha | \leq \eta \cneun$ each
function $\xi_{n,\alpha }$ is a bijection.
The bound \eqref{bound-derivative-chi} follows also
from \eqref{upper-bound-f-prime}.
\end{proof}\vspace{0.2cm}

With these definitions we introduce the {\it deformed measure $\Pi_{\alpha }$} 
\begin{equation}\label{pigamma}
\Pi_{\alpha } (A)= \xi_{\alpha }[ \Pint ] (A):= 
\Pint  \left( \xi_{\alpha }^{-1}  A \right) \qquad \forall 
A\subseteq \Bomega_L  \mbox{ measurable.}
\end{equation}

\begin{lemma}
\label{upper-bound-Z-Pi-gamma}
Using the deformed measure $\Pi_{\alpha }$, we have
\begin{align}
\ln \Zint \leq  \Egamma + \Sgamma{\alpha} 
\label{boundfree}\end{align}
where 
\begin{align}\label{Egamma} 
 \Egamma &= \mathbb{E}_{\Pi_{\alpha } } [\Delta H]
=
\frac12\sum_{n=0}^{l-1}\mathbb{E}_{ \Pint  } 
\left[\nabla t^\bb_{p_{n+1/2}} +
\alpha\chi_{n+1/2}
\right]
   \\
 \Sgamma{\alpha} & = \mathbb{E}_{\Pi_{\alpha }  }
\left[ \ln  \frac{d\Pi_{\alpha } }{d\Pint  } \right] =
\mathbb{E}_{ \Pint  } \left[\Hint\circ \xi_{\alpha }- 
 \Hint-  \ln |\det D\xi_{\alpha }|\  \right]\label{Sgamma}
\end{align}
and $ D\xi_{\alpha }$ is the Jacobian matrix for the deformation.
\end{lemma}
\begin{proof}
The inequality \eqref{boundfree} follows from Lemma \ref{measurebound}. 
To obtain \eqref{Egamma} we use the representation
\[
\Delta H=\frac{t_{\ell}-t_{\0 }}{2}=\frac12 \sum_{n=0}^{l-1}\nabla t^\bb_{p_{n+1/2}}
\]
from \eqref{def-Z}  and the deformation \eqref{defxi1}
to see
\begin{align}
 \Egamma &= \mathbb{E}_{\Pi_{\alpha } } [\Delta H]=
\mathbb{E}_{ \Pint  } \left[\ \Delta H \circ \xi_{\alpha }\   \right]
=
\frac12\sum_{n=0}^{l-1}\mathbb{E}_{ \Pint  } 
\left[\nabla t^\bb_{p_{n+1/2}} +
\alpha\chi_{n+1/2}
\right].
\end{align}
To obtain \eqref{Sgamma}, we notice that 
\begin{align*}
 \frac{d\Pi_{\alpha } }{d\Pint  }&=  
\frac{d \xi_{\alpha }[\Pint ] }{d\Pint  }=
\frac{d \xi_{\alpha }[\Pint ] }{d \xi_{\alpha }[\lambda ] }
\frac{d \xi_{\alpha }[\lambda ] }{d\lambda }   
\left( \frac{d\Pint  }{d\lambda  } \right)^{-1}\cr
&\quad =
\left( \frac{d\Pint  }{d\lambda  }\circ \xi_{\alpha }^{-1}  \right) 
\frac{1}{\left|\det D\xi_{\alpha }  \right|\circ \xi_{\alpha }^{-1}} 
\left( \frac{d\Pint  }{d\lambda  } \right)^{-1}
\end{align*}
where $d\lambda(\vec\bomega)=d\vec\bomega$ is the Lebesgue measure times
the counting measure. Then 
\begin{align*}
 \Sgamma{\alpha} &=  \mathbb{E}_{\xi_{\alpha }[\Pint ] }
\left[ \ln  \frac{d\Pi_\alpha}{d\Pint  } \right]
=  \mathbb{E}_{\Pint  } 
\left[ \ln  
\frac{d\Pi_\alpha}{d\Pint  } \circ \xi_{\alpha } \right]
\cr
& =  \mathbb{E}_{\Pint  } 
\left[
 \ln  \frac{d\Pint  }{d\lambda  } \ - \ 
\ln \left(\frac{d\Pint  }{d\lambda  } \circ \xi_{\alpha } \right)
- \ln \left|\det D\xi_{\alpha }  \right|\right]
\end{align*}
Using $d\Pint /d\lambda = e^{-\Hint}/\Zint$ 
we conclude the proof.
\end{proof}

In the next two sections we prove separately the bounds on the entropy and 
energy term. The techniques for the two bounds are quite different: for the entropy 
we use a Taylor expansion while for the energy term we need to set up a transfer
operator approach.

\section{The entropy contribution}\label{sect:entropy}

\begin{theorem}\label{entropyth}
For any given $\eta>0$ and $\vec\beta$ the entropy contribution satisfies
\begin{equation}\label{eq-entr}
\Sgamma{\alpha}\ = \  \mathbb{E}_{\Pi_{\alpha }}
\left[  \ln \frac{d\Pi_{\alpha }  }{d\Pint }  \right]\ 
\leq \  \cfuenf \alpha^2l
\end{equation}
for all $\alpha\in\R$ with $|\alpha|\le\cneun\eta$ and some constant
$\cfuenf(\betamax,G_0,\eta)>0$. 
Here, $\cneun>0$ is the constant from Lemma \ref{lemma-xi-gamma-invertible},
$\betamax = \max_{e\in E_{1/2} } \{\beta_e \}$ and 
$\eta$ is the parameter appearing in the definition of $\xi_{\alpha }$. 
This bound holds uniformly in $L$.
\end{theorem}
\begin{proof}{}\label{}
The derivatives of $\Sgamma{\alpha}$ can be calculated by differentiating 
the argument of the expectation in \eqref{Sgamma}. This is possible
because the cutoff function $\chi$, defined in 
\eqref{chidef}, is compactly supported.
By relation  \eqref{entropylowerb}, the entropy is always positive
or zero: 
$\Sgamma{\alpha}\geq 0$ for all $\alpha$.
Moreover 
\[
\Sgamma{0}=   \mathbb{E}_{\Pint}
\left[  \ln \frac{d\Pint}{d\Pint}  \right] = 0.
\]
Therefore $[\partial_{\alpha }\Sgamma{\alpha}]_{\alpha=0}=0$ 
and the first non zero
term in the  Taylor expansion for $\alpha $ is the second derivative.
Hence
\begin{align}
\label{S-alpha-derivative}
\Sgamma{\alpha} =  
\frac{\alpha^{2}}{2}\  \frac{\partial^{2}}{\partial\tilde{\alpha }^{2}} 
\Sgamma{\tilde\alpha}  
\  = \ \frac{\alpha^{2}}{2}\ \mathbb{E}_{\Pint   } \left[  
 \frac{\partial^{2}}{\partial\tilde{\alpha }^{2}} 
\Hint\circ \xi_{\tilde\alpha }- 
\frac{\partial^{2}}{\partial\tilde{\alpha }^{2}}  \ln |\det D\xi_{\tilde{\alpha }}| \right]
\end{align}
for some $\tilde{\alpha }\in  [0,\alpha ]$.
In  the last equality we used \eqref{Sgamma}.
In the following, we prove a bound for the argument of the expectation in 
\eqref{S-alpha-derivative} for any $0\le \tilde\alpha\le\cneun\eta$. 
Below, we write $\alpha$ instead of $\tilde\alpha$ for simplicity. 

\paragraph{Bound on the energy: $ \frac{\partial^{2}}{\partial\alpha^{2}} 
(\Hint\circ \xi_{\alpha })$.} 
The deformation $\xi_{\alpha}$ acts only on the horizontal variables 
$\nabla t^{\bb }_{p_{n+1/2}}$ belonging
to the backbone segment connecting $\0 $ to $\ell$, hence 
for $0\leq n\leq l-1$. 
Using the decomposition \eqref{decomp HLlo} and \eqref{Hn12} 
for $\Hint$, each variable $\nabla t^{\bb }_{p_{n+1/2}}$ 
appears only inside the terms $h_e$ for edges $e\in E_{n+1/2}$, therefore  
\begin{equation}\label{entr1}
\frac{\partial^{2}}{\partial\alpha^{2}} 
(\Hint\circ \xi_{\alpha }) = 
 \sum_{n=0}^{l-1} \sum_{e\in E_{n+1/2}}
\frac{\partial^{2}}{\partial \alpha^{2} }  (h_e\circ \xi_\alpha).
\end{equation}
Applying \eqref{Hn12},  for each  $e\in E_{n+1/2}$ 
with $0\leq n\leq l-1$ we have 
\begin{align*}
 h_e\circ \xi_{\alpha  } (\vec{\bomega})
&=  \beta_e 
\left[ \cosh(\nabla t_e+ \alpha \chi_{n+1/2})-1+  
\frac{ (y_e\circ \xi_{\alpha  })^2 }{2} \right]-
\ln \frac{\beta_e}{2\pi }\, \one_{\{e\in T \}}\cr
&\qquad 
  + [\nabla t^{T}_e + \alpha \chi_{n+1/2}] \one_{\{e\in B^{c} (T) \}}
\end{align*}
where all terms from the backbone tree present in \eqref{Hn12} disappear
since  $ B^{c}(\tback)\cap E_{n+1/2}=\emptyset$ for all $n$.
Taking the second derivative in $\alpha $, we have
\begin{equation}\label{entr2}
\frac{\partial^{2}}{\partial  \alpha^{2} } h_e\circ \xi_{\alpha }=
\beta_e \ 
\Big [   \chi_{n+1/2}^{2} \cosh(\nabla t_e+ \alpha \chi_{n+1/2})\  + \  
\left[\partial_{\alpha } ( y_e\circ \xi_{\alpha  }) \right]^{2}\  + \ 
(y_e\circ \xi_{\alpha  }) [\partial_{\alpha }^{2} (y_e\circ \xi_{\alpha  })]  \Big ].
\end{equation}
First, we study the terms involving $y_e\circ\xi_\alpha$ for 
a horizontal edge $e= (i\sim j)\in E_{n+1/2}$ with $i\in V_{n}$ and 
$j\in V_{n+1}$. Using \eqref{nablat1} we can write 
\begin{align*}
& \pm y_e =  \sum_{e'\in \gamma^{ij}_{\tback }}  Y_{e'}
\end{align*}
with plus sign if $e$ is directed from $i$ to $j$ with respect to the
bookkeeping orientation and minus 
sign otherwise; the sign is irrelevant below.  
Note that $\gamma^{ri}_{\tback }\cap  \gamma^{ij}_{\tback }\subset E_{n}$ and
$ \gamma^{rj}_{\tback }\cap  \gamma^{ij}_{\tback }\subset E_{n+1/2}\cup E_{n+1}$.
The only term which changes when we apply $\xi_\alpha$ to $Y_{e'}$ is 
$\nabla t^\bb_{e''}$ where $e''=p_{n+1/2}$. Consequently, for $e'=p_{n+1/2}$, 
one has $Y_{e'}\circ\xi_\alpha=Y_{e'}$ and for all other 
$e'\in \gamma^{ij}_{\tback}$
one has $Y_{e'}\circ\xi_\alpha=Y_{e'}e^{\pm\frac12\alpha\chi_{n+1/2}}$. 
More precisely, we get  
\begin{align}
\pm y_e\circ \xi_{\alpha } & =\  \sum_{e'\in \gamma^{ij}_{\tback }} 
Y_{e'} \left[ e^{\frac12\alpha \chi_{n+1/2}}
\one_{\{ e'\in E_{n} \}} + \one_{\{e'=p_{n+1/2} \}}+
 e^{-\frac12\alpha \chi_{n+1/2}}
\one_{\{ e'\in E_{n+1} \}}
  \right],\cr
\pm \partial_{\alpha }(y_e\circ \xi_{\alpha  }) & =  \frac12\chi_{n+1/2} 
\Big [
\sum_{e'\in \gamma^{ij}_{\tback }\cap E_{n}} 
Y_{e'}\,   e^{\frac12\alpha \chi_{n+1/2}}  - 
\sum_{e'\in \gamma^{ij}_{\tback }\cap E_{n+1}} 
Y_{e'}\, e^{-\frac12\alpha \chi_{n+1/2}}
  \Big ],\cr
\pm \partial_{\alpha }^{2} (y_e\circ \xi_{\alpha  }) & =  
\frac14\chi^{2}_{n+1/2} \Big [
\sum_{e'\in \gamma^{ij}_{\tback }\cap E_{n}} 
Y_{e'}\,   e^{\frac12\alpha \chi_{n+1/2}}   + 
\sum_{e'\in \gamma^{ij}_{\tback }\cap E_{n+1}} 
Y_{e'}\, e^{-\frac12\alpha \chi_{n+1/2}}
  \Big ].
\end{align}
Let us assume the constraint $\chi_{n+1/2}\neq 0$ holds. It
ensures that  $|\nabla t^{\bb}_{e'}|<\eta$ and $|y_{e'}^{\bb }|<\eta$  
for all $e'\in E_{n}\cup E_{n+1/2}\cup E_{n+1}$. Furthermore, 
$|\gamma^{ij}_{\tback }|\leq 2|S|+1$ holds for any horizontal edge with endpoints
$i,j$. Hence, using \eqref{nablat1},
for each $e= (i\sim j)\in  E_{n+1/2}$, we have $|Y_{e'}| \leq  \eta e^{|S| \eta}$
for all $e'\in \gamma^{ij}_{\tback }$.
Thus,
\begin{align}
| \partial_\alpha (y_e\circ\xi_\alpha)| \le |S|\eta e^{\eta|S|+\frac{\alpha }{2}},\  
|\partial^{2}_\alpha (y_e\circ\xi_\alpha)| \le \frac{|S|\eta}{2} 
e^{\eta|S|+\frac{\alpha }{2}},\ 
|(y_e\circ\xi_\alpha)| \le (2|S|+1)\eta e^{\eta|S|+\frac{\alpha }{2}}.
\nonumber\end{align}
Since $|\alpha|\le\cneun\eta$ by the assumption of the theorem, we have 
\begin{align}
\left| 
\left[\partial_{\alpha } ( y_e\circ \xi_{\alpha  }) \right]^{2}\  + 
(y_e\circ \xi_{\alpha  }) [\partial_{\alpha }^{2} (y_e\circ \xi_{\alpha  })] 
\right| 
\le    3 |S|^{2} \eta^2 e^{\eta(2|S|+\cneun)}.
\end{align}
Moreover, since $|\nabla t_e|= |\nabla t_{i,j}|\leq  \eta (2|S|+1)$ by 
\eqref{repr-nabla-t}, we have
\begin{align} 
\chi_{n+1/2}^{2} \cosh(\nabla t_e+ \alpha \chi_{n+1/2})
\le \cosh( \eta (2|S|+1+\cneun)).
\end{align} 
Inserting all these bounds in \eqref{entr1}-\eqref{entr2} above, 
we have
\begin{equation}\label{entropy-est}
 \frac{\partial^{2}}{\partial\alpha^{2}} 
\left(H^{\0 \ell }\circ\xi_{\alpha }\right) \leq  
l |E_{1/2}| [\max_{e\in E_{1/2}}\beta_e] \left( \cosh [\eta(2|S|+1+\cneun)]+ 
 3|S|^2 \eta^2 e^{\eta(2|S|+\cneun)} \right)
=:l \czwoelf
\end{equation}

\paragraph{Bound on the determinant 
$\frac{\partial^{2}}{\partial{\alpha }^{2}}  \ln |\det D\xi_{\alpha}|$.}
The Jacobi matrix of the deformation 
$(\nabla t_\bb,y_\bb)\mapsto(\nabla t^\alpha,y_\bb)$ has a block
structure with the block $\partial y_\bb/\partial\nabla t_\bb=0$ 
and $\partial y_\bb/\partial y_\bb=\id$. Thus,
the Jacobi determinant for the deformation $\xi_{\alpha }$ 
is given by
\[
\det D\xi_{\alpha } = 
\det \left( \frac{\partial \nabla t^\alpha_e}{\partial \nabla t_{e'}^\bb }
\right)_{e,e'\in\tback}.
\]
Recall from \eqref{defxi1} that only the variables $\nabla t_{p_{n+1/2}}^\bb$ for 
edges $p_{n+1/2}$ on the backbone between levels $0$ and $l$ are deformed. 
We get 
\begin{align}
\label{derivative-nabla-t-alpha}
 \left( \frac{\partial \nabla t^\alpha_e}{\partial \nabla t_{e'}^\bb }
\right)_{e e'}= \delta_{e e'}+ \alpha  X_{e e'}, \quad 
\mbox{where} \quad 
\left\{ \begin{array}{lll}
X_{ee'} &= 0 & \mbox{if} \ e\not\in \gamma^{\0 \ell}_{\tback }\\
X_{p_{n+1/2}e'} &=  \frac{\partial\chi_{n+1/2}}{\partial \nabla t_{e'}^\bb }  
& \ n=0,\dotsc l-1\\
\end{array}\right.
\end{align}
where $\chi_{n+1/2}$ depends on $\nabla t_{p_{n+1/2}}^\bb$ and 
$\nabla t_{e}^\bb$  for $e\in E_{n}\cup E_{n+1} $. 
We order the variables $\nabla t^\bb$ so that
 the horizontal variables $\nabla t^\bb_{p_{n+1/2}}$, $n=0,\ldots,l-1$, 
come first and then all the others. With  this ordering and 
using  \eqref{derivative-nabla-t-alpha}, $\partial\nabla t^\alpha/\partial\nabla t^\bb$ 
becomes a triangular matrix.
The only diagonal entries that may be unequal to $1$  are $1+ \alpha  X_{p_{n+1/2}p_{n+1/2}}$,
$n=0,\ldots,l-1$.
Therefore,
\begin{align}
\label{det-D-xi-alpha}
\det D\xi_{\alpha } =\prod_{n=0}^{l-1}(1+\alpha  X_{p_{n+1/2}p_{n+1/2}}).
\end{align}
By Lemma \ref{lemma-xi-gamma-invertible}, we have the bound 
$|X_{p_{n+1/2}p_{n+1/2}}|\le 2\eta^{-1}\|\tilde\chi'\|_\infty$. 
The assumption $|\alpha|\le\cneun\eta< \eta(2\|\tilde\chi'\|_\infty)^{-1}$ 
of the theorem implies that every factor in the product \eqref{det-D-xi-alpha} is strictly positive. 
Therefore,
\[
\ln \det D\xi_\alpha 
= \sum_{n=0}^{l-1} \ln (1+\alpha  X_{p_{n+1/2}p_{n+1/2}} ).
\]
Taking the second derivative in $\alpha $ we have
\begin{equation}\label{det-est}
 0\leq  -\frac{\partial^{2}}{\partial {\alpha }^{2}}  \ln \det D\xi_{{\alpha }}
=  \sum_{n=0}^{l-1} \frac{ X_{p_{n+1/2}p_{n+1/2}}^{2} }{(1+\alpha  X_{p_{n+1/2}p_{n+1/2}})^{2}}
\leq  l\, \frac{  4\|\tilde{\chi}'\|_\infty^{2} }{\eta^{2} 
(1-\cneun 2\|\tilde{\chi}'\|_\infty  )^{2}}=:\, l\,  \cdreizehn 
\end{equation}
Inserting  \eqref{entropy-est} and \eqref{det-est} in \eqref{S-alpha-derivative} above and setting 
$\cfuenf = (\czwoelf +\cdreizehn)/2$ the result follows.
\end{proof}

\section{Local variables}\label{sect:locv}

\subsection{Local tree variables}
Our next goal is to describe each spanning tree $T$ of $\G_L$ by a 
sequence of local tree variables. 
A similar but different local representation of spanning trees was done 
in the case that the finite graph $G_0$ is a tree in \cite{rolles2005}. 

\subsubsection{Definitions and properties}

\paragraph{Preliminary definitions.}
In the following it will be more convenient to work on the infinite graph 
$\G$ instead of $\G_L$. We define the backbone tree $\tback_{\infty}$ 
as the unique spanning tree on $\G$ that coincides 
with $\tback$ on every finite piece $\G_L$.
Similarly, we extend every spanning tree $T\in\T_L$ to a spanning tree $T_\infty$
of the two-sided infinite graph $\G$ by attaching copies of $S$ to the backbone:
\begin{align}
T_\infty:=
T\cup \Big(\tback_\infty\cap\bigcup_{n\in\frac12\Z:\, n<-\ul \text{ or }n>\ol}E_n\Big). 
\end{align}
We identify the tree $T$ with its infinite extension $T_\infty$.
Let $\T_\infty:=\bigcup_L\T_L$ denote the set of 
all the possible spanning trees of $\G$ which agree far outside on both 
sides with $\tback_\infty$. 
For $T\in\T_\infty$, there is a 
unique two-sided infinite simple path in $T$ which goes 
from levels near $-\infty$ to levels near $\infty$; 
we call it the backbone in $T$ and denote the set of its edges by $B(T)$.
Since every tree $T\in\T_\infty$ is the infinite extension of a finite tree
$T_{L}\in \T_L$,    this definition coincides inside $T_{L}$ with the definition 
of $B(T)$ we introduced in Sect.\  \ref{sect:grad-tree}.

\paragraph{Translation.}
For any $m\in\Z$, we define translation operations  on vertices
$\theta^m:V\to V$, by $\theta^m(n,v)=(n+m,v)$ and on edges
$\theta^m:E\to E$, by $\theta^m e_n=e_{n+m}$ for $e_n\in E_n$, $n\in\Z$, and
$\theta^mv_{n+1/2}=v_{n+m+1/2}$ for $v_{n+1/2}\in E_{n+1/2}$.
Furthermore, we define a translation operation on trees
$\theta^m:\T_\infty\to \T_\infty$,
$T\mapsto \theta^mT$, by $\theta^mT=\{\theta^m e:\;e\in T\}$.

\paragraph{The tree structure near level 0.}
Looking at the expression for the interpolated measure \eqref{decomp HLlo}, and
more precisely, 
the contribution $h_{e} (\vec{\bomega})$ in \eqref{Hn12} 
of each edge $e\in E$, we see that 
the only information we need on the tree $T$ to compute 
$h_{e} (\vec{\bomega})$ are:
$(a)$ whether $e\in T$, $(b)$ if yes its orientation
in the tree  and $(c)$  whether $e\in B(T)$ or $B^{c} (T)$.
We consider the level $n=0$ first. To describe a tree $T\in\T_\infty$ 
locally near level $0$, we introduce an auxiliary tree which is a
simplification of $T$ with the same connectedness properties 
near level $0$. For the description of the auxiliary tree, we do not
need the {\it full} tree $T$, but only five ingredients 
$(A^\links,b^\links,F,b^\rechts,A^\rechts)$ coming
out of a {\it fixed} finite set: $F$ is the set of tree lines near level $0$,
$A^\links/A^\rechts$ encode which vertices in $V_{0}$ are connected by $T$ on
the left/right and finally $b^\links/b^\rechts$ identify the beginning/end of the
backbone segment at level $0$. The same definitions hold for the local tree structure at level $n$.
These local informations are enough to reconstruct all
properties of the tree $T$ we need for the energy at level 0. 
We will see, that if a compatibility condition is satisfied,
they are also sufficient to reconstruct the full tree. The precise definitions are given below.

\begin{definition}[Local tree structure]
\label{def:local tree vars}
Consider a tree $T\in\T_\infty$.

\noindent\textbf{The tree lines.} We denote by 
\begin{align}
F=F_{0,T}:=T\cap(E_{-1/2}\cup E_0 \cup E_{1/2})
\end{align}
the set of tree lines at levels $-1/2$, $0$, and $1/2$.
Note that $F$ is a forest. This means it is a set of lines  with no loops.\\[2mm]
\textbf{Connectedness to the left or to the right.}
We introduce a partition $A^\links=A^\links_{0,T}$ of the subset 
$\{v\in V_0: v_{-1/2}\in T\}$ of vertices in $V_0$ that have a horizontal 
edge in $T$ attached to their left. 
Vertices $u,v\in V_0$ belong to the same class in  $A^\links$ if they are 
connected by a path in $T\cap\bigcup_{m\in\frac{\Z}{2}, m < 0} E_m$, i.e.\ 
using only edges on levels $\le -1/2$. For $v\in V_0$ with $v_{-1/2}\in T$, 
its class in $A^\links$ is denoted by $[v]^\links=[v]^\links_{0,T}$ (this class
may in some cases consist of a single point).

The partition $A^\rechts=A^\rechts_{0,T}$ of the set 
$\{v\in V_0: v_{1/2}\in T\}$ is defined similarly, using 
paths in $T\cap\bigcup_{m\in\frac{\Z}{2}, m > 0} E_m$, i.e.\ only edges
on levels $\ge 1/2$. The class of any $v\in V_0$ with $v_{1/2}\in T$ 
is denoted by $[v]^\rechts=[v]_{0,T}^\rechts$. 
\\[2mm]
\textbf{The backbone.}
Finally, when traversing $B(T)$ from $-\infty$ to 
$\infty$, there is a vertex $b^\links=b^\links_{0,T}$ in $V_{0}$
traversed first
among all vertices on level $0$.
Similarly, there is a vertex $b^\rechts=b^\rechts_{0,T}$
on level $0$ traversed last.
\end{definition}
\begin{figure}
\centerline{\psfig{figure=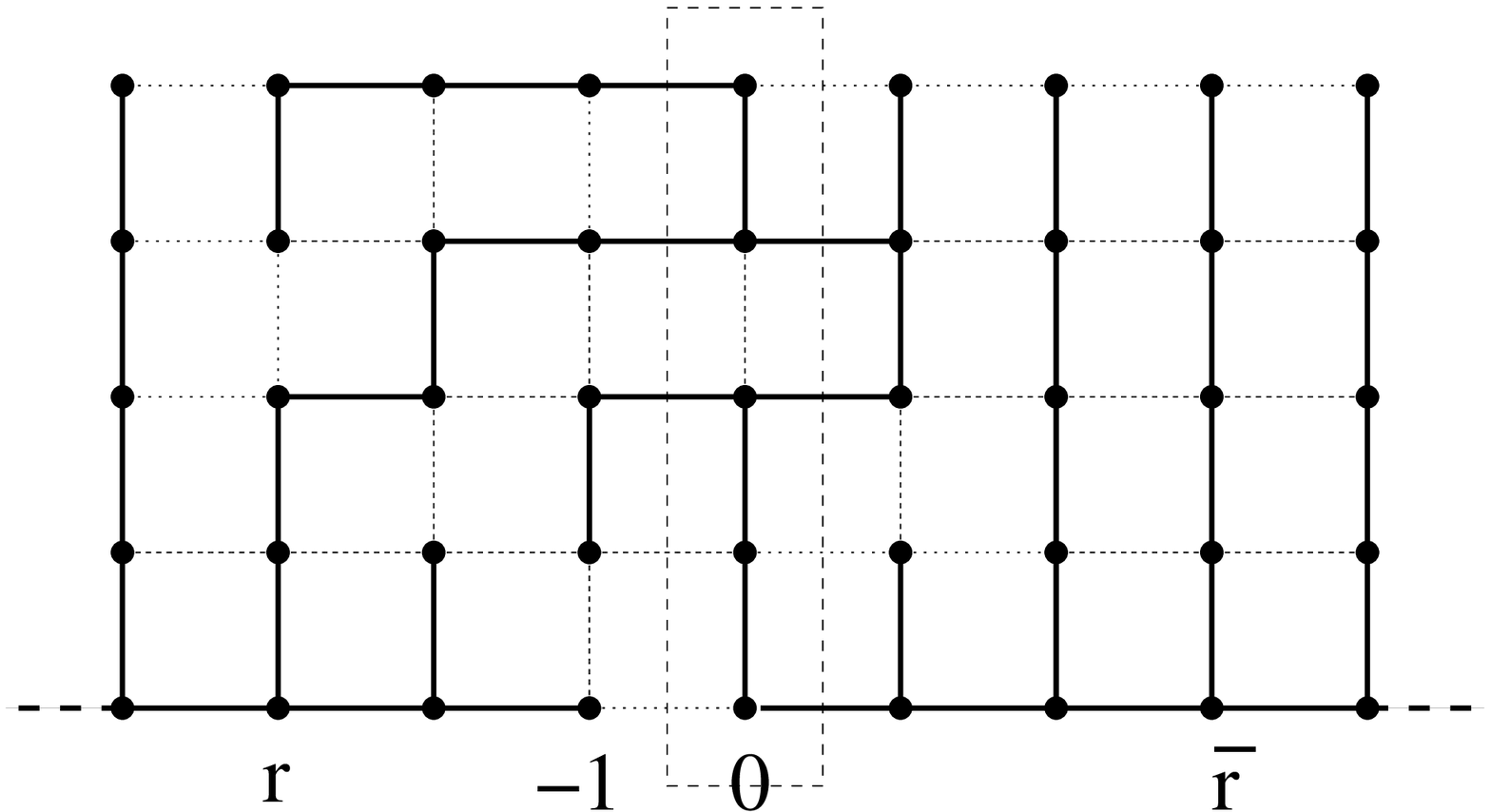,width=7cm}\hspace{0.8cm} 
\psfig{figure=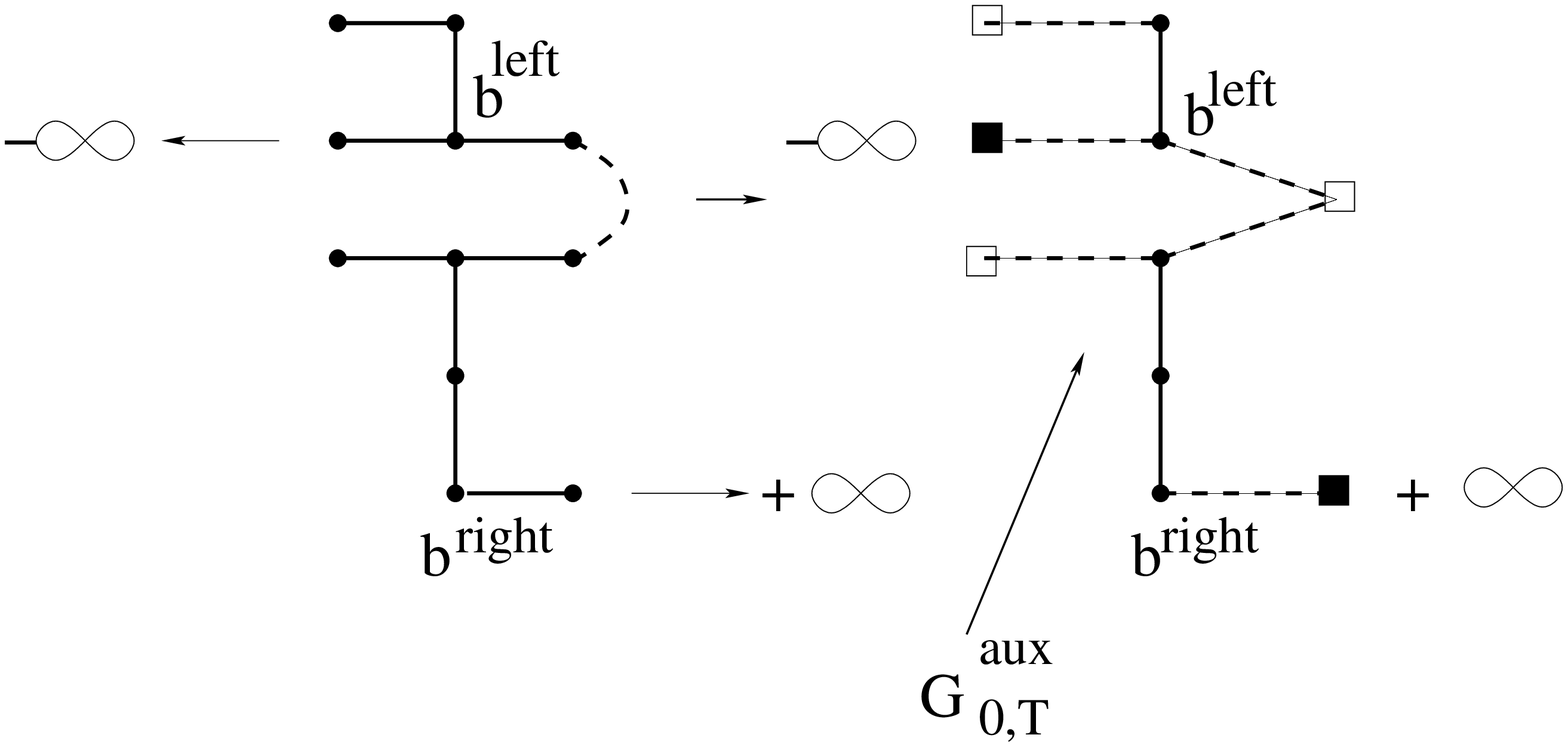 ,width=7cm}}
\caption{(a) an example of spanning tree $T\in \T_{L}$ and its extension in
 $\T_\infty$ 
with the structure at level $0$ put in evidence; (b) extracting the local structure
of $T$  at  level $0$; c) the corresponding auxiliary graph  $\aux_{0,T}$.}
\label{fig1}
\end{figure}
One may visualize the elements of $A^\links$ and $A^\rechts$ as being 
distinct auxiliary 
vertices ``to the left at level $-1$'' and ``to the right at level $1$'', 
respectively. For any vertex
$v\in V_0$ such that there is a horizontal line $v_{-1/2}$ in $T$ connecting it 
to $(-1,v)$, we draw an auxiliary line 
from $v$ to the vertex $[v]^\links$ (a square dot in Figure \ref{fig1}). 
Similarly, we draw
auxiliary lines on the right.
In this way, we get a graph $\aux=\aux_{0,T}$ as follows.  
Its set of vertices is the 
union of the set $V_{0}$ of vertices at level $0$ together
with the new vertices corresponding to the classes in 
$A^\links$ and $A^\rechts$. 
Its set of (undirected) edges consists of the forest $F\cap E_0$ and 
the auxiliary lines introduced above. By construction, 
$\aux_{0,T}$ is a tree; see Figure \ref{fig1} for an example.
We will denote the auxiliary vertex associated to 
$b^\links$ on the left (in $\aux_{0,T}$) by $-\infty:=[b^\links]^\links$. 
Similarly, $\infty:=[b^\rechts]^\rechts$.
With these notions, we can define the local tree variables. 
\begin{definition}[Local tree variables]
For any tree $T\in\T_\infty$, we define its ``local tree variable'' 
at level $0$ by 
\begin{align}
\label{def-treevar}
\tau=\tau_{0,T}:=(A^\links ,b^\links,F,b^\rechts,A^\rechts).
\end{align}
We define also its ``local tree variable'' at level $n\in\Z$ by taking the 
local tree variable at level $0$ of the shifted tree $\theta^{-n}T$: 
\begin{align}
\tau_{n,T}=(A_{n,T}^\links ,b_{n,T}^\links,F_{n,T},
b_{n,T}^\rechts,A_{n,T}^\rechts):=\tau_{0,\theta^{-n}T}.
\label{tau-n}\end{align}
All other quantities of the form $\operatorname{``something''}_{0,T}$ can 
be generalized
to $\operatorname{``something''}_{n,T}:=\operatorname{``something''}_{0,\theta^{-n}T}$,
as well.
Both $\tau_{0,T}$ and $\tau_{n,T}$ belong to the set $\treevar$
defined as follows 
\begin{align}
\treevar = \{ \tau_{0,T}:T\in\T_\infty\}. 
\end{align}
\end{definition}
Note that $\treevar$ is a {\it finite} set since there are only finitely
many choices for $A^{\links/\rechts}$, $b^{\links/\rechts}$, and $F$. 

The local tree structure near level $n$ can be completely
recovered from $\tau_{n,T}$. This is proved in the following theorem.
\begin{theorem}
\label{thm:recover information from word} 
For all $T\in\T_\infty$, $n\in\Z$, for any given, known
edge $e\in E_{-1/2}\cup E_0 \cup E_{1/2}$,
all information we need on its copy $\theta^ne$ at level $n$, that is 
$(a)$ whether it belongs to $T$, $(b)$ if the answer is yes,
which end point of $\theta^ne$ is closer to $-\infty$ in $T$, 
hence its orientation in $T$, and $(c)$
whether $\theta^ne$ belongs to the backbone $B(T)$ of $T$, 
can be recovered from knowing $\tau_{n,T}$ without knowing
$n$, $T$ or $B(T)$ explicitly.
Moreover, the following map is one-to-one.
\begin{align}
\word:\T_\infty\to\treevar^\Z, \quad
\word(T)=(\tau_{n,T})_{n\in\Z}
\end{align}
\end{theorem}
The map $\word$ above is not onto (except in the trivial case 
of the one-dimensional chain $\G=\Z$). 
To describe its range, we need to introduce a matching
condition as follows.
\begin{definition}[Matching relation for tree variables]
Let $\tau,\tau'\in\treevar$. We say that $\tau$ can be followed by $\tau'$, 
in symbols $\tau\fol\tau'$, if there is a tree $T\in\T_\infty$ such 
that $\tau_{0,T}=\tau$ and $\tau_{1,T}=\tau'$. Furthermore, for all 
$L=(-\ul,\ol)$, using the abbreviation $\tau_\bb=\tau_{0,\tback_\infty}$, 
we define 
\begin{align}
\words_L:= 
 \{ (\tau_n)_{n\in\Z}\in\treevar^\Z:\; \forall n\in\Z \; \tau_n\fol \tau_{n+1}
\text{ and } \forall n\in\Z\setminus[-\ul,\ol] \; 
\tau_n=\tau_\bb \}. 
\end{align}
\end{definition}
With this definition we are finally able to reconstruct spanning trees from
sequences of local tree variables, as follows. 
\begin{theorem}
\label{thm:bijection}
For any $L$, the function $\word$ maps $\T_L$ bijectively onto $\words_L$. 
The map $\word:\T_\infty\to\words:= \bigcup_L \words_L$ is a bijection. 
Moreover, there is $N\in\N$ such that for all $\tau,\tau'\in\treevar$,
there are $\tau_0,\ldots,\tau_N\in\treevar$ with
\begin{align}
\tau=\tau_0\fol \tau_1\fol\ldots\fol \tau_N=\tau'.
\end{align}
\end{theorem}
We will also need to use some \textbf{reflection properties of trees}.
We define a reflection operation $\refl:E\to E$, $e\mapsto e^\refl$
on edges,
by $e_n^\refl=e_{-n}$ for any vertical edge $e_n\in E_n$, $n\in\Z$, and
$v_{n+1/2}^\refl=v_{-n-1/2}$ for any horizontal edge $v_{n+1/2}\in E_{n+1/2}$.
In the same way, we define a reflection $\refl:\T_\infty\to \T_\infty$,
$T\mapsto T^\refl$,
on trees, by $T^\refl=\{e^\refl:\;e\in T\}$. Finally, for any local tree variable
$\tau=(b^\links, A^\links, F,  A^\rechts , b^\rechts )\in
\treevar$, we set 
$\tau^\refl=(b^\rechts,  A^\rechts, F^\refl,A^\links, b^\links )\in
\treevar$, where $F^\refl=\{e^\refl:\;e\in F\}$.
Note that $\refl^2=\id$ holds for these reflection operations.

\begin{lemma}\label{lemma:reflection}
The reflection operation satisfies
\begin{equation}
\tau_{0,T^{\refl}}= [\tau_{0,T}]^{\refl} \qquad \mbox{and} \qquad 
\tau \fol \tau'\  \Leftrightarrow \ (\tau')^\refl \fol \tau^\refl.
\end{equation}
\end{lemma}
The rest of this subsection is devoted to the proofs of the above statements.
However, in the next section, only the claims of the theorems
are used, but no details from the proofs. 

\subsubsection{Proofs}

\noindent
\begin{proof}[Proof of Lemma \ref{lemma:reflection}]
Since ``left'' and ``right'' are exchanged in the definition of both,
 $T^\refl$ and $\tau^\refl$, the first claim 
$\tau_{0,T^{\refl}}= [\tau_{0,T}]^{\refl}$ follows immediately. For the second
claim, let $\tau,\tau'\in\treevar$ with $\tau\fol\tau'$. 
Then, there exists 
$T\in\T_\infty$ with $\tau_{0,T}=\tau $ and $\tau_{1,T}=\tau'$. Then
the reflected tree satisfies $\tau_{0,T^\refl}=\tau_{0,T}^\refl=\tau^\refl$ 
and $\tau_{-1,T^\refl}=\tau_{1,T}^\refl=(\tau')^\refl$. 
This shows that $(\tau')^\refl\fol\tau^\refl$. 
\end{proof}

\paragraph{Auxiliary finite trees.}
Given a tree $T\in\T_\infty$ and two levels  
$m,n\in\Z$ with $m\le n$, we define an auxiliary graph
$\aux_{[m,n],T}$ as follows. Its set of vertices consists
of the union of $\bigcup_{k=m}^n V_k$ together with a copy 
$\theta^{m}A^\links_{m,T}$ of $A^\links_{m,T}$ 
formally
associated to level $m-1$ and a copy $\theta^nA^\rechts_{n,T}$ 
of $A^\rechts_{n,T}$ formally 
associated to level $n+1$.
The set of edges in $\aux_{[m,n],T}$ consists of all edges in
$T\cap \bigcup_{k\in\frac{\Z}{2}\cap[m,n]}E_k$ together with an auxiliary line
connecting every $(m,v)\in V_m$ with $v_{m-1/2}\in T$
to the copy $\theta^m[v]^\links_{m,T}$
of its class $[v]^\links_{m,T}$, and an auxiliary line
connecting every $(n,v)\in V_n$ with $v_{n+1/2}\in T$
to the copy $\theta^n[v]^\rechts_{n,T}$ of its class $[v]^\rechts_{n,T}$.
Note that $\aux_{[m,n],T}$ is again a spanning tree, 
and that  $\aux_{[n,n],T}=\theta^{n}\aux_{n,T}$ where $\aux_{n,T}=\aux_{0,\theta^{-n}T}$.
In particular 
$\aux_{[0,0],T}=\aux_{0,T}$. See Fig. \ref{fig2} for an example.

In the following, the word ``path'' means ``simple path''.
\begin{figure}
\centerline{\psfig{figure=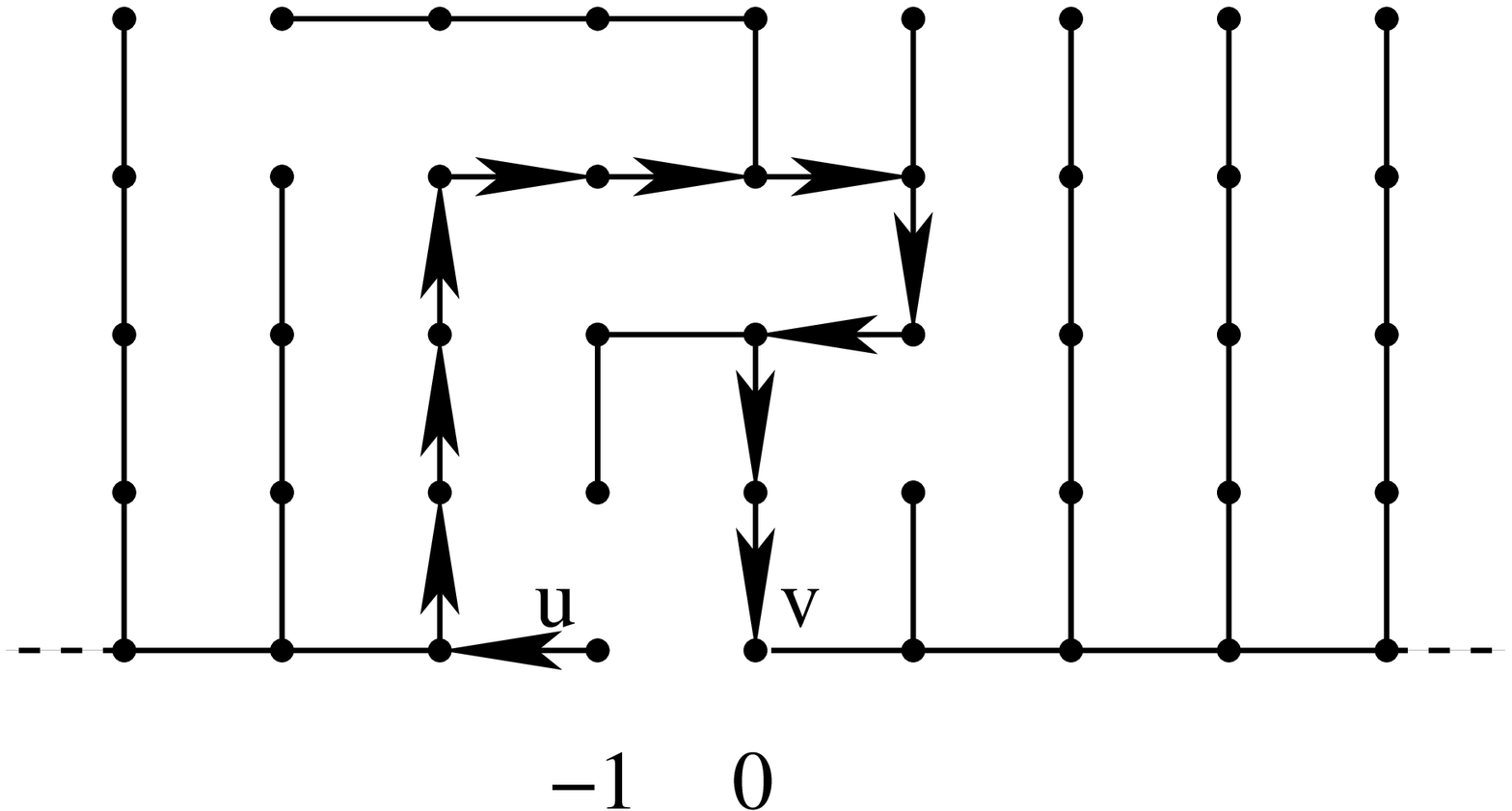 ,width=7cm}\hspace{0.8cm} 
\psfig{figure=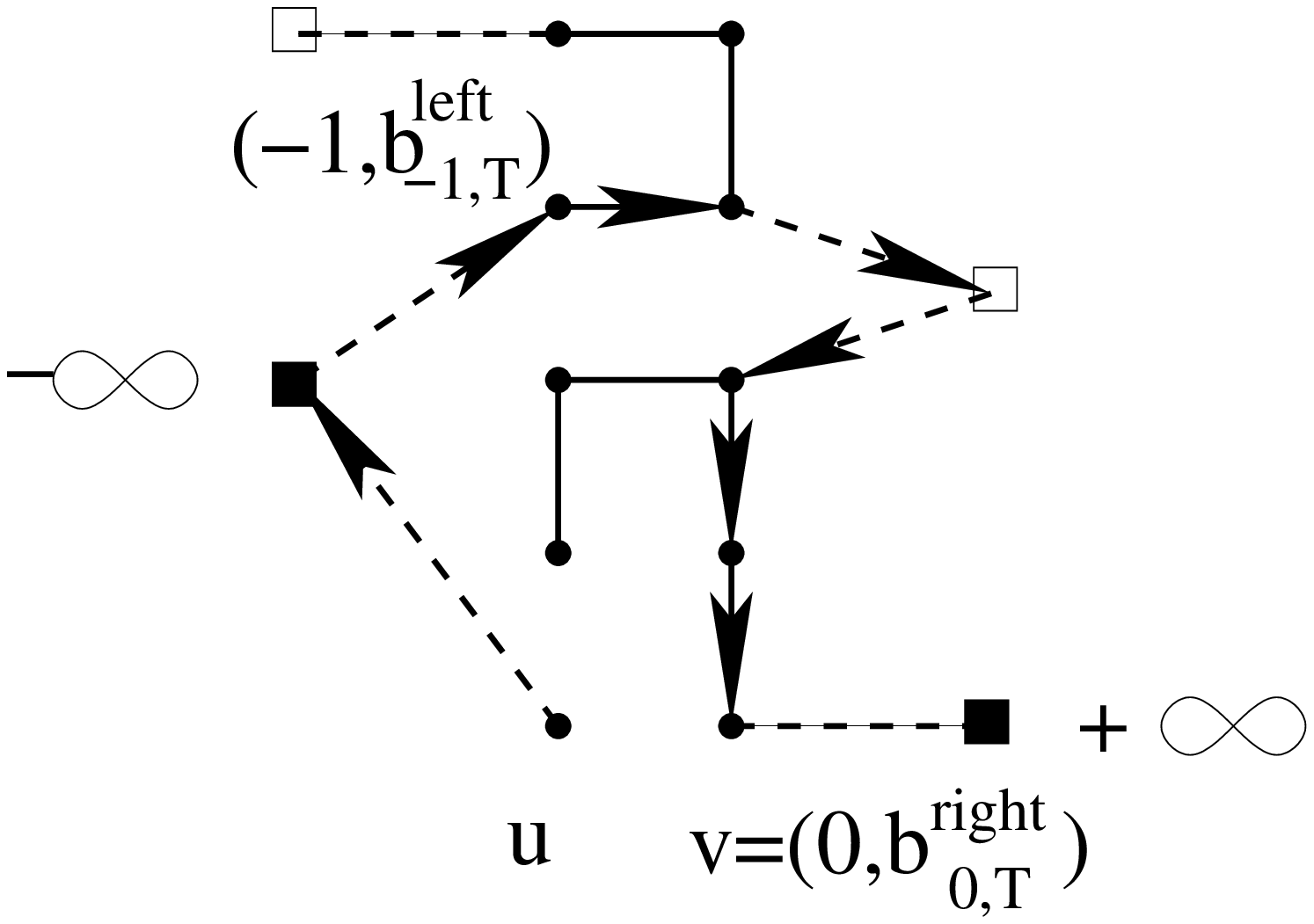 ,width=7cm}}
\caption{(a) another example of spanning tree $T\in\T_\infty$ with the 
path $\gamma_{T}^{uv}$ in evidence; b) 
the corresponding auxiliary graph  $\aux_{[-1,0],T}$,
with  $\gamma_{{\rm aux},T}^{uv}$, $ (-1, b_{-1,T}^\links)$
and  $ (0, b_{0,T}^\rechts)$ in evidence.}
\label{fig2}
\end{figure}

\begin{lemma}
\label{lemma: coinciding pieces}
Consider two vertices $u,v\in \bigcup_{k=m}^n V_k$
and the paths $\gamma^{uv}_T$ and $\gamma^{uv}_{{\rm aux},T}$ connecting them in
$T$ and in $\aux_{[m,n],T}$, respectively.
Let $e_1,\ldots,e_j$ be the edges in 
$\gamma^{uv}_T\cap\bigcup_{k\in\frac{\Z}{2}\cap[m,n]}E_k$, arranged
in the order that they are traversed when walking along $\gamma^{uv}_T$ from
$u$ to $v$. 
Let $e'_1,\ldots,e'_{j'}$ be defined similarly, using 
$\aux_{[m,n],T}$ instead of $T$.
Then $(e_1,\ldots,e_j)=(e'_1,\ldots,e'_{j'})$. 
Moreover, the edges
$e_1,\ldots,e_j$ are traversed in the same direction by
$\gamma^{uv}_T$ and $\gamma^{uv}_{{\rm aux},T}$.

Of particular interest is the case 
$u=\theta^{m} (b_{m,T}^\links)= (m,b_{m,T}^\links)$ and
$v= (n, b_{n,T}^\rechts)$. In this case, 
$\gamma^{uv}_T$ is the piece of $B(T)$ between $u$ and $v$.
\end{lemma}
\begin{proof}
The proof is straightforward by replacing pieces in $\gamma^{uv}_T$
not in $\bigcup_{k\in\frac{\Z}{2}\cap[m,n]}E_k$ by auxiliary lines. See 
Fig. \ref{fig2} for an example.
\end{proof}

\medskip\noindent\begin{proof}[Proof 
of Theorem \ref{thm:recover information from word}]
Recall from \eqref{def-treevar} and \eqref{tau-n} that 
$\tau_{n,T}$ contains the information $A_{n,T}^\links$, $b_{n,T}^\links$,
$F_{n,T}$, $b_{n,T}^\rechts$,  $A_{n,T}^\rechts$ and  is equivalent to
the auxiliary graph $\aux_{n,T}$ with the additional markings ``$\pm\infty$''. 

(a)
For any edge $e\in E_{-1/2} \cup E_{0}\cup E_{1/2}$ 
we have $\theta^{n} e \in T $ if and only
if $e\in F_{n,T}$. 

(b) Now assume $\theta^ne\in T$.
If $e=(b^\links_{n,T})_{-1/2}$ then $e\in B(\theta^{-n}T)\cap F_{n,T}$,
 by definition of $b^{\links}$, which is equivalent to 
$\theta^ne\in B(T)\cap \theta^nF_{n,T}$. 
 Then the endpoint 
$(n-1,b^\links_{n,T})$ is the one closer to $-\infty$ in $T$.
For any other edge $e\in  \theta^{-n}T\cap F_{n,T}$, 
let $u\in V_{0}$  be one endpoint.  For a horizontal
edge we have $e=u_{\pm 1/2}\in E_{\pm 1/2}$ and 
the edge corresponds to the auxiliary line connecting 
$u_{\pm 1/2}$
to $[u]^{\rechts/\links }_{n,T}$ in $\aux_{n,T}$.
Consider the path  $\gamma$ from $(0, b^\links_{n,T})$ to $(0,u)$ 
in $\aux_{n,T}$. Using Lemma \ref{lemma: coinciding pieces}, it follows 
that $(n,u)$ is the endpoint of $\theta^{n}e$ closer to 
 $-\infty$ in $T$
 if and only if the path $\gamma$ does not contain $e$ (or the corresponding
auxiliary line).

$(c)$
Using again Lemma \ref{lemma: coinciding pieces}, the edge 
$e\in E_{-1/2}\cup E_0\cup E_{1/2}$
belongs to $B(\theta^{-n}T)\cap F_{n,T}$ if and only if one of the 
following three cases holds:
$e=(b^\links_{n,T})_{-1/2}$, or $e=(b^\rechts_{n,T})_{1/2}$,
or $e$ belongs to the path from $b^\links_{n,T}$ to $b^\rechts_{n,T}$
in $\aux_{n,T}$. 

Finally, the map $\word$ is one-to-one, because
$T=\bigcup_{n\in\Z}\theta^n F_{n,T}$
holds for all $T\in\T_\infty$. This concludes the proof.
\end{proof}

When two different trees coincide somewhere, their corresponding
local tree variables  satisfy the following matching properties.
\begin{lemma}
\label{lemma: shift}
Let $T,T'\in\T_\infty$ be two trees and $m\in\Z$.
\begin{enumerate}
\item
\label{item: shift 1}
If $F_{m,T}=F_{m,T'}$ and $A_{m,T}^\links=A_{m,T'}^\links$ hold, then
$A_{m+1,T}^\links=A_{m+1,T'}^\links$ holds as well.
Similarly, the assumptions $F_{m+1,T}=F_{m+1,T'}$ and 
$A_{m+1,T}^\rechts=A_{m+1,T'}^\rechts$ imply
$A_{m,T}^\rechts=A_{m,T'}^\rechts$.
\item
\label{item: shift 2}
Assume that  
$\aux_{[m,m+1],T}=\aux_{[m,m+1],T'}$,
$b_{m,T}^\links=b_{m,T'}^\links$, and
$b_{m+1,T}^\rechts=b_{m+1,T'}^\rechts$ hold.
Then 
$b_{m+1,T}^\links=b_{m+1,T'}^\links$ and
$b_{m,T}^\rechts=b_{m,T'}^\rechts$
hold also.
\item
\label{item: shift 3}
Assume that $F_{n,T}=F_{n,T'}$ for all $n\le m$ and
$\tau_{m,T}=\tau_{m,T'}$ hold.
Then $\tau_{n,T}=\tau_{n,T'}$ holds for all $n\le m$.
The same holds when ``$n\le m$'' is replaced by ``$n\ge m$''
in the assumption and in the claim.
\end{enumerate}
\end{lemma}
\begin{proof}
{\it Part \ref{item: shift 1}:}
Consider a path in $T$ connecting two vertices $u$ and $v$ on level $m+1$
using only edges on levels $\le m+1/2$. Since 
$A_{m,T}^\links=A_{m,T'}^\links$, any excursion in the
path from one vertex on level $m$ to another vertex on level $m$
using only edges on levels $\le m-1/2$ can be replaced by an 
excursion in $T'$ between the same vertices, also using only 
edges on levels $\le m-1/2$. In this way, we get a path in $T'$
from $u$ to $v$ which uses also only edges on levels $\le m+1/2$.
The same holds when $T$ and $T'$ are exchanged. The second statement
follows by the same argument, exchanging ``left'' and ``right''.
As an example, compare the trees in Fig. \ref{fig1} and \ref{fig2}
on levels $-1$ and $0$. 

\noindent {\it Part \ref{item: shift 2}:}
This follows directly from Lemma  \ref{lemma: coinciding pieces},
applied to $u=\theta^{m}b_{m,T}^\links$ and
$v=\theta^{m+1}b_{m+1,T}^\rechts$.

\noindent {\it Part \ref{item: shift 3}:}
The assumption $F_{n,T}=F_{n,T'}$ for $n\le m$
implies that $T$ and $T'$ coincide on all levels $\le m+1/2$.
It follows that $A^\links_{n,T}=A^\links_{n,T'}$ 
and $b^\links_{n,T}=b^\links_{n,T'}$ hold for $n\le m$.
From this, the first claim in 3.\  
follows by induction over $n$, starting with $n=m$ and
using parts 1.\ and 2.\ of the lemma in the induction step.
The second claim in 3.\ follows similarly, exchanging the roles of 
``left'' and ``right''.
\end{proof}

The next lemma gives criteria to paste pieces of different trees together.
\begin{lemma}[Glueing trees]
\label{lemma: glueing trees}
Let $m\in\Z$ and let $T_\links,T_\rechts\in\T_\infty$ be spanning trees with 
$\tau_{m,T_\links}\fol\tau_{m+1,T_\rechts}$. 
Then there is a unique tree $T\in\T_\infty$, called $\glue_m(T_\links,T_\rechts)$,
with $\tau_{n,T}=\tau_{n,T_\links}$ for $n\le m$
and $\tau_{n,T}=\tau_{n,T_\rechts}$ for $n\ge m+1$.
\end{lemma}
\begin{proof}
{\it Uniqueness} follows from the fact that $F_{n,T}$ is a component of 
$\tau_{n,T}$, and thus
\begin{align}
\label{align-uniqueness}
T=\bigcup_{n\in\Z} \theta^nF_{n,T}
= \bigcup_{n\le m} \theta^nF_{n,T_\links} \cup \bigcup_{n\ge m+1} 
\theta^nF_{n,T_\rechts}=:\glue_m(T_\links,T_\rechts).
\end{align}
{\it Existence.} We take as definition the unique possible choice
$T=\glue_m(T_\links,T_\rechts)$ from above. 
We will prove below this is a spanning tree and 
$\tau_{n,T}=\tau_{n,T_\links}$ for all $n\le m$ and 
$\tau_{n,T}=\tau_{n,T_\rechts}$ for all $n\ge m+1$.

By definition, $T$ agrees with $T_\links$ on all levels $\le m$ and with 
$T_\rechts$ on all levels $\ge m+1$, i.e.\ 
$T\cap E_k=T_\links\cap E_k$ for all $k\in\frac{\Z}{2}$ with 
$k\le m$ and $T\cap E_k=T_\rechts\cap E_k$ for all $k\in\frac{\Z}{2}$ with 
$k\ge m+1$. Furthermore, from the definition of $T$, 
$T\cap E_{m+1/2}=(T_\links\cup T_\rechts)\cap E_{m+1/2}$. 
Since $\tau_{m,T_\links}\fol\tau_{m+1,T_\rechts}$, there is
a tree $T'\in\T_\infty$ with $\tau_{m,T_\links}=\tau_{m,T'}$ and 
$\tau_{m+1,T_\rechts}=\tau_{m+1,T'}$. The tree $T'$ plays an important role in 
the remainder of the proof. 
We have $T_\links\cap E_{m+1/2}= T'\cap E_{m+1/2}
= T_\rechts\cap E_{m+1/2}$, hence
\begin{align}
\label{T-m-plus-halb}
T'\cap E_{m+1/2}= T_\links\cap E_{m+1/2}=T_\rechts\cap E_{m+1/2}=T\cap E_{m+1/2}. 
\end{align}
In the same way we have
$A_{m,T_{\links }}^{\links }= A_{m,T'}^{\links }$ and
$A_{m+1,T_{ \rechts}}^{\rechts }= A_{m+1,T'}^{\rechts }$. 

{\it $T$ is acyclic.} We  prove this by contradiction. Assume that $T$ contains
a cycle $C$. If all edges in $C$ are on levels $\le m$ (resp.  $\ge m+1$), 
then it is already a cycle in $T_\links$ (resp.  $T_\rechts$), a contradiction. 
Now suppose that $C$ crosses level $m+1/2$ a non-zero even number of times.
This path consists of alternating pieces in $T_\links$ 
made of edges at levels $\le m$
and pieces in $T_\rechts$ made of edges at levels $\ge m+1$, 
connected by horizontal lines at level $m+1/2$.
Any piece in $T_\links$ together with the two horizontal lines
at level $m+1/2$ attached to it can be replaced by  
a path in $T'$, made of edges on levels $\le m$ plus the same two 
horizontal lines. This is true because $A^\links_{m+1,T_\links}=A^\links_{m+1,T'}$
by Part \ref{item: shift 1} of Lemma \ref{lemma: shift} applied to $T_\links$ 
and $T'$.
Similarly any piece in $T_\rechts$ together with the two horizontal lines
at level $m+1/2$ attached to it can be replaced by  
a path in $T'$, made of edges on levels $\ge m+1$ plus the same two 
horizontal lines,
since $A^\rechts_{m,T_\rechts}=A^\rechts_{m,T'}$.
Joining these pieces in $T'$, we obtain a cycle in $T'$, which is impossible.

{\it $T$ connects any two vertices in $\Z\times G_0$ to each other.}
First, from the argument just described above we know that any two horizontal edges 
in $T$ on level $m+1/2$ are connected by a path in $T$ if and 
only if they are connected in $T'$. Since $T'$ is 
a spanning tree, this implies that any two horizontal edges in $T$ 
on level $m+1/2$ are connected in $T$. 
Second, 
for any vertex $(n,u)$ on any level $n\le m$ there exists at least one
horizontal edge on level $m+1/2$ connected to it in $T_\links$ 
by a path using only edges on levels 
$\le m+1/2$. This path is also a path in $T$. Third, similarly, 
for any vertex $(n,v)$ on any level $n\ge m+1$ there is a path 
in $T_\rechts$ and hence in $T$ to some horizontal edge on level $m+1/2$. 
Combining these three arguments the claim follows. 

We have shown that $T$ is a spanning tree, therefore $A_{n,T}^\links,A_{n,T}^\rechts$ 
are well defined for all $n\in \mathbb{Z}$. Obviously  
$A_{m,T}^\links=A_{m,T_\links}^\links=A_{m,T'}^\links$
and  $A_{m+1,T}^\rechts=A_{m+1,T_\rechts}^\rechts=A_{m+1,T'}^\rechts$.
Similarly $b_{m,T}^\links=b_{m,T_\links}^\links=b_{m,T'}^\links$, and 
$b_{m+1,T}^\rechts=b_{m+1,T_\rechts}^\rechts=b_{m+1,T'}^\rechts$. 
 Then  $\aux_{[m,m+1],T}=\aux_{[m,m+1],T'}$ and  by a straightforward application 
of parts \ref{item: shift 1} and \ref{item: shift 2}
of Lemma \ref{lemma: shift} we get $\tau_{m,T}=\tau_{m,T_\links}$ and
$\tau_{m+1,T}=\tau_{m+1,T_\rechts}$.
Finally, the claims $\tau_{n,T}=\tau_{n,T_\links}$ 
for $n\le m$
and $\tau_{n,T}=\tau_{n,T_\rechts}$ for $n\ge m+1$
follow from part \ref{item: shift 3} of Lemma \ref{lemma: shift}.
\end{proof}

Note that  if a spanning tree $T\in\T_\infty$ satisfies $F_{0,T}=F_{0,\tback_\infty}$, then 
$\tau_{0,T}=\tau_{0,\tback_\infty}=\tau_\bb$. 
In other words, $\tau_\bb$ is the only tree variable representing
a tree locally at $0$ which looks like $\tback_\infty$ near $0$.
This comes from the fact that $\tback_\infty$ has only one horizontal
line on level $1/2$ and only one on level $-1/2$.

\smallskip\noindent
\begin{proof}[Proof of Theorem \ref{thm:bijection}]
The map $\word$ is one-to-one by Theorem 
\ref{thm:recover information from word}.
To prove that the map $\word$ is onto, let $\tau=(\tau_n)_{n\in\Z}\in\words$. 
We prove by induction
that for any $m\in\Z$, there is a tree $T^m_\links\in\T_\infty$ 
such that $\tau_{n,T^m_\links}=\tau_n$ for all $n\in\Z$ with $n\le m$.
First, we see that this claim is true for all $m$ sufficiently close to 
$-\infty$. 
Indeed, for $m$ so small that $\tau_n=\tau_\bb$ for all $n\le m$,
we can just take $T^m_\links=\tback_\infty$.
For the induction step, assume that the claim holds for a given $m$.
Since $\tau\in\words$, it follows $\tau_m\fol \tau_{m+1}$. Hence, 
there is a spanning tree
$T_\rechts^m$ with $\tau_{m,T_\rechts^m}=\tau_m$ and 
$\tau_{m+1,T_\rechts^m}=\tau_{m+1}$. 
Taking $T^{m+1}_\links:=\glue_m(T^m_\links,T^m_\rechts)$  
from Lemma \ref{lemma: glueing trees}, we obtain 
$\tau_{n,T^{m+1}_\links}=\tau_{n,T^m_\links}
=\tau_n$ for all $n\in\Z$ with $n\le m$,
and $\tau_{m+1,T^{m+1}_\links}=\tau_{m+1,T_\rechts^m}=\tau_{m+1}$.
This finishes the inductive proof.
Now take $m\in\N$ so large that $\tau_n=\tau_\bb$ holds for all
$n\ge m$. Using $\tau_\bb\fol\tau_\bb$, we can take 
$T:=\glue_m(T_\links^m,\tback_\infty)$. Then $\word(T)=\tau$.

Next, we prove that for any $L=(-\ul,\ol)$, the function
$\word$ maps $\T_L$ onto $\words_L$.
For $T\in\T_L$, note that $\tau_{n,T}=\tau_\bb$ for all 
$n\in\Z\setminus [-\ul,\ol]$. 
Consequently, $\word(\T_L)\subseteq\words_L$. 
To prove that $\word(\T_L)\supseteq\words_L$, take $\tau\in\words_L$. 
By the above, there exists $T\in\T_\infty$ with $\word(T)=\tau $. 
Since $T=\bigcup_{n\in\Z} \theta^n F_{n,T}$, 
the tree $T$ agrees with $\tback_\infty$
on $(\Z\setminus [-\ul,\ol])\times G_0$. Consequently, $T\in\T_L$. 

It remains to prove that
there exists $N\in\N$ such that for any tree variables $\tau,\tau'\in\treevar$, 
there is a tree $T\in\T_\infty$ with $\tau_{0,T}=\tau$ and $\tau_{N,T}=\tau'$.

For any $\tau\in\treevar$, choose a tree $T_\tau$
with $\tau=\tau_{0,T_\tau}$. Since the set $\treevar$ is finite,
there is $M\in\N$ such that for all $n\in \Z$ with $|n|\ge M$ and 
all $\tau\in\treevar$, one has
$\tau_{n,T_\tau}= \tau_\bb$.
Take $N=2M+1$. 
Given $\tau,\tau'\in\treevar$,
the tree $T_\tau$ equals $\tback_\infty$ on all levels $\ge M$,
while $\theta^N T_{\tau'}$ equals $\tback_\infty$ on all levels $\le M+1$.
Gluing $T_\tau$ and $\theta^N T_{\tau'}$ together at level $M+1/2$,
we obtain a tree $T=\glue_M(T_\tau,\theta^N T_{\tau'})$ which 
satisfies the claim.
\end{proof}

\subsection{Joining gradient and local tree variables}

Recall that $\Bomega_L=\Omega_L\times\T_L$ 
denotes the set of all possible values of $\vec{\bomega}=
(\nabla t_\bb, y_\bb, T)$. 
In the following we identify $\vec{\bomega}\in\Bomega_L$ with 
the set of local gradient and tree variables
\begin{equation}
\vec{\bomega} \equiv\ 
((\bomega_{n} )_{n\in \Z \cap [-\ul,\ol]},
(\omega_{n+1/2} )_{n+1/2\in (\Z+1/2) \cap [-\ul,\ol]})
\end{equation}
where for  $n\in \Z \cap [-\ul,\ol] $
\begin{equation}
\bomega_n=(\omega_n,\tau_n)=
((\nabla t^\bb_e)_{e\in S_{n}},(y^\bb_{e})_{e\in S_{n}},\tau_n (T))
\in\Bomega_\vertical:=\R^S\times\R^S\times\treevar,
\end{equation}
and for  $n+1/2\in (\Z+{1}/{2}) \cap [-\ul,\ol] $
\begin{equation}
 \omega_{n+1/2} =(\nabla t^\bb_{p_{n+1/2}},y^\bb_{p_{n+1/2}})\in\Omega_\hor. 
\end{equation}
The set $\Bomega_\vertical$ is the domain of definition for 
the gradient variables associated to
vertical edges in $S$ plus the local tree variables.
Using \eqref{repr-nabla-t} and \eqref{nablat1}, we view $\nabla t_e$ and $y_e$ for any $e\in E_L$ 
in the following as functions of $\vec{\bomega}$. 
By Theorem \ref{thm:bijection}, 
there is a bijection between the set of spanning trees $\T_L$ and 
the set $\words_L$ consisting of words of local
tree variables  $(\tau_{n})_{n=-\ul,\dots ,\ol}$ with suitable matching
conditions. Thus, the set $\Bomega_L$ is identified
with the subset $\Bomega_L\subseteq\hat\Bomega_L$ of the set 
\begin{align}
\hat\Bomega_L:=
\Bomega_\vertical^{\Z\cap[-\ul,\ol]}\times(\Omega_\hor)^{(\Z+1/2)\cap[-\ul,\ol]}
\end{align}
consisting of all $\vec\bomega$ with
\begin{align}
\label{compatibility-tau}
\tau_\bb\fol\tau_{-\ul}\fol\tau_{-\ul+1}\fol \ldots \fol\tau_\ol\fol\tau_\bb. 
\end{align}
With these definitions we can reorganize the interpolated measure
in order to set up a transfer operator approach.
Recall the definition of $h_e$ from \eqref{Hn12}.
Using the results of Lemma \ref{lemma3.1}, Remark \ref{rem-path-gamma} 
and Theorem \ref{thm:recover information from word}, 
the value $h_e(\vec{\bomega})$ for any edge $e$ 
(vertical or horizontal) can be written in terms of
local variables. 
More precisely,
for a vertical edge $e_n$
on an integer level
$n\in\Z\cap[-\ul,\ol]$, the value $h_{e_n}(\vec\bomega)$
depends only on $e\in E_0$ and $\bomega_n$, but not {\it explicitly} on
$n$ or any other component of $\vec\bomega$ (recall that  
$\beta_{e_{n}}=\beta_{e_{0}}$ for all $n\in \mathbb{Z}$). 
Thus we can write for $e\in E_0$
\[
h_{e_n}(\vec{\bomega})=h_e^\vertical(\bomega_n)
\]
for some function $h_e^\vertical:\Bomega_\vertical\to\R$.
Similarly, 
for a horizontal edge $v_{n+1/2}$ on a half-integer level
$n+1/2\in\Z\cap(-\ul,\ol)$, the value $h_{v_{n+1/2}}(\vec\bomega)$
depends only on $v\in V_0$, $\bomega_n$, $\omega_{n+1/2}$ and $\bomega_{n+1}$.
Thus we write in this case 
\[
h_{v_{n+1/2}}(\vec{\bomega})=h_v^\hor(\bomega_n,\omega_{n+1/2},\bomega_{n+1}).
\]
We set $\Bomega_\mitte:=\Bomega_\vertical\times\Omega_\hor\times\Bomega_\vertical$.
For arguments $(\bomega,\omega_\hor,\bomega')\in\Bomega_\mitte$
that cannot be written in the form $(\bomega_n,\omega_{n+1/2},\bomega_{n+1})$,
we set 
\[
h_v^\hor(\bomega,\omega_\hor,\bomega')=+\infty.
\]
Note that this is precisely the
case when the tree variable $\tau$ in $\bomega$ and the tree variable 
$\tau'$ in $\bomega'$ do not fulfill $\tau\fol\tau'$.
Using this extension, we get a well-defined function
\[
h_v^\hor:\Bomega_\mitte\to\R\cup\{\infty\}
\]
for any $v\in V_0$.
We define $H_\vertical:\Bomega_\vertical\to\R$
and $H_\hor:\Bomega_\mitte\to\R\cup\{\infty\}$ by
\[
H_{\vertical} = \sum_{e\in E_0}  h_e^\vertical,\quad 
H_\hor=
 \sum_{v\in V_0}  h_v^\hor.
\]
With the above abbreviations, we can write for $\vec\bomega\in\hat\Bomega_L$
the interpolated Hamiltonian defined in \eqref{decomp HLlo} as 
\begin{align}
H_L^{\0\ell}(\vec{\bomega})= & \sum_{n\in \Z\cap [-\ul,\ol]} H_\vertical(\bomega_n)
+ \sum_{n+1/2\in (\Z+1/2)\cap [-\ul,\ol]} H_\hor(\bomega_n,\omega_{n+1/2},\bomega_{n+1})
\cr
& - \frac12 \sum_{n=-\ul}^{-1}\nabla t^{\bb }_{p_{n+1/2}} 
+ \frac12 \sum_{n=l}^{\ol-1}\nabla t^{\bb }_{p_{n+1/2}},
\end{align}
where $H_L^{\0\ell}(\vec{\bomega})=\infty$ holds if and only if 
$\vec\bomega\notin\Bomega_L$.
Furthermore, we define 
$H_\mitte,H_\mitte^\pm:\Bomega_\mitte\to\R\cup\{\infty\}$ as follows. 
For $(\bomega,\omega_\hor,\bomega')\in\Bomega_\mitte$ 
with $\omega_\hor=(\nabla t_\hor,y_\hor)$ we set
\begin{align}
\label{def Hmitte}
H_{\mitte} (\bomega,\omega_\hor,\bomega') 
&= \frac12 H_\vertical(\bomega)+ H_\hor(\bomega,\omega_\hor,\bomega')
+\frac12 H_\vertical(\bomega'),\\
\label{def Hmittepm}
H_{\mitte}^\pm (\bomega,\omega_\hor,\bomega')
&=H_{\mitte} (\bomega,\omega_\hor,\bomega')
\pm \frac12 \nabla t_\hor. 
\end{align}
Finally, we set for $\bomega=(\omega,\tau)\in\Bomega_\vertical$, 
\begin{align}
\label{def Hlinksrechts}
H_\links(\bomega):= \frac12 H_\vertical(\bomega) + \infty 
\one_{\{\tau_\bb\not\fol\tau\}}, \quad
H_\rechts(\bomega):=  \frac12 H_\vertical(\bomega) + \infty 
\one_{\{\tau\not\fol\tau_\bb\}}. 
\end{align}
With these definitions we have the following result.
\begin{lemma}
\label{representation-H-interpolated}
The interpolated Hamiltonian 
$H_L^{\0\ell}:\hat\Bomega_L\to\R\cup\{\infty\}$ can be written as
\begin{align}
&H_L^{\0\ell}(\vec{\bomega})=  H_\links(\bomega_{-\ul}) + 
\sum_{n=-\ul}^{-1} H_\mitte^-(\bomega_{n},\omega_{n+1/2},\bomega_{n+1}) 
\\ 
&+ \sum_{n=0}^{l-1} H_\mitte(\bomega_{n},\omega_{n+1/2},\bomega_{n+1}) 
 + \sum_{n=l}^{\ol-1} H_\mitte^+(\bomega_{n},\omega_{n+1/2},\bomega_{n+1}) 
+ H_\rechts(\bomega_\ol).
\nonumber
\end{align}
It takes finite values precisely on $\Bomega_L$, represented by the constraint 
(\ref{compatibility-tau}). 
\end{lemma}
\begin{proof}
This is an immediate consequence of the definitions above.
\end{proof}
\paragraph{Reflection Symmetry.}
For later use, we define 
\begin{align}
\label{bomega-refl}
\bomega^\refl=(\omega,\tau^\refl)\in\Bomega_\vertical \quad\text{ for }\quad
\bomega=(\omega,\tau)\in\Bomega_\vertical.
\end{align}
Note that the reflection operation changes the orientation in $T$ for
the edges along the backbone $B (T)$, but not for the ones along $B^{c} (T)$.
Moreover $\beta_{e^{\refl }}=\beta_{e}$ for all $e\in E$. Then 
for $e\in E_0$, $v\in V_0$, 
$\bomega,\bomega'\in\Bomega_\vertical$
and $\omega_\hor\in\Omega_\hor$ we see from \eqref{Hn12}
\begin{align}
\label{symmetry-h}
h_e^\vertical(\bomega^\refl)=h_e^\vertical(\bomega),\quad
h_v^\hor(\bomega'^\refl,-\omega_\hor,\bomega^\refl)=
h_v^\hor(\bomega,\omega_\hor,\bomega') .
\end{align}
In particular, by Lemma \ref{lemma:reflection} we have
\begin{align}
h_v^\hor(\bomega,\omega_\hor,\bomega')< \infty \ \Leftrightarrow  \ 
\tau \fol \tau'   \ \Leftrightarrow  \  {\tau'}^{\refl}  \fol
 \tau^{\refl}\ \Leftrightarrow  \
h_v^\hor({\bomega'}^\refl,-\omega_\hor,\bomega^\refl)< \infty.
\end{align}
The symmetry properties (\ref{symmetry-h}) imply 
for $(\bomega,\omega_\hor,\bomega')\in\Bomega_\mitte$
\begin{align}
\label{symmetry-Hmitte}
H_\mitte(\bomega'^\refl,-\omega_\hor,\bomega^\refl)=
H_\mitte(\bomega,\omega_\hor,\bomega').
\end{align}

\section{The energy contribution}\label{sect:energy}

\begin{theorem}
\label{thm: energy contribution}
Take a fixed $G_0$ and $\vec\beta$. For any $\alpha\in\R$, 
the energy contribution \eqref{Egamma} satisfies
\begin{equation}\label{energyb}
\Egamma=\
\alpha\, l\, \cvier\ +\ \cdrei(L,l,\alpha)\  \le\  \alpha \, l\,  \cvier\ +\ 
\cdrei^{\max}(\alpha )
\end{equation}
where $\cvier>0$ and  $\cdrei (L,l,\alpha )\in\R$ are constants depending 
also on $G_0$ and $\vec\beta$, and 
\[
\cdrei^{\max}(\alpha):= 
\sup_{L,l}  |\cdrei(L,l,\alpha )| < \infty.
\]
\end{theorem}
The rest of the section is devoted to the proof of this result.
In Sect.\ \ref{sect:locv} we introduced  local tree variables,
replacing the global variable $T$. With these new variables
we set up a transfer operator method in Sect.\ \ref{sect:transf} below.
Finally  Sect.\ \ref{sect:energyproof} contains the proof of the theorem.

\subsection{Setting up the transfer operator}
\label{sect:transf}

We endow $\Bomega_\vertical$ with the reference measure 
$d\bomega = \prod_{e\in S} d\nabla t_e^\bb\, dy_e^\bb\, d\tau$, where 
$d\nabla t_e^\bb$ and $dy_e^\bb$ denote the Lebesgue measure on $\R$ and 
$d\tau$ denotes the counting measure on $\treevar$. The scalar product on 
$L^2(\Bomega_\vertical,d\bomega)$ is defined by 
\[
 \left \langle F,G \right \rangle:= 
\int_{\Bomega_\vertical} \overline{F (\bomega )} G (\bomega ) d\bomega.
\]
However, here we are using mostly real functions.

\begin{definition}
We define the integral kernels 
$k, k^\pm,\tilde{k}_{\alpha }: \Bomega_\vertical\times \Bomega_\vertical \to [0,\infty)$
by
\begin{align}
& k (\bomega ,\bomega') =  \int_{\Omega_\hor} 
e^{-H_{\mitte} (\bomega,\omega_\hor,\bomega')} d\omega_\hor, \quad  
k^\pm (\bomega ,\bomega') =  \int_{\Omega_\hor} 
e^{-H_{\mitte}^\pm (\bomega,\omega_\hor,\bomega')} d\omega_\hor,\\
&\mbox{and} \quad  \tilde{k}_{\alpha }(\bomega ,\bomega')=  \int_{\Omega_\hor} 
\left[ \nabla t_\hor+ \alpha \chi(\omega,\omega_\hor,\omega')\right]
e^{-H_{\mitte} (\bomega,\omega_\hor,\bomega')} d\omega_\hor,\nonumber
\end{align}
where $\omega_\hor=(\nabla t_\hor,y_\hor)$, $\bomega=(\omega,\tau)$,
$\bomega'=(\omega',\tau')$, the function $\chi$ is given by (\ref{chidef}) 
and $\alpha\in\R$. We also define two functions
$\Psi_\links, \Psi_\rechts : \Bomega_\vertical\to [0,\infty)$ by
\begin{equation}
\Psi_\links(\bomega)=e^{-H_\links(\bomega)} \qquad  \mbox{and} \qquad    
 \Psi_\rechts(\bomega)=e^{-H_\rechts(\bomega)}.
\end{equation}
\end{definition}

\begin{lemma}\label{lemmaHS}
$\Psi_\links$ and $\Psi_\rechts$ belong to 
$L^2(\Bomega_\vertical,d\bomega)\setminus\{0\}$. 
The integral kernels $k$, $k^\pm$, and $\tilde{k}_\alpha$ belong all to 
$L^2(\Bomega_\vertical\times\Bomega_\vertical,d\bomega\, d\bomega')$. 
\end{lemma}
\begin{proof}
Consider any edge $e$ in $\G_L$ and any 
$\vec\bomega=(\nabla t_\bb,y_\bb,T)\equiv 
((\omega_{n},\tau_n)_{n\in \Z \cap [-\ul,\ol]}$,
$(\omega_{n+1/2} )_{n+1/2\in (\Z+1/2) \cap [-\ul,\ol]})$.
We bound the contribution $h_e(\vec\bomega)$
to the Hamiltonian
from \eqref{Hn12} as follows from below:
\begin{align}
h_e(\vec\bomega)= &
\beta_e\left[ \cosh \nabla t_e-1 +  \frac{y_{e}^2 }{2}  \right]
+f_{e,T}(\nabla t_\bb)-\log\frac{\beta_e}{2\pi} \one_{\{e\in T \}}
\cr
\ge & 
\frac{\beta_e}{2}\left[ (\nabla t_e)^2 +  y_{e}^2 \right]
+f_{e,T}(\nabla t_\bb)-\log\frac{\beta_e}{2\pi}\one_{\{e\in T \}}
\label{lower-bound-h-e}
\end{align}
with the linear function $f_{e,T}:\nabla t_\bb\mapsto 
\nabla t_e^{T} \one_{\{e\in B^{c} (T) \}} - \nabla t_e^{\bb } 
\one_{\{e\in B^{c}(\tback) \}}/2$. Given $e$ on an integer level $n$,
note that $f_{e,T}(\nabla t_\bb)$ depends only on $\tau_{n,T}$ and linearly on the
$\nabla t_\bb$-components in $\omega_n$.
Similarly, given $e$ on level $n+1/2$,
the value $f_{e,T}(\nabla t_\bb)$ depends
only on $\tau_{n,T}$ and linearly
on the $\nabla t_\bb$-components in $(\omega_n, \omega_{n+1/2},\omega_{n+1})$.
Summing over edges and dropping the terms $(\nabla t_e)^2 +  y_{e}^2$
in \eqref{lower-bound-h-e}
for edges $e\notin\tback$, 
we conclude the following 
for $\bomega=(\omega,\tau)$, $\bomega'=(\omega',\tau') \in \Bomega_\vertical$
and $\omega_\hor\in\Omega_\hor$
with some $\beta$-dependent constants 
$\ceins^\vertical,\ceins^\hor>0$, 
$\czwei^\vertical,\czwei^\hor\in\R$
and some linear functions $f^\vertical_\tau$ and  $f^\hor_\tau$:
\begin{align}
H_\vertical(\bomega)&\ge \ceins^\vertical\|\omega\|^2 
+f^\vertical_\tau(\omega)+\czwei^\vertical,\\
H_\hor(\bomega,\omega_\hor,\bomega')&\ge 
\ceins^\hor\|(\omega,\omega_\hor,\omega')\|^2 
+f^\hor_\tau(\omega,\omega_\hor,\omega')+\czwei^\hor.
\end{align}
Using the definitions \eqref{def Hmitte} and \eqref{def Hmittepm} 
of $H_\mitte$ and $H_\mitte^\pm$, 
we get that $e^{-H_\mitte(\bomega,\omega_\hor,\bomega')}$
and $e^{-H_\mitte^\pm(\bomega,\omega_\hor,\bomega')}$
are bounded by a $(\tau,\tau')$-dependent Gaussian in 
the arguments $(\omega,\omega_\hor,\omega')$.
Integrating over $\omega_\hor$, square integrability of 
$k$ and $k^\pm$ follows.
Similarly, using the definition \eqref{def Hlinksrechts}
of $H_\links$ and $H_\rechts$, it follows that
$\Psi_\links(\omega)=e^{-H_\links(\bomega)}$ and 
$\Psi_\rechts(\omega)=e^{-H_\rechts(\bomega)}$
are bounded by $\tau$-dependent Gaussians in $\omega$
and hence square integrable.
Since $\chi$ is bounded and $\nabla t_\hor$ depends linearly on $\omega_\hor$,
square integrability of $\tilde{k}_\alpha$ follows by the same argument.
\end{proof}
\begin{definition}\label{def-tr}
We define the transfer operators $\K$, $\K^\pm$, and $\tilde\K_\alpha$  by 
\begin{align}
 \K F (\bomega )= \int_{\Bomega_\vertical} k (\bomega ,\bomega') 
F (\bomega') \, d\bomega'
\end{align}
and similarly for $\K^\pm$ and $\tilde\K_\alpha$ using the integral kernels 
$k^\pm$ and $\tilde k_\alpha$ instead of $k$. 
\end{definition}

By Lemma \ref{lemmaHS} above, these transfer operators are 
Hilbert-Schmidt operators from $L^2(\Bomega_\vertical,d\bomega)$ 
to $L^2(\Bomega_\vertical,d\bomega)$. They satisfy the following properties.
\begin{lemma}
\label{lemma:spectral}
The spectral radii $\lambda$, $\lambda^\pm$ of the integral operators 
$\K$, $\K^\pm$, and their adjoints $\K^*$, $(\K^\pm)^*$ are strictly positive 
eigenvalues of the corresponding operator and its adjoint. The 
corresponding eigenspaces are one-dimensional and spanned by strictly
positive functions, denoted by $\Phi_\rechts$, $\Phi^\pm_\rechts$, 
$\Phi_\links$, $\Phi^\pm_\links$, respectively. 
We normalize these functions such that 
$\sk{\Phi_\links,\Phi_\rechts}=1$ and $\sk{\Phi^\pm_\links,\Phi^\pm_\rechts}=1$.
Projections 
to the eigenspaces of $\K$, $\K^\pm$ are given
by $P\Psi=\sk{\Phi_\links,\Psi}\Phi_\rechts$ and
$P^\pm\Psi=\sk{\Phi_\links^\pm,\Psi}\Phi^\pm_\rechts$, respectively.
They fulfill $\K P=P\K=\lambda P$ and
\begin{align}
\label{perron-frobenius}
&\|\K^m(\id-P)\|=
\|\K^m-\lambda^m P\|=O(a^m\lambda^m)\text{ as } m\to\infty,
\\
\label{perron-frobenius2}
&\|(\K^\pm)^m-(\lambda^\pm)^m P^\pm\|=O((a^\pm)^m(\lambda^\pm)^m)\text{ as } m\to\infty
\end{align}
with some constants $a,a^\pm\in[0,1)$.
\end{lemma}
\begin{proof}
From Theorem \ref{thm:bijection}, it follows that 
some power $\K^N$ of $\K$ has a strictly positive integral kernel.
The same holds for some power of $\K^\pm$. Furthermore, the values 
of the integral kernels $k(\bomega,\bomega')$ and 
$k^\pm(\bomega,\bomega')$ 
are strictly positive  whenever the tree variables $\tau$ in $\bomega$ 
and $\tau'$ in $\bomega'$ both equal $\tau_\bb$. 
Hence, the lemma follows 
by the Perron-Frobenius-Jentzsch theory; see appendix. 
\end{proof}

We remark that $P$ and $P^\pm$ need not be self-adjoint.\\

\subsection{Bound on the energy}
\label{sect:energyproof}
Using the transfer operator representation, we can now prove the estimate on the
energy term. In the following we abbreviate for $m\in\N$: 
\begin{align}
\Psi_\links^m:=  ((\K^-)^m)^* \Psi_\links , \quad
\Psi_\rechts^m:=(\K^+)^m\Psi_\rechts.
 \end{align}
We have the following result.
\begin{lemma}
\label{lemma-energy-in-terms-of-operator}
The energy term $\Egamma$ defined in (\ref{Egamma}) can be written as
\begin{align}
\label{eq:represent Egamma}
\Egamma= 
\frac12\sum_{n=0}^{l-1}
\frac{\sk{ \Psi^\ul_{\links},  \mathcal{K}^n\  \tilde{\mathcal{K}}_\alpha 
\  \mathcal{K}^{l-1-n}   \Psi^{\ol-l}_{\rechts}}}
{\sk{\Psi^{\ul}_{\links},  
\mathcal{K}^l   \Psi^{\ol-l}_{\rechts}}}.
\end{align}
\end{lemma}
\begin{proof}{}
Using Lemma \ref{representation-H-interpolated}, 
this is just a rewriting of the integral 
in \eqref{Egamma} in terms of transfer operators. 
\end{proof}

We will prove below that  each term in this sum can be written as a leading term 
independent of $n,$ $l$ and $L$ plus a rest that is summable over $n$ and 
uniformly bounded in $l$ and $L$. The key estimate is proved in the following lemma.

\begin{lemma}
\label{lemma: fin volume to inf volume}
For any $m,n,m',n'\in\N$, we have 
\begin{align}
\label{claim transfer}
\frac{ \sk{\Psi^m_{\links}, \K^n\tilde{\K}_\alpha\K^{n'}\Psi^{m'}_{\rechts}}}{
\sk{\Psi^m_{\links}, \K^{n+n'+1}\Psi^{m'}_{\rechts}}}= 
\frac{ \sk{\Phi_\links,  \tilde{\K}_\alpha\Phi_\rechts}}
{ \sk{\Phi_\links, \K\Phi_\rechts}}
+ R_{m,n,m',n'}(\alpha) 
\end{align}
with a rest term $R_{m,n,m',n'}(\alpha)$ that fulfills
\begin{align}
\label{remainder-transfer}
\sup_{m,n,m',n'\in\N}\frac{|R_{m,n,m',n'}(\alpha)|}{a^{\min\{n,n'\}}}
<\infty,
\end{align}
where $a\in(0,1)$ is taken from Lemma \ref{lemma:spectral},
and $\sk{\Phi_\links, \K\Phi_\rechts}=\lambda$ by construction.
\end{lemma}
\begin{proof}
Since $\Psi_\links(\omega,\tau_\bb)>0$ for any $\omega \in \Omega_{\vertical }$
and $k^-((\omega,\tau_\bb),(\omega',\tau_\bb))>0$ for any $\omega,\omega'$,
we have  $\|\Psi_\links^m\|>0$ for any $m\geq 0$.
Therefore  the normalized quantities
$\hat\Psi^m_\links=\Psi^m_\links/\|\Psi^m_\links\|$ and 
$\hat\Phi^-_\links=\Phi^-_\links/\|\Phi^-_\links\|$
are well defined. Replacing $\links $ by $\rechts $
and $k^{-}$ by $k^{+}$ we also find  $\|\Psi_\rechts^m\|>0$ for any $m\geq 0$,
so 
$\hat\Psi_\rechts^m=\Psi_\rechts^m/\|\Psi_\rechts^m\|$ and
$\hat\Phi^+_\rechts=\Phi^+_\rechts/\|\Phi^+_\rechts\|$ are well defined too. 
We will work with
the normalized operators $\hat\K=\lambda^{-1}\K$ and 
$\hat\K_\alpha=\lambda^{-1}\tilde{\K}_\alpha$. Then
\begin{align}
\label{claim transfer2}
\frac{ \sk{\Psi^m_{\links}, \K^n\tilde\K_\alpha\K^{n'}
\Psi^{m'}_{\rechts}}}{
\sk{\Psi^m_{\links}, \K^{n+n'+1}\Psi^{m'}_{\rechts}}}=
\frac{ \sk{\hat\Psi^m_{\links}, \hat\K^n\hat\K_\alpha\hat\K^{n'}
\hat\Psi^{m'}_{\rechts}}}{
\sk{\hat\Psi^m_{\links}, \hat\K^{n+n'+1}\hat\Psi^{m'}_{\rechts}}}
\end{align}
Abbreviating $P^c=\id-P$,
we split the numerator in (\ref{claim transfer2}) in four pieces:
\begin{align}
\label{split sk}
&\sk{\hat\Psi^m_{\links}, \hat\K^n\hat\K_\alpha\hat\K^{n'}
\hat\Psi^{m'}_{\rechts}}
=
\sk{\hat\Psi^m_{\links}, P\hat\K^n\hat\K_\alpha\hat\K^{n'}P
\hat\Psi^{m'}_{\rechts}}
+\sk{\hat\Psi^m_{\links}, P^c\hat\K^n\hat{\K}_\alpha\hat\K^{n'}P\hat\Psi^{m'}_{\rechts}}
\nonumber\\
&\qquad \qquad 
+\sk{\hat\Psi^m_{\links}, P\hat\K^n\hat\K_\alpha\hat\K^{n'}P^c\hat\Psi^{m'}_{\rechts}}
+\sk{\hat\Psi^m_{\links}, P^c\hat\K^n\hat\K_\alpha\hat\K^{n'}P^c\hat\Psi^{m'}_{\rechts}}.
\end{align}
Using $P\hat\K^n=P$ and $\hat\K^{n'}P=P$,
the first piece in the sum equals
\begin{align}
&\sk{\hat\Psi^m_{\links}, P\hat\K_\alpha P\hat\Psi^{m'}_{\rechts}}
=\sk{\hat\Psi^m_{\links},\Phi_\rechts}
\sk{\Phi_\links,\hat\K_\alpha\Phi_\rechts}\sk{\Phi_{\links},\hat\Psi^{m'}_\rechts}.
\end{align}
In the remainder of this proof, the notation ``$b_n=O(a^n)$'' means that 
there is a constant $c<\infty$ such that $\sup_{n\in\N} |b_n/a^n| \le c$. 
In particular, it implies that $b_n$ is finite for every $n$. 
Using (\ref{perron-frobenius}), the second term in the sum
in (\ref{split sk}) is bounded by
\begin{align}
&\left|\sk{\hat\Psi^m_{\links}, P^c\hat\K^n\hat\K_\alpha\hat\K^{n'}P\hat\Psi^{m'}_{\rechts}}\right|
\le 
\|P^c\hat\K^n\|\|\hat\K_\alpha\|\|P\|
= O(a^n)
\end{align}
where the constant in $O(a^n)$ may depend on $\alpha$, but not on 
$m,m',n,n'$. The third and fourth terms fulfill a similar bound with 
$O(a^{n'})$ and $O(a^{n+n'})$, respectively, instead of $O(a^n)$. Then
the numerator in (\ref{claim transfer2}) can be written as
\begin{align}
\label{bound-numerator}
&\sk{\hat\Psi^m_{\links}, \hat\K^n\hat\K_\alpha\hat\K^{n'}\hat\Psi^{m'}_{\rechts}}
=\sk{\hat\Psi^m_{\links},\Phi_\rechts}
\sk{\Phi_\links,\hat\K_\alpha\Phi_\rechts}\sk{\Phi_{\links},\hat\Psi^{m'}_\rechts}+
O(a^{\min\{n,n'\}}).
\end{align}
Now, we claim
\begin{align}
\label{term-reference}
(a)\ \inf_{l,m,m'\in\N}
\sk{\hat\Psi^m_{\links},\hat\K^l   \hat\Psi^{m'}_{\rechts}}>0, \qquad  (b)\  
\inf_{m,m'\in\N}\sk{\hat\Psi^m_{\links},\Phi_\rechts}
\sk{\Phi_\links,\hat\Psi^{m'}_\rechts}>0.
\end{align}
Assuming this is true and combining the estimate \eqref{bound-numerator} 
for the numerator with the fact \eqref{term-reference}$(a)$
that the denominator is uniformly bounded away from 0, the right-hand 
side of \eqref{claim transfer2} can be written as 
\begin{align}
& \sk{\Phi_\links,\hat\K_\alpha\Phi_\rechts}\frac{\sk{\hat\Psi^m_{\links},\Phi_\rechts}
\sk{\Phi_{\links},\hat\Psi^{m'}_\rechts}}
{\sk{\hat\Psi^m_{\links}, \hat\K^{n+n'+1}\hat\Psi^{m'}_{\rechts}}}
+O(a^{\min\{n,n'\}}).
\label{frac-phi-k-alpha} 
\end{align}
To estimate the denominator in the last expression, 
we use the following bound 
\begin{align}
&\sup_{m,m'}\left|
\sk{\hat\Psi^m_{\links},\hat\K^l   \hat\Psi^{m'}_{\rechts}}
-\sk{\hat\Psi^m_{\links},\Phi_\rechts}\sk{\Phi_\links, \hat\Psi^{m'}_{\rechts}}
\right|
\cr&
\qquad =\sup_{m,m'}
\left|
\sk{\hat\Psi^m_{\links},(\hat\K^l-P)   \hat\Psi^{m'}_{\rechts}}
\right|
\le \|\hat\K^l-P\|   
\le O(a^l)
\label{leading term}
\end{align}
with $l=n+n'+1$. Note that the proof of this estimate does not use \eqref{term-reference}. 
The fraction in \eqref{frac-phi-k-alpha} is bounded 
from above using \eqref{term-reference}$(a)$ and the fact that 
$\|\hat\Psi_\links^m\|=1=\|\hat\Psi_\rechts^{m'}\|$. 
Moreover, this fraction is 
bounded from below by a positive constant using 
\eqref{term-reference}$(b)$ and \eqref{leading term}. 
These estimates hold uniformly in $m,m',n,n'$. In particular
by \eqref{leading term} 
the reciprocal is estimated as follows
\begin{align}
\sup_{\substack{n,n':\\ n+n'+1=l}} \sup_{m,m'}
\left| \frac{\sk{\hat\Psi^m_{\links}, \hat\K^{n+n'+1}\hat\Psi^{m'}_{\rechts}}}
{\sk{\hat\Psi^m_{\links},\Phi_\rechts}
\sk{\Phi_{\links},\hat\Psi^{m'}_\rechts}} -1 \right|
=O(a^{l}). 
\end{align}
These estimates give (\ref{claim transfer}) 
with the bound \eqref{remainder-transfer}.
To complete the proof of the lemma we now prove claim \eqref{term-reference}.
Recall that $k((\omega,\tau_\bb),(\omega',\tau_\bb))>0$,
$\hat\Psi_\links^m(\omega,\tau_\bb)>0$, and
$\hat\Psi_\rechts^{m'}(\omega',\tau_\bb)>0$ for all $\omega,\omega'$ and
$m,m'\in\N$.
It follows 
\begin{align}
\label{pointwise-positive}
\sk{\hat\Psi^m_{\links},\hat\K^l   \hat\Psi^{m'}_{\rechts}}>0  \qquad  \text{ and }
\qquad \sk{\hat\Psi^m_{\links},\Phi_\rechts}\sk{\Phi_\links,\hat\Psi^{m'}_\rechts}>0
\end{align}
for all $l,m,m'\in\N$. Similarly, since $ \hat\Phi_\links^{-}>0$ and  $ \hat\Phi_\rechts^{+}>0$,
for all $l,m,m'\in\N$ we have
 \begin{align}
\label{pointwise-positive2}
\sk{\hat\Psi^m_{\links},\hat\K^l   \hat\Phi^{+}_{\rechts}}>0  \qquad  \text{ and }
\qquad \sk{\hat\Phi^{-}_{\links},\hat\K^l   \hat\Psi^{m'}_{\rechts}}>0. 
\end{align}
We apply the Perron-Frobenius-Jentzsch theory to
$\hat\K^\pm$. More specifically, we observe first
$(P^-)^*\hat\Psi_\links=c^-\hat\Phi_\links^-$
and $P^+\hat\Psi_\rechts=c^+\hat\Phi_\rechts^+$ where  $c^-,c^+>0$
since  $\sk{\hat\Psi_\links,\hat\Phi_\rechts^-}>0$ and
 $\sk{\hat\Phi_\links^+,\hat\Psi_\rechts}>0$.
Then, using \eqref{perron-frobenius2} we get 
\begin{align}
\label{limit-psim}
\hat\Psi^m_{\links}\stackrel{m\to\infty}{\longrightarrow}
\hat\Phi^-_\links
\quad\text{and}\quad
\hat\Psi^m_{\rechts}\stackrel{m\to\infty}{\longrightarrow}
\hat\Phi^+_\rechts,
\end{align}
where the limits are taken with respect to $\|{\cdot}\|$.
Thus,
\begin{align}
\sk{\hat\Psi^m_{\links},\Phi_\rechts} 
\stackrel{m\to\infty}{\longrightarrow}
\sk{\hat\Phi_{\links}^-,\Phi_\rechts}>0
\text{ and } 
\sk{\Phi_\links,\hat\Psi^m_{\rechts}} 
\stackrel{m\to\infty}{\longrightarrow}
\sk{\Phi_{\links},\hat\Phi_\rechts^+}>0.
\end{align}
Combining this with \eqref{pointwise-positive}, we get 
\eqref{term-reference}(b).
Using \eqref{leading term}, 
this implies for
$l_0\in\N$ large enough 
\begin{align}
\label{infinf psiKpsi}
\inf_{l\ge\ l_0}\inf_{m,m'\in\N}
\sk{\hat\Psi^m_{\links},\hat\K^l\hat\Psi^{m'}_{\rechts}}
>0.
\end{align}
We consider a given $l<l_0$ next.
From \eqref{pointwise-positive2}, 
\eqref{limit-psim}, and $\sk{\hat\Phi^-_\links,\hat\K^l\hat\Phi_\rechts^+}>0$, 
we have
\begin{align}
\inf_{m\in\N}\sk{\hat\Psi^m_\links,\hat\K^l\hat\Phi_\rechts^+}>0
\text{ and }
\inf_{m'\in\N}\sk{\hat\Phi_\links^-,\hat\K^l\hat\Psi^{m'}_\rechts}>0.
\end{align}
Using this, \eqref{pointwise-positive}, 
and \eqref{limit-psim} again, we find for our given $l$:
$\inf_{m,m'\in\N}\sk{\hat\Psi^m_{\links},\hat\K^l\hat\Psi^{m'}_{\rechts}}>0$.
Combining this with \eqref{infinf psiKpsi}, the 
claim (\ref{term-reference})(a) follows.
\end{proof}

\noindent
\begin{proof}[Proof of Theorem \ref{thm: energy contribution}]
Combining Lemmas \ref{lemma-energy-in-terms-of-operator} 
and \ref{lemma: fin volume to inf volume} above,
 the energy term $\Egamma$ defined in (\ref{Egamma}) can be written as
\begin{align}\label{egamma1}
\Egamma= &
\frac12 l\frac{\sk{\Phi_\links,  \tilde{\K}_\alpha\Phi_\rechts}}
{\sk{\Phi_\links,  \K\Phi_\rechts}}+ 
\frac12
\sum_{n=0}^{l-1}R_{\ul,n,\ol-l,l-1-n}(\alpha)
\end{align}
where,  for any given $\alpha$, the rest 
\begin{equation}\label{egamma2}
\frac12
\sum_{n=0}^{l-1}R_{\ul,n,\ol-l,l-1-n}(\alpha)=\cdrei(G_0,\beta,L,l,\alpha )
\end{equation}
is bounded uniformly in $L$ and $l$ since  $0\le a<1$. Now we claim 
there is a constant $\cvier=\cvier(G_0,\beta)>0$
such that for all $\alpha\in\R$ one has 
\begin{equation}\label{thm: symmetric and positive}
 \sk{\Phi_\links,  \tilde{\K}_\alpha \Phi_\rechts } = 
2\alpha\lambda  \cvier.
\end{equation}
To prove this
we split $\tilde\K_\alpha=\tilde\K_0 + (\tilde\K_\alpha-\tilde\K_0)$.
We claim
\begin{align}
\label{claim symmetry}
\sk{\Phi_\links,  \tilde{\K}_0 \Phi_\rechts } = 0.
\end{align}
This is proved by symmetry. Recall from \eqref{bomega-refl} that
we set $\bomega^\refl=(\nabla t, y,\tau^\refl)\in\Bomega_\vertical$
for $\bomega=(\nabla t, y,\tau)\in\Bomega_\vertical$. 
From (\ref{symmetry-Hmitte}) it follows for 
$\bomega,\bomega'\in\Bomega_\vertical$
\begin{align}
k(\bomega',\bomega)=k(\bomega^\refl,\bomega'^\refl)
\quad\text{ and }\quad  
\tilde{k}_0(\bomega',\bomega)=-\tilde{k}_0(\bomega^\refl,\bomega'^\refl).
\end{align}
Since $k$ is real-valued, the first equation implies that
$(\bomega,\bomega')\mapsto k(\bomega^\refl,\bomega'^\refl)$
is the integral kernel of the adjoint  $\K^*$ of $\K$.
Consider the ``reflected'' eigenfunctions
$\Phi_\rechts^\refl,\Phi_\links^\refl:\Bomega_\vertical\to(0,\infty)$,
$\Phi_\rechts^\refl(\bomega)=\Phi_\rechts(\bomega^\refl)$,
$\Phi_\links^\refl(\bomega)=\Phi_\links(\bomega^\refl)$.
Since the reflection $\refl$
leaves the reference measure $d\bomega$ invariant and $\K\Phi_\rechts=\lambda\Phi_\rechts$,
 we get  $\K^*\Phi_\rechts^\refl=\lambda\Phi_\rechts^\refl$. But
the eigenspace $E_\lambda(\K^*)$ is spanned by $\Phi_\links$, then 
$\Phi_\rechts^\refl=c\Phi_\links$ for some constant $c>0$, and therefore
 $\Phi_\links^\refl=c^{-1}\Phi_\rechts$.
We conclude
\begin{align}
&\sk{\Phi_\links,  \tilde{\K}_0 \Phi_\rechts } =
\int_{\Bomega_\vertical}\int_{\Bomega_\vertical}
\Phi_\links(\bomega^\refl)\tilde{k}_0(\bomega^\refl,\bomega'^\refl)
\Phi_\rechts(\bomega'^\refl)\,d\bomega\,d\bomega'
\cr
&\qquad =
\sk{\Phi_\rechts^\refl,  -\tilde{\K}_0 \Phi_\links^\refl }
=-\sk{\Phi_\links,  \tilde{\K}_0 \Phi_\rechts }.
\end{align}
This proves claim \eqref{claim symmetry}. The 
contribution of $\tilde\K_\alpha-\tilde\K_0$ is given by
\begin{align}
\label{eq: def c1}
\sk{\Phi_\links,(\tilde{\K}_\alpha-\tilde{\K}_0)\Phi_\rechts }
=
\alpha \int_{\Bomega_\mitte}
\chi(\bomega,\omega_\hor,\bomega')
e^{-H_{\mitte} (\bomega,\omega_\hor,\bomega')} 
d\bomega\,d\omega_\hor\,d\bomega'=:2\alpha \lambda \cvier.
\end{align}
Note that $\cvier>0$, because the integrand in
\eqref{eq: def c1} is nonnegative everywhere and positive 
on a set of positive measure.
This proves claim \eqref{thm: symmetric and positive}.
Finally,  combining \eqref{egamma1}, \eqref{egamma2} and \eqref{thm: symmetric and positive} 
the proof of the theorem follows.
\end{proof}

\section{Putting pieces together}\label{sect:end}

\noindent\begin{proof}[Proof of Theorem \ref{th1}]
We prove the estimate for $0<l\le\ol$. The case 
$-\ul \leq l<0$ follows by reflection symmetry.
Take any $\eta>0$ and $\alpha\in\R$ with $|\alpha|\le\cneun\eta$. 
Using \eqref{expectation-exp-tl-t0}, \eqref{def-Z}, and 
Lemma \ref{upper-bound-Z-Pi-gamma}
\begin{align}
\ln \mathbb{E}_{ \mu_L^{\0}}
\left[e^{\frac{t_{\ell}-t_{\0 } }{2}} \right] =
\ln \mathbb{E}_{ \mu_L^{\grad,\0}} 
\left[e^{\frac{t_{\ell}-t_{\0 } }{2}} \right] 
=\ln Z_{\0 \ell}^L 
\leq  \Egamma + \Sgamma{\alpha} .
\end{align}
Inserting the expression for the internal energy $\Egamma$ from 
Theorem \ref{thm: energy contribution} and 
the estimate for the entropy term $\Sgamma{\alpha}$ from Theorem \ref{entropyth}
we conclude
\begin{align}
\ln \mathbb{E}_{ \mu^{\grad,\0}_L}
\left[e^{\frac{t_{\ell}-t_{\0 } }{2}} \right]  \le 
\cfuenf \alpha^2l + \alpha l \cvier+ \cdrei^{\max}(\alpha ) 
= \alpha (\cfuenf\alpha + \cvier )\,l + \cdrei^{\max}(\alpha ) .
\end{align}
Taking $\alpha<0$ such that $|\alpha|<\min\{\cvier/\cfuenf,\cneun\eta\}$, 
Claim \eqref{claim-main-thm} follows with $\czehn=e^{\cdrei^{\max}(\alpha )}$
and $\celf=- \alpha (\cfuenf\alpha + \cvier )>0$.

To prove the second part of the theorem, we notice that, 
by Theorem \ref{th:muzero},  
the variables $(t_{\0},s_\0)$ and $(\nabla t_{\bb },y_\bb)$ are stochastically 
independent and the distribution of $(t_{\0},s_\0)$ is independent of $L$. 
Moreover using \eqref{repr-nabla-t} 
with $i=\0$ we can reconstruct any $t_{j}$ knowing only $t_{\0 }$ and a 
{\it finite} number
of variables $\nabla t^{\bb }_{e}$, independent of $L$, for $L$ large 
enough. A similar statement holds for the $s_j$ using only a finite 
number of variables $\nabla t^\bb_e$ and $y^\bb_e$. 
Then any local observable of $(t_{j},s_j)_{j\in V}$ 
(i.e.\ depending only on a finite number of lattice sites) 
can be written as a local observable of  $(t_{\0 },s_\0,\nabla t_{\bb },y_\bb)$ 
uniformly in $L$ for $L$ large enough. Using the definitions of the previous 
section, in analogy to Lemma \ref{lemma-energy-in-terms-of-operator}, 
the average of 
any bounded local observable can be written as a ratio of scalar products
\begin{align}
\label{expectation-O}
\mathbb{E}_{\mu_L^{\0}}[\mathcal{O}]=  
\frac{\sk{ \hat\Psi^{\ul-j_{1}}_{\links},  \mathcal{K}_{\mathcal{O}}   
\  \hat\Psi^{\ol-j_{2}}_{\rechts}}}{\sk{\hat\Psi^{\ul-j_{1}}_{\links},  
\left(\mathcal{K}^{-}\right)^{j_{1}}
\left(\mathcal{K}^{+}\right)^{j_{2}}    \hat\Psi^{\ol-j_{2}}_{\rechts}}},
\end{align}
where the operator $ \mathcal{K}_{\mathcal{O}}$ depends on the observable 
and the level indices
$j_{1},j_2\geq 0$ need to be large enough depending on the choice of 
$ \mathcal{K}_{\mathcal{O}}$.
Since $ \hat\Psi^m_{\links/\rechts }\rightarrow\hat\Phi^{-/+}_{\links/\rechts }$ as 
$m\to\infty$ by \eqref{limit-psim}, the limit of \eqref{expectation-O}
as $\ul,\ol\to\infty$ is well defined.
If for $L$ large enough one replaces $j_1$ by $j_1+n_1$ and $j_2$ by 
$j_2+n_2$ and $\mathcal{K}_{\mathcal{O}}$ by 
$\mathcal{K}_-^{n_1}\mathcal{K}_{\mathcal{O}}\mathcal{K}_+^{n_2}$, the 
expression \eqref{expectation-O} remains unchanged. 
This holds also in the limit as $\ul,\ol\to\infty$. 
Hence, the consistency conditions in Kolmogorov's extension theorem ensure 
the limiting measure $\mu^\0_\infty$ exists. This completes the proof.
\end{proof}

\smallskip\noindent\begin{proof}[Proof of Corollary \ref{cor-vrjp}]
Let $\tilde{Y}= (\tilde{Y}_{n})_{n\in \mathbb{N}_{0}}$ be the discrete 
time process associated to VRJP on $\G_\rho$. 
Let $\G_{\rho,L}$ denote the graph obtained from $\G_L$ by adding the 
additional vertex $\rho$ connected by an edge to $\0$. 
For any fixed time $T$, 
the process $\tilde{Y}$ can jump at most a distance of $T$ away from 
its starting point. Consequently, the law of $(\tilde{Y}_n)_{n=0,\ldots,T}$ 
agrees with the law of the discrete time process 
$(\tilde{Y}_n^L)_{n=0,\ldots,T}$ associated with the VRJP on 
$\G_{\rho,L}$ for all $L=(-\ul,\ol)$ with 
$\ul,\ol>T$. Thus, by Theorem 2 and the remarks in Sect.\  6 of 
\cite{sabot-tarres2012}, the process $(\tilde{Y}_n)_{n=0,\ldots,T}$ 
is a mixture of reversible Markov chains with mixing measure given by
random weights on the edges $W_{ij}(t,s)=\beta_{ij}e^{t_i+t_j}$ with 
$t_i,t_j$ distributed
according to $\mu^\0_L$ with $L$ large enough. Note that 
the edge weights are strictly positive. In particular, 
for any finite path $\vec{v}:=(v_0=\rho,v_1,\ldots,v_T)$ in $\G_\rho$ starting 
in $\rho$ and all $L$ with $\ul,\ol>T$, one has 
\begin{align}
\label{mixture}
\p( (\tilde{Y}_n)_{n=0,\ldots,T}=\vec{v})= 
\int\p^{W (t,s)} ( (\tilde{Y}_n)_{n=0,\ldots,T}=\vec{v}) \,  d\mu_\0^L (t,s).
\end{align}
Since $\p^{W (t,s)} ( (\tilde{Y}_n)_{n=0,\ldots,T}=\vec{v})\in[0,1]$ 
is a bounded observable, Theorem \ref{th1} implies that the right-hand side
of \eqref{mixture} equals 
$\int\p^{W (t,s)} ( (\tilde{Y}_n)_{n=0,\ldots,T}=\vec{v}) \,  d\mu^\0_\infty(t,s)$.
The events $\{ (\tilde{Y}_n)_{n=0,\ldots,T}=\vec{v} \}$
together with the empty set are closed under intersections 
and generate the whole space. This shows that $\tilde Y$ is a mixture of 
reversible Markov chains with mixing measure $\mu^\0_\infty$. Since 
all edge weights are strictly positive,
one has a mixture of irreducible Markov chains. 

To prove positive recurrence,  
let $x=(|x|=l,v)\in V$ be an arbitrary site and set $\ell:=(l,p)$.
By Theorem \ref{th:muzero}, $\nabla t_\bb$ and $t_\0$ are independent. 
Using this fact and the Cauchy-Schwarz inequality, we get 
for all $L$ with $-\ul\le l\le\ol$, 
\begin{equation}\label{7.5}
\mathbb{E}_{\mu_\0^L}\left[e^{\frac{t_x}{4}}\right]
= \mathbb{E}_{\mu_\0^L}\left[
e^{\frac{t_x-t_\ell}{4}} e^{\frac{t_\ell-t_\0}{4}} e^{\frac{t_\0}{4}}\right]
\le  \mathbb{E}_{\mu_\0^L}\left[ e^{\frac{t_x-t_\ell}{2}} \right]^{1/2}
\mathbb{E}_{\mu_\0^L}\left[ e^{\frac{t_\ell-t_\0}{2}} \right]^{1/2}
\mathbb{E}_{\mu_\0^L}\left[ e^{\frac{t_\0}{4}}\right].
\end{equation}
It follows from Theorem \ref{th:muzero} that 
$\mathbb{E}_{\mu_\0^L}\left[ e^{\frac{t_\0}{4}}\right]<\infty$.
Since $x$ and $\ell$ are on the same level, 
$\gamma^{x\ell}_{\tback }\subset S_l$ and thus
\begin{equation}\label{7.6}
\mathbb{E}_{\mu_\0^L}\left[ e^{\frac{t_x-t_\ell}{2}} \right]\leq  
\mathbb{E}_{\mu_\0^L}
\Bigg [ \prod_{e\in \gamma^{x\ell}_{\tback }}e^{\frac{|\nabla t^{\bb }_{e}|}{2}} \Bigg ]
\leq C_{\beta }^{|S|} 
\mathbb{E}_{\mu_\0^L}
\left[ 
e^{\sum_{e\in \gamma^{x\ell}_{\tback }}
\frac{\beta_{e} }{2} ( \cosh \nabla t^{\bb }_{e} -1 +\frac{y_e^2}{2})} \right] 
\leq (2C_{\beta})^{|S|}
\end{equation}
where $C_{\beta }$ is a constant that depends only on $(\beta_{e})_{e\in S}$ and
in the last inequality we used a straightforward generalization of 
\cite[Lemma 3, eq. (6.2) and (2.10)]{disertori-spencer-zirnbauer2010} to the case
of variable $\beta$.

Using Theorem \ref{th1}, we conclude 
\begin{align}
\mathbb{E}_{\mu_\0^L}\left[e^{\frac{t_x}{4}}\right]
\le \cvierzehn e^{-\celf\frac{|x|}{2}}
\end{align}
with a constant $\cvierzehn>0$. Then, by monotone convergence, we get
\begin{align}
\mathbb{E}_{\mu^\0_\infty}\left[e^{\frac{t_x}{4}}\right]
=\lim_{M\to\infty}\mathbb{E}_{\mu^\0_\infty}\left[
e^{\frac{t_x}{4}}\one_{\{|t_x|\le M\}}\right]
=\lim_{M\to\infty}\lim_{L\to\infty}
\mathbb{E}_{\mu_\0^L}\left[ e^{\frac{t_x}{4}}\one_{\{|t_x|\le M\}}\right]
\le \cvierzehn e^{-\celf\frac{|x|}{2}}. 
\end{align}
Using the exponential Chebyshev inequality and a Borel-Cantelli argument, 
it follows that $\sum_{x\in V} e^{t_x}<\infty$ $\mu^\0_\infty$-a.s. 
Consequently, 
 $\mu^\0_\infty$-a.s.\ the edge weights are summable
\begin{align}
\sum_{i\sim j}\beta_{ij}e^{t_i+t_j}\le (\max_{ij}\beta_{ij})
\sum_{i\in V} e^{t_i} \sum_{j\in V} e^{t_j} <\infty,
\end{align}
and hence we have a mixture of positive recurrent Markov chains. 
\end{proof}

\smallskip\noindent\begin{proof}[Proof of Corollary \ref{lemma-exp-loc}]
The argument is similar to the one used in Theorem 2.1 in
\cite{merkl-rolles-asymptotic}. For
$(t,s)\in\R^{V\times V}$, let $\p_v^{W(t,s)}$ denote the distribution
of the Markovian random walk on $\G_\rho$ with weights 
$W_{ij}=W_{ij}(t,s)=\beta_{ij}e^{t_i+t_j}$ starting at $v\in V\cup\{\rho\}$.
This random walk is reversible with a reversible measure given by
\begin{align}
\pi^{W}_{i}=\sum_{j\sim i}W_{ij},\quad i\in V\cup\{\rho\}.
\end{align}
For all $n\in\N_0$ and $v\in V$, one has
\begin{align}
\pi^{W}_{\rho }\p_\rho^{W}(\tilde{Y}_n=v)
=\pi^{W}_{v}\p_v^{W}(\tilde{Y}_n=\rho).
\end{align}
Then  for any $\alpha\in(0,1)$
\begin{align}
\p_\rho^{W}(\tilde{Y}_n=v)
&=\left[ \frac{\pi^{W}_v}{\pi^{W}_\rho}\right]^\alpha
\p_v^{W}(\tilde{Y}_n=\rho)^\alpha\, 
\p_\rho^{W}(\tilde{Y}_n=v)^{1-\alpha}
\cr &\le
\left[\frac{\pi^{W}_v}{\pi^{W}_\rho}\right]^\alpha
=
\left[\sum_{i\sim v}\frac{\beta_{iv}}{\varepsilon}e^{t_i+t_v-t_\0}\right]^\alpha
\le 
\frac{\betamax^\alpha}{\varepsilon^\alpha}\sum_{i\sim v}e^{\alpha(t_i+t_v-t_\0)}.
\end{align}
Integrating over $s$ and $t$ with respect to $\mu^\0_\infty$, as 
in \eqref{mixing}, we conclude
\begin{align}
\p(\tilde{Y}_n=v)
=
\int\p_\rho^{W(t,s)} (\tilde{Y}_n=v) \,  d\mu^\0_\infty(t,s)
\le
\frac{\betamax^\alpha}{\varepsilon^\alpha}\sum_{i\sim v}
\mathbb{E}_{\mu^\0_\infty}[e^{\alpha(t_i+t_v-t_\0)}].
\end{align}
Let $\ell=(m,p)$ be the copy of the pinning point $\0$ at the level
$m=|v|$ of $v$. Using independence of $t_\0$ from the gradient variables
 (see Theorem \ref{th:muzero} and Lemma \ref{lemma3.1})  and the 
Cauchy-Schwarz inequality, we get, for any $i\sim v$, 
\begin{align}
\mathbb{E}_{\mu^\0_\infty}[e^{\alpha(t_i+t_v-t_\0)}]&=
\mathbb{E}_{\mu^\0_\infty}[e^{\alpha(t_i-t_\ell+t_v-t_\ell)}
e^{2\alpha(t_\ell-t_\0)}e^{\alpha t_\0}]
\cr&\le
\mathbb{E}_{\mu^\0_\infty}[e^{2\alpha(t_i-t_\ell+t_v-t_\ell)}]^{1/2}\ 
\mathbb{E}_{\mu^\0_\infty}[e^{4\alpha(t_\ell-t_\0)}]^{1/2}\ 
\mathbb{E}_{\mu^\0_\infty}[e^{\alpha t_\0}]
\end{align}
Now, setting $\alpha=1/8$, we can use Theorem \ref{th1}, plus  
the same arguments we used in 
\eqref{7.5} and \eqref{7.6} above.  We obtain
\[
\frac{\betamax^\alpha}{\varepsilon^\alpha}\sum_{i\sim v}
\mathbb{E}_{\mu^\0_\infty}[e^{\alpha(t_i+t_v-t_\0)}]
\ \leq \ \csechs e^{-\csieben|v|}
\]
for some positive constants $\csechs,\csieben$. This proves the first
claim.
The second claim follows with precisely the same Borel-Cantelli 
argument as in Corollary 2.2 \cite{merkl-rolles-asymptotic}.
\end{proof}

\begin{appendix}
\section{Spectral properties of transfer operators}
This appendix reviews the results from the Perron-Frobenius-Jentzsch
theory that we need. For more background on this theory, we refer to 
\cite{Schaefer74} and \cite{Zaanen1997}.

Let $(\Omega,\mathcal{A},\mu)$ be a $\sigma$-finite
measure space and $k:\Omega\times\Omega\to\R$
be a measurable function with the following properties:
\renewcommand{\labelenumi}{(\alph{enumi})}
\begin{enumerate}
\item
$
\int_\Omega\int_\Omega k(x,y)^2\,\mu(dx)\,\mu(dy)<\infty.
$
\item
$k(x,y)\ge 0$ holds for all $x,y\in\Omega$.
\item
There are $S\in \mathcal{A}$ with $\mu(S)>0$
and $\epsilon>0$ such that $k(x,y)\ge \epsilon$
holds for all $x,y\in S$.
\item
For $N\in\N$ with $N\ge 1$, 
let $k_N:\Omega\times\Omega\to\R$ denote 
the iterated integral kernel,
recursively defined by $k_1=k$ and 
$k_{N+1}(x,y)=\int_\Omega k_N(x,z)k(z,y)\,\mu(dz)$.
For some $N\in\N$ with $N\ge 1$, the kernel $k_N$ is strictly positive.
\end{enumerate}
Let $\mathcal{H}=L^2(\Omega,\mathcal{A},\mu;\C)$
and $\K,\K^{*}:\mathcal{H}\to \mathcal{H}$ 
 be the linear operators defined by
\[
\K f(x)=\int_\Omega k(x,y) f(y)\,\mu(dy),
\qquad 
\K^*f(y)=\int_\Omega f(x) k(x,y) \,\mu(dx)
\]
for  $\mu$-a.e. $x\in\Omega$, resp. for  $\mu$-a.e. $y\in\Omega$, where $\K^*$
is the  adjoint operator of $\K$. $\K$ and $\K^*$ are Hilbert-Schmidt operators; see 
e.g.\ Theorem VI.23 in \cite{Reed-SimonI}. In particular, they are
compact; see e.g.\ Theorem VI.22(e) in \cite{Reed-SimonI}. 
For $z\in\C$, let
\[
E_z(\K)=\{f\in \mathcal{H}:\;\K f=zf\},\  \text{and}\ 
S_z(\K)=\{f\in \mathcal{H}:\;(\K-z\id)^nf=0\text{ for some }n\in\N\}
\]
denote the eigenspace $E_z(\K)$ and the spectral subspace
$S_z(\K)\supseteq E_z(\K)$ corresponding to $z$, respectively.
The eigenspaces $E_z(\K^*)$ and the spectral subspaces $S_z(\K^*)$
for the adjoint operator are defined in the same way.
Let $\rho(\K)=\{z\in\C:\;\K-z\id:\mathcal{H}\to \mathcal{H}\text{ is bijective}\}$
denote the resolvent set of $\K$, $\sigma(\K)=\C\setminus \rho(\K)$
be the spectrum,
and let $\lambda=r(\K):=\sup_{z\in\sigma(\K)}|z|
=\lim_{n\to\infty}\|\K^n\|^{1/n}$
be the spectral radius. (Equality of these two representations of
$r(\K)$ is proven in Lemma VII 3.4 in 
\cite{Dunford-SchwarzI}.)
Then the following holds:
\renewcommand{\labelenumi}{(\arabic{enumi})}
\begin{enumerate}
\item
\label{claim 1} 
$\sigma(\K)\setminus\{0\}$ consists of isolated points $z$
with $0<\dim E_z(\K)<\infty$, which can  accumulate at most at $0$.
If $\dim \mathcal{H}=\infty$, then $0\in\sigma(\K)$.
\item
\label{claim 2} 
For every $z\in\sigma(\K)\setminus\{0\}$, one has 
$\dim S_z(\K)=\dim S_{\overline{z}}(\K^*)<\infty$.
\item
\label{claim 3} 
$\sigma(\K^*)=\{\overline{z}:z\in\sigma(\K)\}=\sigma(\K)$ and
$\lambda=r(\K)=r(\K^*)>0$.
\item
\label{claim 4} 
$\lambda\in\sigma(\K)$, and $E_\lambda(\K)$ 
contains a $\mu$-a.e.\ positive function $\Phi_\rechts >0$. 
Similarly, $E_\lambda(\K^*)$ contains a $\mu$-a.e.\ positive function $\Phi_\links>0$.
\item 
\label{claim 5}
Let $N\in\N$ with $N\ge 1$,
$f,g\in \mathcal{H}$ with $g\ge 0$ ($\mu$-a.e.).
If $\K^N f=\lambda^N f+g$, it follows that $g=0$ ($\mu$-a.e.).
\item
\label{claim 6} 
$\dim E_\lambda(\K)= 1$ and $\dim E_\lambda(\K^*)= 1$.
\item
\label{claim 7} 
$S_\lambda(\K)=E_\lambda(\K)$ and $S_\lambda(\K^*)=E_\lambda(\K^*)$.
\item
\label{claim 8} 
For every $z\in\sigma(\K)$ with $z\neq\lambda$, one has $|z|<\lambda$.
Furthermore, one has\\
$\sup\{|z|:\;z\in\sigma(\K)\setminus\{\lambda\}\}<\lambda$.
\item
\label{claim 9} 
Let $\Phi_\rechts $ and $\Phi_\links$ from (\ref{claim 4}) 
be normalized such that $\langle \Phi_\links,\Phi_\rechts \rangle=1$.
Set $ P :\mathcal{H}\to \mathcal{H}$, $ P  f=\langle \Phi_\links,f\rangle \Phi_\rechts $.
Then
$\sigma(\K-\lambda P )\setminus\{0\}=\sigma(\K)\setminus\{\lambda,0\}$
and $r(\K-\lambda P )<\lambda$.
\item
\label{claim 10} 
One has
$\lambda^{-n}\K^n- P =(\lambda^{-1}\K- P )^n$ for all $n\in\N$, $n\ge 1$,
and $\|\lambda^{-n}\K^n- P \|=O(a^n)\text{ as }n\to\infty$
for any $a<1$ with $r(\K-\lambda P )<\lambda a$.
Note that such an $a$ exists because of (\ref{claim 9}).
\end{enumerate}
\noindent
\begin{proof}
{\it Claim (\ref{claim 1}).}  This is the content of a theorem by F. Riesz, proven e.g.\ 
as Theorem 7.1 in Chapter VII of \cite{Conway}. 

{\it Claim (\ref{claim 2}).} For any $z\in\C\setminus\{0\}$, the operator
$\K-z\id$ is a Fredholm operator with Fredholm index 
$\operatorname{ind}(\K-z\id)=\dim\operatorname{ker}(\K-z\id)-
\dim\operatorname{ker}(\K-z\id)^*=0$. This follows from the
Fredholm alternative,
see e.g.\ Proposition 3.3 in Chapter XI of \cite{Conway}. 
Then, for any $n\in\N$, 
$(\K-z\id)^n$ is a Fredholm operator with Fredholm index 
$\dim\operatorname{ker}(\K-z\id)^n-
\dim\operatorname{ker}(\K^*-\overline z\id)^n=
\operatorname{ind}(\K-z\id)^n=n\operatorname{ind}(\K-z\id)=0$
as well; see e.g.\ Theorem 3.7 in Chapter XI of \cite{Conway}.
This implies 
$\dim S_z(\K)=\dim S_{\overline{z}}(\K^*)$. To see that these dimensions
are finite, it suffices to show that $\operatorname{ker}(\K-z\id)^n=
\operatorname{ker}(\K-z\id)^{n+1}$ holds for some $n\in\N$,
since this implies $\operatorname{ker}(\K-z\id)^m=
\operatorname{ker}(\K-z\id)^{m+1}$ for all $m\ge n$.
If the inclusion $\operatorname{ker}(\K-z\id)^n\subset
\operatorname{ker}(\K-z\id)^{n+1}$ was strict for all $n\in\N$,
we could choose for every $n\in\N$ some 
$f_n\in\operatorname{ker}(\K-z\id)^{n+1}$ with $\|f_n\|=1$
and $f_n\perp\operatorname{ker}(\K-z\id)^n$.
But then $\K f_n=zf_n+(\K-z\id)f_n\in zf_n+\operatorname{ker}(\K-z\id)^n$
has at least distance $|z|$ from the space $\operatorname{ker}(\K-z\id)^n$.
Because $\K f_m\in\operatorname{ker}(\K-z\id)^n$ holds for all $m<n$,
it follows that $\|\K f_n-\K f_m\|\ge  |z|$ holds for all $m<n$.
This means that the sequence $(\K f_n)_{n\in\N}$ cannot have an accumulation
point, contradicting the fact that $\K$ is a compact operator.

{\it Claim (\ref{claim 3}).} The first equality  is contained in Theorem VI.7 in 
\cite{Reed-SimonI}. The second equality  
follows from the fact that the integral kernel $k$ is real-valued. 
The claim $r(\K)=r(\K^*)$ follows immediately from $\sigma(\K)=\sigma(\K^*)$.
To prove $r(\K)>0$, we use hypotheses (b) and (c) as follows: 
$\K\one_S\ge a \one_S$, where we abbreviate $a=\epsilon \mu(S)>0$. 
Because of $k\ge 0$, 
the operator $\K$ is positive in the sense that $f\ge g$ implies
$\K f\ge \K g$ for any $f,g\in \mathcal{H}$. Inductively, it follows
that $\K^n\one_S\ge a^n \one_S$ holds for all $n\in\N$.
This implies $\|\K^n\|\ge a^n$ for all $n$ and hence
$r(\K)=\lim_{n\to\infty}\|\K^n\|^{1/n}\ge a>0$.

{\it Claim (\ref{claim 4}).} Since $\sigma(\K)$ is nonempty and compact,
there is a $z\in\sigma(\K)$ with $|z|=r(\K)$.
By part (\ref{claim 1}), there is an eigenfunction $g\in E_z(\K)$
with $\|g\|=1$. We abbreviate $f:=|g|$. In particular, $\|f\|=\|g\|=1$.
From positivity of $\K$, it follows
that $\K f=\K|g|\ge|\K g|=|z||g|=\lambda f\ge 0$
and then $\K^{n+1}f\ge \lambda \K^nf\ge \lambda^{n+1}f\ge 0$ 
for all $n\in\N$ by iteration.
Take a sequence $(\lambda_m)_{m\in\N}$ of positive numbers $\lambda_m>\lambda$
with $\lambda_m\downarrow\lambda$ as $m\to\infty$ and
consider for the moment a fixed $m\in\N$. 
Since $\|\K^nf\|\le\|\K^n\|=(\lambda+o(1))^n$ as $n\to\infty$,
the series
\begin{align}
h_m:=(\id-\lambda_m^{-1}\K)^{-1}f
=\sum_{n=0}^\infty \lambda_m^{-n}\K^nf
\end{align}
converges. 
The facts $f\neq 0$ and $\K^n f\ge 0$ for all $n$ imply $h_m\neq 0$
and $h_m\ge 0$; hence $v_m:=h_m/\|h_m\|\ge 0$ is well-defined.
Furthermore,
\[
h_m\ge \sum_{n=0}^\infty \lambda_m^{-n}\lambda^n f
=(1-\lambda/\lambda_m)^{-1}f\ge 0
\]
implies $\|h_m\|\ge (1-\lambda/\lambda_m)^{-1}\|f\|=(1-\lambda/\lambda_m)^{-1}
\stackrel{m\to\infty}{\longrightarrow} \infty.$
From
$\K h_m-\lambda_mh_m=-\lambda_m f$
we conclude
$\|\K v_m-\lambda_mv_m\|=\lambda_m \|f\|/\|h_m\|\to 0$ as $m\to\infty$.
Because of $\|\lambda_mv_m-\lambda v_m\|=\lambda_m-\lambda\to 0$
it follows also $\|\K v_m-\lambda v_m\|\to 0$ as $m\to\infty$.
In particular, $\lim_{m\to\infty}\|\K v_m\|=\lim_{m\to\infty}\|\lambda v_m\|
=\lambda$.
By compactness of the operator $\K$ and $\|v_m\|=1$,
some subsequence $(\K v_{m_l})_{l\in\N}$ converges to
some $\Phi_\rechts \in \mathcal{H}$. 
By the positivity  of $\K$ and $v_m\ge 0$ we know $\K v_m\ge 0$ for all $m$,
hence $\Phi_\rechts \ge 0$ ($\mu$-a.e.). 
Furthermore, $\|\Phi_\rechts \|=\lim_{l\to\infty}\|\K v_{m_l}\|=\lambda$
and
\[
\|\K \Phi_\rechts -\lambda \Phi_\rechts \|=\lim_{l\to\infty}\|\K^2v_{m_l}-\lambda \K v_{m_l}\|
\le\lim_{l\to\infty}\|\K\|\|\K v_{m_l}-\lambda v_{m_l}\|=0.
\]
This means that $0\neq \Phi_\rechts \in E_\lambda(\K)$, hence
$\lambda\in \sigma(\K)$.
By hypothesis (d), for some $N\in\N$
the integral kernel $k_N$ of $\K^N$ takes only positive values, then
the facts $\Phi_\rechts \ge 0$, 
$\Phi_\rechts \neq 0$ and $\Phi_\rechts =\lambda^{-N}\K^N \Phi_\rechts $ imply 
$\Phi_\rechts >0$ (modulo changes on null sets).
The same arguments, applied to $\K^*$ instead of $\K$,
show  that $E_\lambda(\K^*)$ contains a
positive function $\Phi_\links$ (modulo changes on null sets). 

{\it Claim (\ref{claim 5}).} 
Take $N\in\N$, $f,g\in \mathcal{H}$ with $g\ge 0$ $\mu$-a.e.\
and $\K^N f=\lambda^N f+g$.
Using the eigenfunction $\Phi_\links>0$ of $\K^*$ from Claim (\ref{claim 4}),
we obtain
\[
\lambda^N \langle \Phi_\links,f\rangle
=\langle (\K^*)^N\Phi_\links,f\rangle= \langle \Phi_\links,\K^N f\rangle
=
\lambda^N \langle \Phi_\links,f\rangle+\langle \Phi_\links,g\rangle
\]
and therefore $\langle \Phi_\links,g\rangle=0$. Since $\Phi_\links>0$ and $g\ge 0$,
this implies $g=0$ $\mu$-a.e.

{\it Claim (\ref{claim 6}).} Let $u\in E_\lambda(\K)$:
our goal is to show that $u$ is a multiple of $\Phi_\rechts $ $\mu$-a.e.
Take again $N\in\N$ as in hypothesis (d); 
then $u\in E_{\lambda^N}(\K^N)$ also holds.
Since the integral kernel $k_N$ of $\K^N$ is real-valued,
it follows $\re u\in E_{\lambda^N}(\K^N)$ and  $\im u\in E_{\lambda^N}(\K^N)$;
thus it suffices to show that every real-valued $u\in E_{\lambda^N}(\K^N)$
is a multiple of $\Phi_\rechts$ 
$\mu$-a.e. Assume by contradiction that this was
false. Then there is $a\in\R$ such that neither $au+\Phi_\rechts \ge 0$ nor
$au+\Phi_\rechts \le 0$ holds $\mu$-a.e. 
Setting $h=au+\Phi_\rechts \in E_{\lambda^N}(\K^N)$, $f=|h|$ and 
$g=\K^N|h|-|\K^N h|$,
it follows $g\ge 0$, and
$\K^N f=|\K^N h|+g=|\lambda^N h|+g=\lambda^N f+g$. 
Since the integral kernel $k^N$ of $\K^N$ takes only positive values,
$g$ cannot be $0$ $\mu$-a.e.
This contradicts Claim  (\ref{claim 5}).

{\it Claim (\ref{claim 7}).} 
Assume that there was a strict inclusion 
$E_\lambda(\K)\subsetneqq S_\lambda(\K)$. Then there is 
$f\in \ker (\K-\lambda\id)^2$ with $\K f=\lambda f+\Phi_\rechts $,
since $\ker (\K-\lambda\id)=E_\lambda(\K)$ is spanned by $\Phi_\rechts $
by claims (\ref{claim 4}) and (\ref{claim 6}).
This contradicts Claim (\ref{claim 5}) since $\Phi_\rechts >0$.
The same argument applied to $\K^*$ instead of $\K$ shows 
$E_\lambda(\K^*)=S_\lambda(\K^*)$.

{\it Claim (\ref{claim 8}).} Let $z\in\sigma(\K)$ with 
$|z|=\lambda$.
Then, by (\ref{claim 1}), there is an eigenfunction
$f\in E_z(\K)$, $f\neq 0$. 
Taking $N\in\N$ from hypothesis (d),
we get 
$\lambda^N |f| =|z^Nf|=|\K^N f|\le \K^N|f|$, by  (\ref{claim 5})
$\lambda^N |f| = |\K^N f|= \K^N|f|$ $\mu$-a.e.. This means
for $\mu$-a.e.\ $x\in\Omega$
\[
|\K^N f(x)|=\left|\int_\Omega k_N(x,y) f(y)\,\mu(dy)\right|
=\int_\Omega k_N(x,y) |f(y)|\,\mu(dy)=\K^N|f|(x),
\]
therefore, again for $\mu$-a.e.\ $x\in\Omega$, 
there is a constant $c_x\in\C$ with $|c_x|=1$ such that
$k_N(x,y)f(y)=c_x k_N(x,y)|f(y)|$ holds for $\mu$-a.e.\ 
$y\in\Omega$. Using that $k_N$ is strictly positive, 
this implies $f=c|f|$ ($\mu$-a.e.)
for some $c\in \C$
with $|c|=1$,
therefore $\K^N f=c\K^N|f|=c\lambda^N |f| =\lambda^N f$.
We obtain $f\in S_\lambda(\K)$ and hence $f\in E_\lambda(\K)$
from Claim (\ref{claim 7}).
This implies $z=\lambda$ since $0\neq f\in E_z(\K)\cap E_\lambda(\K)$.
This shows that for all $z\in\sigma(\K)$ we have $z=\lambda$ or $|z|<\lambda$.
Since the spectral values of $\K$ can only accumulate at $0$,
this implies $\sup\{|z|:\;z\in\sigma(\K)\setminus\{\lambda\}\}<\lambda$.

{\it Claim (\ref{claim 9}).}
Let $\Sigma:=\K-\lambda P $.
$\Sigma$ is a compact operator, since $\K$ is compact and $ P $ has rank $1$.
We show now that $\lambda\notin \sigma(\Sigma)$.
Let $f\in E_\lambda(\Sigma)$; we need to show $f=0$.
From
\begin{align}
\K f-\lambda f&=\K f-\Sigma f=\lambda P  f= P  \Sigma f
=\sk{\Phi_\links,\K f}\Phi_\rechts -\lambda\sk{\Phi_\links, P  f}\Phi_\rechts 
\cr&
=
\sk{\K^*\Phi_\links,f}\Phi_\rechts -\lambda \sk{\Phi_\links,\Phi_\rechts }\sk{\Phi_\links,f}\Phi_\rechts 
\cr
&=\lambda\sk{\Phi_\links,f}\Phi_\rechts -\lambda\sk{\Phi_\links,f}\Phi_\rechts =0
\end{align}
we get that $f$ is a multiple $c\Phi_\rechts $ of $\Phi_\rechts $, with $c\in\C$,  since
$E_\lambda(\K)$ is spanned by $\Phi_\rechts $.
But $f\in E_\lambda(\Sigma)$ and
$\Sigma \Phi_\rechts =\K \Phi_\rechts -\lambda\sk{\Phi_\links,\Phi_\rechts }\Phi_\rechts =\lambda \Phi_\rechts -\lambda \Phi_\rechts =0$
then imply $\lambda f=\Sigma f=c\Sigma \Phi_\rechts =0$ and therefore $f=0$.

Next, let $z\in\C\setminus\{0,\lambda\}$.
We need to show that $z\in\sigma(\K)$ holds if and only if $z\in\sigma(\Sigma)$.
Now every non-zero spectral value of $\K$ (resp. $\Sigma$) is an eigenvalue of $\K$
 (resp. $\Sigma$).
Therefore, it suffices to show that $E_z(\K)=E_z(\Sigma)$.
To prove $E_z(\K)\subseteq E_z(\Sigma)$, let $f\in E_z(\K)$. 
Then $\lambda  P  f=\lambda \sk{\Phi_\links,f}\Phi_\rechts =\sk{\K^*\Phi_\links,f}\Phi_\rechts =\sk{\Phi_\links,\K f}\Phi_\rechts 
=z\sk{\Phi_\links,f}\Phi_\rechts =z P  f$, hence $ P  f=0$ since $z\neq \lambda$.
This shows $\Sigma f=\K f=zf$, i.e. $f\in E_z(\Sigma)$.

To prove $E_z(\K)\supseteq E_z(\Sigma)$, let $f\in E_z(\Sigma)$.
Note that $ P ^2= P $ since for  $g\in \mathcal{H}$ we have
 $ P ^2g=\sk{\Phi_\links,g}\sk{\Phi_\links,\Phi_\rechts }\Phi_\rechts 
=\sk{\Phi_\links,g}\Phi_\rechts = P  g$. 
We obtain 
\begin{align}
z  P  f&= P \Sigma f= P  \K f-\lambda P ^2f=
\sk{\Phi_\links,\K f}\Phi_\rechts -\lambda P  f
\cr
&=\sk{\K^*\Phi_\links,f}\Phi_\rechts -\lambda P  f
=
\lambda\sk{\Phi_\links,f}\Phi_\rechts -\lambda P  f=0
\end{align}
hence again $ P  f =0$. This shows $\K f=\Sigma f=zf$,
 i.e. $f\in E_z(\K)$.
Thus we have proven $\sigma(\Sigma)\setminus\{0\}=
\sigma(\K)\setminus\{\lambda,0\}$. The remaining claim $r(\Sigma)<\lambda$
follows now from assertion (\ref{claim 8}).

{\it Claim (\ref{claim 10}).}
Note that $ P \Sigma=0=\Sigma P $ hold, because for $f\in \mathcal{H}$, we have
\begin{align*}
& P \Sigma f=\sk{\Phi_\links,\K f}\Phi_\rechts -\lambda P ^2f=
\sk{\K^*\Phi_\links,f}\Phi_\rechts -\lambda P  f=\lambda\sk{\Phi_\links,f}\Phi_\rechts -\lambda P  f=0,
\\
&\Sigma P  f=\K P  f-\lambda  P ^2 f
= \sk{\Phi_\links,f}\K \Phi_\rechts -\lambda  P  f=\lambda\sk{\Phi_\links,f}\Phi_\rechts -\lambda  P  f=0.
\end{align*}
As a consequence, we obtain Claim (\ref{claim 10}) as follows:
\[
\lambda^{-n}\K^n=(\lambda^{-1}\Sigma+ P )^n=
\lambda^{-n}\Sigma^n+ P ^n=\lambda^{-n}\Sigma^n+ P 
=(\lambda^{-1}\K- P )^n+ P 
\]
and $\|\lambda^{-n}\K^n- P \|^{1/n}=\|\lambda^{-n}\Sigma^n\|^{1/n}
\stackrel{n\to\infty}{\longrightarrow} r(\lambda^{-1}\Sigma)
=\lambda^{-1}r (\Sigma)<1.$
\end{proof}
\end{appendix}

{\small
\bibliographystyle{alpha}

}

\end{document}